\documentclass[final, UKenglish, envcountsect, envcountsame, runningheads]{lmcs}
\pdfoutput=1

\usepackage{lastpage}
\lmcsdoi{21}{4}{4}
\lmcsheading{}{\pageref{LastPage}}{}{}%
{Nov.~01,~2024}{Oct.~08,~2025}{}

\usepackage[T1]{fontenc}
\usepackage[utf8]{inputenc}
\usepackage{graphicx}
\usepackage{hyperref}
\makeatletter
\hypersetup{
  pdfauthor={Fabian Lenke, Henning Urbat and Stefan Milius},
  pdftitle={Monoidal Extended Stone Duality},
  hidelinks,final,
}
  \renewcommand{\subsectionautorefname}{Section}%
\makeatother
\usepackage{color}

\newcommand{\iref}[2]{\autoref{#1}\ref{#1:#2}}
\usepackage{etex,etoolbox}


\usepackage{microtype}
\usepackage[notref,notcite]{showkeys}
\usepackage{amssymb}
\usepackage{amsmath}
\usepackage{stmaryrd}
\usepackage{mathtools}
\usepackage{xspace}
\usepackage{tikz-cd}
\tikzset{shiftarr/.style={
        rounded corners,%
        to path={--([#1]\tikztostart.center)
                     -- ([#1]\tikztotarget.center) \tikztonodes
                     -- (\tikztotarget)},
}}
\usepackage{bussproofs}
\EnableBpAbbreviations 

\usepackage[inline]{enumitem}
\setlist[enumerate,1]{label=(\arabic*),font=\normalfont,align=left,leftmargin=0pt,labelindent=0pt,listparindent=\parindent,labelwidth=0pt,itemindent=!,topsep=3pt,parsep=0pt,itemsep=3pt,start=1}
\setlist[enumerate,2]{label=(\alph*),font=\normalfont,labelindent=*,leftmargin=*,start=1}
\setlist[itemize]{labelindent=*,leftmargin=*,topsep=5pt,itemsep=3pt}
\setlist[description]{labelindent=*,leftmargin=*,itemindent=-1 em}

\usepackage{fixme}
\fxuselayouts{inline,nomargin}
\fxusetheme{color}
\FXRegisterAuthor{fb}{afb}{FB}
\FXRegisterAuthor{sm}{asm}{SM}
\FXRegisterAuthor{hu}{ahu}{HU}

\newcommand{\takeout}[1]{\empty}
\newcommand{\eps}{\varepsilon}
\newcommand{\xra}[1]{\xrightarrow{~#1~}}

\newcommand{\hatF}{\widehat{F}}
\newcommand{\hatU}{\widehat{U}}
\newcommand{\hatA}{\widehat{A}}
\newcommand{\hatX}{\widehat{X}}
\newcommand{\hatI}{\widehat{I}}
\newcommand{\hata}{\widehat{a}}
\newcommand{\hatB}{\widehat{B}}
\newcommand{\hatb}{\widehat{b}}
\newcommand{\hatff}{\widehat{f}}
\newcommand{\hateps}{\hat{\eps}}
\newcommand{\hateta}{\hat{\eta}}
\newcommand{\hatC}{\widehat{\cat C}}
\newcommand{\hatD}{\widehat{\cat D}}
\newcommand{\hatM}{\widehat{M}}
\newcommand{\hatminus}{\widehat{-}}
\newcommand{\hatotimes}{\mathbin{\widehat{\rule{0pt}{5pt}\smash{\otimes}}}}
\newcommand{\hatotimesp}{\widehat{\rule{0pt}{4pt}\smash{\otimes}}}

\numberwithin{equation}{section}

\theoremstyle{definition}
\newtheorem{constr}[thm]{Construction}
\newtheorem{myasm}[thm]{Assumption}

\newcommand*\coc{%
        \nobreak
        \mskip6mu plus1mu
        \mathpunct{}%
        \nonscript
        \mkern-\thinmuskip
        {:}%
        \mskip2mu
        \relax
}
\newcommand{\LR}{\,\Longleftrightarrow\,}

\newcommand{\q}[1]{\quad {#1} \quad}
\newcommand{\qq}[1]{\qquad {#1} \qquad}

\newcommand{\epi}{\twoheadrightarrow}

\newcommand{\incl}{\hookrightarrow}
\newcommand{\id}{\mathrm{id}}

\DeclareMathOperator{\mor}{Mor}

\newcommand{\pfm}[1][\Sigma]{\ensuremath{\overline{{#1}^*}}\xspace}

\newcommand{\simeqop}{\simeq^{\mathrm{op}}}
\newcommand{\cat}[1]{\ensuremath{\mathbf{#1}}\xspace}

\newcommand{\rel}{\cat{Rel}}
\newcommand{\ordrel}{\cat{OrdRel}}
\newcommand{\ccat}{\cat{Cat}}
\newcommand{\set}{\cat{Set}}

\newcommand{\pos}{\cat{Pos}}
\newcommand{\ba}{\cat{BA}}
\newcommand{\res}{\cat{Res}}

\newcommand{\relres}{\cat{RelRes}}
\newcommand{\der}{\cat{Der}}
\newcommand{\deracdl}{\cat{DerACDL}}

\newcommand{\relder}{\cat{RelDer}}
\newcommand{\relderacdl}{\cat{RelDerACDL}}

\newcommand{\dl}{\cat{DL}}

\newcommand{\relmon}{\cat{RelMon}}
\newcommand{\relordmon}{\cat{RelOrdMon}}

\newcommand{\relcomon}{\cat{RelComon}}

\newcommand{\comon}{\cat{Comon}}

\newcommand{\ordmon}{\cat{OrdMon}}

\newcommand{\stone}{\cat{Stone}}
\newcommand{\priest}{\cat{Priest}}
\newcommand{\stonemon}{\cat{StoneMon}}

\newcommand{\cslj}{\ensuremath{\cat{CSL}_{\bigvee}}\xspace}
\newcommand{\cslm}{\ensuremath{\cat{CSL}_{\bigwedge}}\xspace}
\newcommand{\caba}{\ensuremath{\cat{CABA}}\xspace}
\newcommand{\rescaba}{\ensuremath{\cat{ResCABA}}\xspace}

\newcommand{\resacdl}{\ensuremath{\cat{ResACDL}}\xspace}
\newcommand{\relresacdl}{\ensuremath{\cat{RelResACDL}}\xspace}
\newcommand{\acdl}{\ensuremath{\cat{ACDL}}\xspace}
\newcommand{\jsorc}{\jsor_{\mathrm{C}}}
\newcommand{\msorc}{\msor_{\mathrm{C}}}

\newcommand{\msl}{\cat{MSL}}
\newcommand{\jsl}{\cat{JSL}}
\newcommand{\stonejsl}{\cat{StoneJSL}}

\newcommand{\viet}{\mathbb{V}}
\newcommand{\dviet}{\mathbb{V}_{\!\downarrow}}
\newcommand{\dvietplus}{\mathbb{V}_{\!\downarrow}^+}
\newcommand{\uviet}{\mathbb{V}_{\uparrow}}
\DeclareMathOperator{\at}{\mathcal{A}}

\newcommand{\pow}{\mathcal{P}} 
\newcommand{\Vj}{V_{\lor}}
\newcommand{\Vm}{V_{\land}}
\newcommand{\Uj}{U_{\lor}}
\newcommand{\Um}{U_{\land}}

\DeclareMathOperator{\cl}{\mathrm{Cl}}
\DeclareMathOperator{\reg}{\mathrm{Reg}}

\DeclareMathOperator{\pro}{Pro}
\DeclareMathOperator{\ind}{Ind}
\DeclareMathOperator{\coalg}{\mathbf{Coalg}}

\DeclareMathOperator{\Id}{\mathrm{Id}}

\newcommand{\N}{\mathbb{N}}

\DeclareMathOperator{\V}{\bigvee} 
\DeclareMathOperator{\A}{\bigwedge}
\newcommand{\lres}{\setminus}
\newcommand{\rres}{\mathrel /}
\newcommand{\genres}[1]{\ensuremath{\lres {#1}/}}

\newcommand{\moni}{\multimap}
\DeclareRobustCommand{\imon}{{\,{\reflectbox{$\moni$}}\,}}

\DeclareMathOperator{\idl}{Idl}

\newcommand{\jsor}{\mathbin{{\otimes}}}
\newcommand{\msor}{\mathbin{{\boxtimes}}}
\newcommand{\mtoj}{\ensuremath{\jtom^{-1}}}
\newcommand{\jtom}{\ensuremath{\omega}}

\DeclareMathOperator{\op}{\mathrm{Op}}

\begin{document}
%
%
\title{ Extended Stone Duality via Monoidal Adjunctions}
%
%
\author[F.~Lenke]{Fabian Lenke\lmcsorcid{0000-0001-5890-9485}}
\author[H.~Urbat]{Henning Urbat\lmcsorcid{0000-0002-3265-7168}}
\author[S.~Milius]{Stefan Milius\lmcsorcid{0000-0002-2021-1644}}
\address{Friedrich-Alexander-Universität Erlangen-Nürnberg, Germany}
\email{\{fabian.birkmann,stefan.milius,henning.urbat\}@fau.de}%

\begin{abstract}
  Extensions of Stone-type dualities have a long history in algebraic logic and have also been instrumental in proving results in algebraic language theory.
  We show how to extend abstract categorical dualities via monoidal adjunctions, subsuming various incarnations of classical extended Stone and Priestley duality as special cases, and providing the foundation for two new concrete dualities:
  First, we investigate residuation algebras, which are lattices with additional residual operators modeling language derivatives algebraically. We show that the subcategory of derivation algebras is dually equivalent to the category of profinite ordered monoids, restricting to a duality between Boolean residuation algebras and profinite monoids.
  We further refine this duality to capture relational morphisms of profinite ordered monoids, which dualize to natural morphisms of residuation algebras.
  Second, we apply the categorical extended duality to the discrete setting of sets and complete atomic Boolean algebras to obtain a concrete description for the dual of the category of all small categories.
\end{abstract}

\maketitle              

\section{Introduction}

Marshall H.~Stone's representation theorem for Boolean algebras~\cite{stone-36}, the
foundation for the so called \emph{Stone duality} between Boolean
algebras and Stone spaces, manifests a tight connection between logic
and topology.  It has thus become an ubiquitous tool in various areas
of theoretical computer science, not only in logic, but also, for
example, in domain theory~\cite{abr91} and automata theory~\cite{pippenger-97,gehrke-gregorieff-pin-08}.

From algebraic logic arose the need for extending Stone duality to
capture Boolean algebras equipped with additional operators (modelling
quantifiers or modalities). Originating from Jónsson and Tarski's
representation theorem for Boolean algebras with
operators~\cite{jonsson-tarski-51,jonsson-tarski-52}, a representation
in the spirit of Stone was proven by Halmos~\cite{halmos-58}; the
general categorical picture of the duality of Kripke frames and modal
algebras is based on an adjunction between operators and continuous
relations developed by Sambin and Vaccaro~\cite{sambin-vaccaro}.

In the context of automata theory, the need for extensions of Stone
duality was only unveiled in this millennium:
Using ordinary Stone duality, Pippenger~\cite{pippenger-97} (see also Almeida~\cite{almeida_94}) has already shown that the Boolean
algebra of regular languages on an alphabet \(\Sigma\) corresponds
to the Stone space \(\widehat{\Sigma^{*}}\) of
profinite words. This result and the theory surrounding it was placed in the bigger picture by categorical frameworks that have identified Stone-type dualities to be one of the cornerstones of
algebraic language theory~\cite{urb-ada-che-mil-17-proc,sal-17,blu21}.
On the other hand, Gehrke~et~al.~\cite{gehrke-gregorieff-pin-08}
discovered that, under Goldblatt's form of extended Priestley
duality~\cite{goldblatt-89}, the \emph{residuals} of language concatenation dualize to \emph{multiplication} of profinite words, but so far this result could not yet be placed in the categorical big picture. One
reason might be that, despite some progress in recent
years~\cite{bosangue-kurz-rewitzky-07,hofmann-15}, extended (Stone)
dualities for (co-)algebras are themselves not fully understood as
instances of a crisp categorical idea.

\paragraph{Contributions} In the present paper, we introduce in \autoref{sec:extending-dualities} a simple categorical framework for extending any categorical duality
\(\cat C \simeqop \hatC\) via \emph{monoidal adjunctions}: for
a given adjunction on~\(\cat C\) with a strong monoidal right adjoint
\(U\), we establish a dual equivalence between
\(U\)-operators on \(\cat C\) and operators in a Kleisli
category on \(\hatC\).
This framework is compositional both in its parameter -- the strong monoidal right adjoint -- and the object level of \(U\)-operators, which yields a simple categorical version of correspondence theory.

We demonstrate the power of this framework by working out several examples.
On one hand, we show how to recover existing dualities and applications thereof, but we also indicate how to extend existing dualities including completely new instances of dualities.

To this end, we show in \autoref{sec:exampl-extend-priest} how to use our framework to recover,
and mildly extend, extended Priestley duality for distributive lattices with
operators~\cite{goldblatt-89} and relational morphisms. We also show how to use the
compositionality of the abstract extended duality to recover results from modal correspondence
theory for free.

Guided by our
categorical foundations for extended Stone duality, we subsequently investigate in \autoref{sec:residuation-algebras} the
correspondence between language derivatives and multiplication of profinite words
in the setting of \emph{residuation algebras}.
The key observation is that, on
complete and completely distributive lattices, the residuals are equivalent to a
\emph{coalgebraic} operator on the lattice. This equivalence can then be composed with  an extended duality based on the discrete duality between complete atomic Boolean algebras (CABAs) and sets to obtain a duality between certain complete residuation algebras and ordered monoids.
We then proceed to lift
this correspondence to locally finite structures, i.e.~structures
built up from finite substructures.  By identifying suitable non-full
subcategories -- derivation algebras and (lattice) comonoids,
respectively -- and an appropriate definition of morphism for
residuation algebras, we augment Gehrke's characterization of
Stone-topological algebras in terms of residuation algebras to a
duality between the categories of derivation algebras and that of
profinite ordered monoids:
\begin{equation}\label{eq:duality}
  \der \;\cong\; \comon\;\simeqop\;
  \cat{ProfOrdMon}.
\end{equation}
The abstract theory of extended duality now suggests that the
dual equivalence between profinite ordered monoids on the one side and comonoids
as well as derivation algebras on the other side extends to a more general
duality capturing morphisms of \emph{relational} type of profinite ordered monoids.
To this end, we identify a natural notion of relational
morphism for residuation algebras and comonoids, and we use our abstract extended duality theorem to obtain the dual equivalence
\[\relder \;\cong\; \relcomon\;\simeqop\; \cat{RelProfOrdMon} \]
which extends \eqref{eq:duality} to relational morphisms.

Finally, we combine the ideas underlying these results
to derive a novel duality between the
category of all small categories and the category of \emph{categorical residuation algebras}:
\begin{equation}
  \label{eq:duality-cat}
  \ccat \;\simeqop\; \cat{CatResCABA}.
\end{equation}
To achieve this, we use the duality between monoids and residuation CABAs and combine it with two observations:
first, small categories  are equivalent to certain \emph{relational monoids};
second, the discrete version of the duality~\eqref{eq:duality} admits an extension to a duality between relational
monoids and \emph{residuation CABAs}.
We prove that the composite of these equivalences restricts to \eqref{eq:duality-cat}.

\section*{Related Work.}
This paper is a completely revised and extended version of our
conference paper~\cite{BirkmannEA24} presented at FoSSaCS 2024. Besides providing detailed proofs
of all results, we have included additional material:
\autoref{prop:composite-duality} simplifies dualization of composite operators,
and it is used in  \autoref{sec:towards-point-free} to show how to derive results from modal correspondence theory
in a categorical way, by encoding modal formulas as morphisms.
\autoref{sec:finite-resid-algebr} has been extended to complete lattices, instead of just finite ones, to set the base duality for the material in \autoref{sec:duality-cat}.
We have added \autoref{prop:every-priest-mon-prof} to complete the picture
relating profinite ordered monoids and Priestley monoids similarly to the unordered case. Finally, the dual characterization of the category \ccat of small categories (\autoref{sec:duality-cat}) is a new application of our abstract methods.
\paragraph{Extended Stone Duality.} Duality for (complete) Boolean algebras with operators goes back to
Jónsson and Tarski~\cite{jonsson-tarski-51,jonsson-tarski-52}.  This
duality was refined by the topological approach via Stone spaces taken
by Halmos~\cite{halmos-58}, which allowed to characterize the
relations arising as the duals of operators, namely \emph{Boolean
  relations}.  Halmos' duality was extended to distributive lattices
with (\(n\)-ary) operators by Goldblatt~\cite{goldblatt-89} and
Cignoli~\cite{cignoli-91}.  Kupke et al.~\cite{kupke-kurz-venema-04}
recognized that Boolean relations elegantly describe descriptive
frames as coalgebras for the (underlying functor of) the Vietoris monad on
the category of Stone spaces; notions of bisimulation for these coalgebras were
investigated by Bezhanishvili
et~al.~\cite{bezhanishvili-fontaine-venema-10}. Bonsangue
et~al.~\cite{bosangue-kurz-rewitzky-07} introduced a framework for
dualities over distributive lattices equipped with a theory of
operators for a signature, which are dual to certain coalgebras.
Hofmann and Nora~\cite{hofmann-15} have taken a categorical approach
to extend natural dualities to algebras for a signature equipped with
{unary} operators preserving only some of the operations
prescribed by the signature; they relate these to coalgebras for (the
underlying functor of) a suitable monad \(T\).    Recent work by Bezhanishvili
et~al.~\cite{bezhanishvili-harding-morandi-23} clarifies the relation
between free constructions on distributive lattices and different
versions of the Vietoris monad to derive several dualities for
distributive lattices with different types of operators and their
corresponding Priestley relations.

Similarly to Hofmann and Nora's work~\cite{hofmann-15} we have tried to fit most of these
developments in a single, categorical framework. Our approach differs in two ways: first,
incorporating monoidal structures into our setting immediately allows us to dualize operators
with multiple in- and outputs. Second, besides a ``nice'' base duality, our framework depends
only on some monoidal right adjoint. In contrast, op.~cit.~requires,
in addition to the right adjoint (represented by a subsignature), a candidate monad \(T\)
and gives a criterion whether an extended duality between operators for the subsignature and
\(T\)-coalgebras exists.

\paragraph{Residuation Algebras.} The original work by Jónsson and Tarski~\cite{jonsson-tarski-51} already used duality for residual operators (also called \emph{conjugates}) on Boolean algebras. Residuated Boolean algebras, i.e.~Boolean algebras with a residuated
binary operator, were explicitly studied by Jónsson and
Tsinakis~\cite{jonsson-93} to highlight the role of the residuals in
relation algebra.  Gehrke et~al.~\cite{gehrke-gregorieff-pin-08}
exposed the connection between the residuals of the concatenation
of regular languages and the multiplication on profinite words and
investigated applications to automata theory, most notably a
duality-theoretic proof of Eilenberg's variety
theorem~\cite{Eilenberg1976}.  The duality theory behind the
correspondence of general residuation algebras and
Priestley-topological algebras developed by Gehrke~\cite{gehrke-16}
is based on Goldblatt's
extension of Stone duality~\cite{goldblatt-89} for distributive lattices.
This duality was further investigated via the theory of canonical
extensions~\cite{gj94,gehrke-09,gehrke-priestley-07} to show that certain crucial parts are not entirely ``algebraic'':
While Gehrke~\cite{gehrke-16} provides a condition under which the dual relation of the
residuals is functional,
Fussner and Palmigiano~\cite{fussner-palmigiano-19} showed that it cannot be equationally defined in the language of
residuation algebras.

Our duality for Priestley monoids is a non-trivial restriction of Gehrke's
duality~\cite{gehrke-16}, and, to our knowledge, the first duality result for relational
morphisms of profinite monoids, which are a ubiquitous tool in algebraic language
theory~\cite{pin-88} and semigroup theory~\cite{rhodes-09}.  We also note that, while our
results are closely related to Gehrke's~\cite{gehrke-16}, our methods are fundamentally
different: while most of the proofs in op.~cit.~are of topological nature, we only work
on the side of ordered structures, and outsource any topology to already existing dualities. In
our opinion, this not only simplifies the theory, but it also clarifies the relation between
Gehrke's results and the duality by Rhodes and Steinberg~\cite{rhodes-09} between profinite
monoids and counital Boolean bialgebras: on the algebraic side, derivation algebras correspond
to Boolean comonoids, which are precisely counital Boolean bialgebras.
%

\section{Preliminaries}\label{sec:prelim}
Readers are assumed to be familiar with basic category theory, such as
functors, natural transformations, adjunctions and monoidal
categories, see Mac Lane~\cite{maclane-98} for an introduction.

We
briefly recall the foundations of Stone duality~\cite{stone-36} and
Priestley duality~\cite{priestley-70}.  By the latter we mean the dual
equivalence $\dl \simeqop \priest$ between the category \(\dl\) of
bounded distributive lattices and lattice homomorphisms, and the
category \(\priest\) of Priestley spaces (ordered compact topological
spaces in which for every $x\not\leq y$ there exists a clopen upset
containing $x$ but not $y$) and continuous monotone maps.  The duality
maps a distributive lattice $D$ to the pointwise-ordered space
\(\dl(D, 2)\) of homomorphisms into the two-element lattice
(equivalently, prime filters, ordered by inclusion), and topologized
via pointwise convergence. In the reverse direction, it maps a
Priestley space $X$ to the distributive lattice \(\priest(X, 2)\) of
continuous maps into the two-element poset $2=\{0\leq 1\}$ with
discrete topology (equivalently, clopen upsets), with the pointwise
lattice structure. Priestley duality restricts to Stone duality
\(\ba \simeqop \stone\) between the full subcategories~$\ba$ of
Boolean algebras and $\stone$ of Stone spaces (discretely ordered
Priestley spaces).

This duality also has a ``topologically discrete'' version:
First recall that an element \(k\) in a lattice \(D\) is \emph{compact} if for every subset \(A \subseteq D\) with \(k \le \bigvee A\) there exists a finite subset \(F \subseteq A\) with \(k \le F\). A lattice is \emph{algebraic} if every element is the join of the compact elements below it.
There exists a duality \(\pos \simeqop \acdl\)
between the category \pos of posets with order-preserving maps and the category \acdl of algebraic completely distributive lattices (ACDLs) with maps preserving all joins and meets. Under this duality a poset \(P\) is mapped to the set \(\pos(P, 2) \cong \mathcal{D} P\), which corresponds to the set of downsets of \(P\). In the reverse direction, an ACDL \(D\) is mapped to the pointwise-ordered poset \(\acdl(D, 2) \cong \mathcal{J} P\), which corresponds to the poset of completely join-prime elements of \(D\), i.e.\ those elements of \(D\) satisfying \(p \le \bigvee A \Rightarrow p \le a\) for some \(a \in A\).
This duality restricts to the well-known duality \(\set \simeqop \caba\) between the category \set and the category \caba of complete atomic boolean algebras.

Both the topological and the discrete duality restrict to Birkhoff
duality~\cite{birkhoff-37}
$\dl_{\mathrm{f}} \simeqop \pos_{\mathrm{f}}$ between finite
distributive lattices and finite posets. For a comprehensive
introduction to ordered structures and their dualities, see the first
two chapters of the classic textbook by Johnstone~\cite{johnstone-82}.


\section{Extending Dualities}
\label{sec:extending-dualities}

Our first contribution is a general categorical framework for extending Stone-type dualities via monoidal adjunctions. It is motivated by the extension of Priestley duality to operators due to Goldblatt~\cite{goldblatt-89} (which is recovered in \autoref{sec:exampl-extend-priest}) and serves as the basis for several concrete duality results derived subsequently.

We start this chapter by setting up some notation for the two ingredients involved: adjunctions and monoidal categories.

\begin{nota}\hfill    
  \begin{enumerate}
  \item Given functors \(U \colon \cat C \rightarrow \cat D\) and
    \(F \colon \cat D \rightarrow \cat C\), we write
    \[ F \colon \cat D \dashv \cat C \coc U \]
    (or simply
    $F \dashv U$) if $F$ is left adjoint to $U$. We denote the unit
    and counit by \begin{equation}\label{eq:unit-counit}\eta\colon \Id \to UF\qquad\text{and}\qquad \varepsilon\colon FU \to \Id;\end{equation} the transposing isomorphisms
    are denoted by
    \[(-)^{+}\colon \cat D(C, UD) \rightleftarrows \cat C(FC, D) \coc (-)^{-}\]
    Hence, for every $f\colon C \to UD$ and $g\colon FC \to D$ we have
    \[
      f^{+} = \varepsilon_D \cdot Ff,
      \qquad\text{and}\qquad
      g^{-} = Ug \cdot \eta_C.
    \]

  \item For a category $\cat C$ with dual $\hatC$ we
    denote both contravariant functors witnessing the dual equivalence by
    \[(\hatminus)\colon \cat C\xra{\simeq} \hatC \qquad\text{and}\qquad
    (\hatminus)\colon \hatC\xra{\simeq} \cat{C}.\] Moreover,
    if \(F \colon \cat C \rightarrow \cat D\) is a functor and~$\hatD$
    is dual to $\cat D$, then we denote the dual of $F$ by
    \[\hatF = (\hatminus) \cdot F \cdot (\hatminus) \colon \hatC
    \rightarrow \hatD.\]
  \item The Kleisli category of a monad \((T,\eta,\mu)\) on the
    category $\cat C$
    is denoted by \(\cat C_{T}\).  It has the same objects as
    \(\cat C\) and \(\cat C_{T}(X, Y) = \cat C(X, TY)\); the
    composition of \emph{Kleisli morphisms} $f\colon C \to TD$ and
    $g\colon D \to TE$ is defined by%

    \[
      g \cdot f = \big(
      C \xra{f} TD \xra{Tg} TTE \xra{\mu_E} TE
      \big).
    \]
    A Kleisli morphism
    \(f \colon C \rightarrow T D\) is
    \emph{pure} if \(f=\eta_D \cdot f'\) for some \(f'\colon C \to D\)
    in~$\cat C$. 
    We write $J_T\colon \cat C\to \cat C_T$ for the usual identity-on-objects functor
    mapping $f\colon A \to B$ to $\eta_B \cdot f\colon A \to TB$.
  \end{enumerate}
\end{nota}
\begin{conv}
  To lighten notation, we omit subscripts indicating components of
  natural transformations when they are clear from the context,
  e.g.~we write $\eta\colon A \to TA$ for \(\eta_{A}\).
\end{conv}

\begin{defi}\hfill
  \begin{enumerate}
          \item A \emph{monoidal category} is a category \cat{C} with a bifunctor \(\otimes \colon \cat{C} \times \cat{C} \rightarrow \cat{C} \) called \emph{tensor}, an object \(I \in \cat{C}\) called \emph{unit} together with natural isomorphisms
          \[\upsilon \colon I \otimes \Id_{\cat{C}} \cong \Id_{\cat{C}} \quad \text{ and } \quad \alpha \colon (\Id_{\cat{C}} \otimes \Id_{\cat{C}}) \otimes \Id_{\cat{C}} \cong \Id_{\cat{C}} \otimes (\Id_{\cat{C}} \otimes \Id_{\cat{C}}), \]
          which are subject to natural \emph{coherence axioms}~(see e.g.~\cite[Section VII]{maclane-98}).
          A monoidal category is \emph{strict} if \(\upsilon\) and \(\alpha\) are identities.
          \item A functor \(U \colon \cat{C} \rightarrow \cat{D}\) between monoidal categories that is equipped with a morphism \(\epsilon \colon I_{\cat D} \rightarrow UI_{\cat C}\) and a transformation
          \[\lambda \colon U(-) \otimes U(-) \rightarrow U(- \otimes -) \]
          which is  natural in both components is \emph{lax monoidal} if it satisfies appropriate {associativity} and {unitarity} axioms with respect to \(\upsilon\) and \(\alpha\).
          If both \(\epsilon\) and \(\lambda\) are isomorphisms (identities) then we say \(U\) is \emph{strong} (\emph{strict}) monoidal.
  \end{enumerate}
\end{defi}

\begin{nota}
  Given an object $X$ of a monoidal category, the $n$th tensor power is denoted
\[ X^{\otimes n}=\bigotimes_{i=1}^n X,\]
and for a functor $G$ the expression $GX^{\otimes n}$ is parsed as $G(X^{\otimes n})$, as usual.
\end{nota}

We now introduce the setting in which our framework for extending dualities applies.
In the simplest sense, the only ingredient is a strong monoidal functor:

\begin{myasm}\label{asm}
  We fix strict monoidal categories \(\cat C, \cat D\) with dually equivalent
  categories $\hatC, \hatD$; we regard $\hatC, \hatD$ as
  monoidal categories with tensor products $\hatotimes$ dual to the tensor products~$\otimes$ of $\cat C, \cat D$.
Moreover, we fix an adjunction \(F \colon \cat D \dashv \cat C \coc U\)
  with unit $\eta$ and counit $\varepsilon$ \eqref{eq:unit-counit},
  and we assume that~\(U\) is a {strong monoidal functor} with
  associated natural isomorphisms \[\lambda \colon UX \otimes UY \cong
  U(X \otimes Y) \qquad\text{and}\qquad \epsilon \colon I_{\cat D} \cong UI_{\cat
    C}.\]
  We denote the monad dual to the comonad \(FU\) by \(T =
  \hatF\hatU\), with unit and multiplication
\[ e=\hateps \colon \Id
  \rightarrow T\qquad\text{and}\qquad m = \hatF\hateta\hatU
  \colon TT \rightarrow T.\]
  Overall, we have the following situation:
  \begin{equation}\label{eq:adjoints}
    \begin{tikzcd}[sep = large]
      \cat D
      \rar[phantom]{\simeqop}
      \dar[bend right=30pt, swap,pos=.52]{F}
      &
      \hatD
      \ar[bend right=30pt]{d}[swap,pos=.55]{\hatF}
      \\
      \cat C
      \rar[phantom]{\simeqop}
      \uar[bend right=30pt, swap, pos=.48]{U}
      \uar[phantom]{\dashv}
      &
      \hatC
      \rar[loop right]{T}
      \uar[phantom]{\vdash}
      \ar[bend right=30pt]{u}[swap,pos=.45]{\hatU}
    \end{tikzcd}
  \end{equation}
\end{myasm}

\begin{rem}\label{rem:mon-adj-kleisli}\hfill
  \begin{enumerate}
  \item By Mac Lane's \emph{Coherence Theorem}~\cite[Section VII.2]{maclane-98}, every monoidal
    category is equivalent to a strict monoidal category, hence the strictness requirement on $\cat C$ in \autoref{asm} is without loss of generality.
  \item  The isomorphism \(\lambda\) can be extended to an isomorphism
    \[
      \lambda \colon UX_{1} \otimes \cdots \otimes UX_{n} \cong
      U(X_{1} \otimes \cdots \otimes X_{n})
    \]
for all finite families of objects $X_1\ldots, X_n$ in $\cat C$.
  \item  The functor \(\hatU \colon \hatC \rightarrow \hatD\) dual to $U$ is a
  strong monoidal \emph{left} adjoint of \(\hatF\), and the unit and
  counit of the adjunction $\hatU \dashv \hatF$ are given by \(\hateps\) and \(\hateta\), respectively. Since \(\hatU\) is strong monoidal with respect to the isomorphisms
    \[{\hat{\epsilon}} \colon \hatI_{\cat{D}} \cong \hatU
    \hatI_{\cat{C}} \qquad\text{and}\qquad
    \hat{\lambda} \colon \hatU X \hatotimes \hatU Y \cong \hatU(X
    \hatotimes Y) \coc \hat{\lambda}^{-1},\] its right adjoint \(\hatF\) is lax monoidal
    (see e.g.~\cite[p.~17]{schwede-shipley-02}) for the natural transformations
  \[ (\hat{\epsilon}^{-1})^{-} \colon \hatI_{\cat C} \rightarrow \hatF \hatI_{\cat D} \qquad\text{and}\qquad  ((\hateta \hatotimes \hateta) \cdot
    \hat{\lambda}^{-1})^{-} \colon \hatF X \hatotimes \hatF Y \rightarrow \hatF(X \hatotimes Y).  \]
          This makes $\hatU \dashv \hatF$ a monoidal adjunction and \(T = \hatF \hatU\)  a \emph{monoidal} or \emph{commutative} monad on $\hatC$ with monoidal structure \(\hat{\delta} \colon TX \hatotimes TY \rightarrow T(X \hatotimes Y)\) given as the appropriate composite of the monoidal structures of \(\hatF\) and \(\hatU\).
          Note that \(\hat \delta\) also
  extends to any arity, that is, for every $n$-tuple of objects $X_1,
  \ldots X_n$, we obtain a natural transformation
  \[
    \hat{\delta}\colon TX_1 \hatotimes \cdots \hatotimes TX_n \to T(X_1 \hatotimes \cdots
    \hatotimes X_n).
  \]
\item The tensor product ${\hatotimes}$ of $\hatC$ lifts to the
  Kleisli category \(\hatC_{T}\); the lifting maps a pair
  \((f\colon X \rightarrow TY, g\colon X'\to TY')\) of
  \(\hatC_{T}\)-morphisms to the \(\hatC_{T}\)-morphism
  \[\hat{\delta} \cdot (f{\hatotimes} g) \colon X {\hatotimes} X'
  \rightarrow TY {\hatotimes} TY' \to T(Y {\hatotimes} Y').\] This
  makes $\hatC_T$ itself a monoidal
  category (see e.g.~\cite[Prop.~1.2.2]{seal03}) with tensor \(\hatotimes\), and
  the canonical left adjoint $J_T\colon \hatC\to\hatC_T$ a strict
  monoidal functor.
\end{enumerate}
\end{rem}

We now show how to extend the duality \(\cat C \simeqop \hatC \) along the adjunction \(F \dashv U\) to yield a duality between operators.
\begin{defi}\label{def:operator}\hfill
  \begin{enumerate}
  \item Let $G\colon {\cat A}\to {\cat B}$ be a functor between
    monoidal categories, and let \(k, n \in \N\).  An
    \emph{\((k, n)\)-ary \( G \)-operator} consists of an object
    \(A \in \cat A\) and a morphism
    \(a \colon (GA)^{\otimes k} \rightarrow (GA)^{\otimes n}\) of
    $\cat B$.  A \emph{\((k, n)\)-ary \( G \)-operator morphism} from
    \((A, a) \) to \((B, b)\) is a morphism $h\colon GA\to GB$ of
    $\cat B$ such that the following square commutes:
    \[
      \begin{tikzcd}
        (GA)^{\otimes k}
        \rar{a}
        \dar[swap]{h^{\otimes k}}
        &
        (GA)^{\otimes n}
        \dar{ h^{\otimes n} }
        \\
        (GB)^{\otimes k}
        \rar{b}
        &
        (GB)^{\otimes n}
      \end{tikzcd}
    \]
    The category of \((k, n)\)-ary \( G \)-operators and their
    morphisms is denoted by \(\op_{G}^{k, n}( \cat A )\). 
\item A \emph{$G$-operator} is a $(k,n)$-ary $G$-operator for some $k$ and $n$.
  \item A \emph{\(G\)-algebra} is a  \((k, 1)\)-ary \(G\)-operator,
    and a \emph{\(G\)-coalgebra} is a \((1, n)\)-ary one.
  \item If $G$ is strong monoidal, then a $(k,n)$-ary $G$-operator $(A, a)$ is \emph{pure} 
    if there exists a morphism $a'\colon A^{\otimes k} \to A^{\otimes
      n}$ such that 
    \[
      a = \big(
      (GA)^{\otimes k}
      \xra{\lambda}
      G(A^{\otimes k})
      \xra{Ga'}
      G(A^{\otimes n})
      \xra{\lambda^{-1}}
      (GA)^{\otimes n} \big),
    \]
    where $\lambda$ is 
    given analogously to \autoref{asm}. An operator morphism $h\colon (A,a)
    \to (B,b)$ is \emph{pure} if there exists a morphism $h'\colon A
    \to B$ such that \( h = Gh'\colon GA \to GB\).
  \end{enumerate}
\end{defi}


Two cases will be important for intuition: (1) if \(G\) is a forgetful functor between varieties of algebraic structures, then a $G$-operator on an object \(A\) is a map preserving only part of the structure on \(A\); (2) if \(T\) is a powerset-like monad on a category of spaces, then an operator for the embedding \(J_{T} \colon \hatC \rightarrow \hatC_{T}\) on a space \(\hatA\) is a (continuous) \emph{relation} \(\hatA \rightarrow T\hatA\).

\begin{thm}[Abstract Extended Duality]\label{thm:extended-duality}
  The category of \((k, n)\)-ary \(U\)-operators is dually equivalent
  to the category of \((n, k)\)-ary \(J_T\)-operators:
  \[\op_{U}^{k, n}(\cat C) \;\simeqop\; \op_{J_T}^{n, k}(\hatC). \]
\end{thm}
\begin{proof}
 The desired equivalence of categories is given by the functor
  \begin{equation}\label{eq:Phi}\Phi\colon \op_{J_{T}}^{n, k}(\hatC) \rightarrow \op_{U}^{k, n}(\cat
  C)\end{equation} which maps a $(n,k)$-ary $J_T$-operator
  \[
    \hata \colon
    \hatA^{\hatotimesp n} \to
    T \hatA^{\hatotimesp k} = \hatF\hatU \hatA^{\hatotimesp k}
  \]
  to the $(k,n)$-ary $U$-operator
  \begin{equation}
    \label{eq:dual-op}
    (UA)^{\otimes k}
    \xra{\lambda}
    UA^{\otimes k}
    \xra{a^-}
    UA^{\otimes n}
    \xra{\lambda^{-1}}
    (UA)^{\otimes n},
  \end{equation}
  and an operator morphism
  \(\hatff \colon (\hatA, \hata) \rightarrow (\hatB, \hatb)\)
  to \(f^{-} \colon UB \rightarrow UA\).

  Let us first verify that the definition of $\Phi$ makes sense on morphisms, i.e.\ that \(f^{-}\) is  an operator
  morphism from \((B, \lambda^{-1} \cdot b^{-} \cdot \lambda)$ to $(A, \lambda^{-1} \cdot a^{-} \cdot \lambda)\).
  Unfolding the definitions of composition and monoidal structure in the Kleisli category, we see that \(\hatff\) is a morphism \(\hatA \rightarrow T \hatB = \hatF \hatU \hatB \) such that
  the following diagram in \(\hatC\) commutes:
  \begin{equation}
    \label{eq:in-Chat}
    \begin{tikzcd}
      \hatA^{ \hatotimesp n}
      \rar{\hata}
      \dar{\hatff^{\widehat{\otimes} n}}
      &
      T \hatA^{\hatotimes k}
      \rar{T \hatff^{\widehat{\otimes} k}}
      &
      T (T \hatB)^{\hatotimesp k}
      \rar{T \delta}
      &
      TT\hatB^{\hatotimesp k}
      \dar{m}
      \\
      (T\hatB)^{\hatotimesp n}
      \rar{\delta}
      &
      T\hatB^{\hatotimesp n}
      \rar{T \hatb}
      &
      T T \hatB^{\hatotimesp k}
      \rar{m}
      &
      T \hatB^{\hatotimesp k}
    \end{tikzcd}
  \end{equation}
  The multiplication \(m\) of the monad is the dual of \(F\eta U\), so
  the dual diagram of~\eqref{eq:in-Chat} is precisely the outside
  of the following diagram in $\cat C$, in which all parts except~\((\#)\) commute by naturality of \(\eta, \varepsilon\) and \(\lambda\), and the triangle equations. 
  \begin{tiny}
    \[
      \begin{tikzcd}
        FUFUB^{\otimes k} \dar{FUF \lambda^{-1}} \ar[shiftarr={xshift=-35pt},pos=.35 ]{dddd}{FU \delta}
        & FUB^{\otimes k} \lar[swap]{F \eta U} \dar{F \lambda^{-1}} \rar{F \eta U} \ar[bend right=10, swap]{rr}{F b^{-}}
        & FUFUB^{\otimes k} \rar{FU b}
        & FUB^{\otimes n} \dar{F \lambda^{-1}} \ar{rr}{\delta}
        &
        & (FUB)^{\otimes n} \ar[shiftarr={xshift=28pt}]{dddd}{f^{\otimes n}}
        \\
        FUF(UB)^{\otimes k} \dar{FUF(\eta U)^{\otimes k}}
        & F(UB)^{\otimes k} \lar[swap]{F \eta} \dar{F( \eta U )^{\otimes k}} \drar[xshift=4, yshift=3]{F(f^{-})^{\otimes k}} \ar[phantom]{rrrrddd}{(\#)}
        &
        & F(UB)^{\otimes n} \rar{F \eta^{\otimes n}} \drar[swap]{F(f^{-})^{\otimes n}}
        & F(UFUB)^{\otimes n} \rar{F \lambda} \dar{F(Uf)^{\otimes n}}
        & FU(FUB)^{\otimes n} \uar[swap]{\varepsilon} \dar[swap]{FU f^{\otimes n}}
        \\
        FUF(UFUB)^{\otimes k} \dar{FUF \lambda} \drar[xshift=-6, pos=.6]{FUF(Uf)^{\otimes k}}
        & F(UFUB)^{\otimes k} \lar[swap]{F \eta} \rar{F(Uf)^{\otimes k}}
        & F(UA)^{\otimes k} \dlar[xshift=6, swap]{F \eta} \dar{F \lambda}
        &
        & F(UA)^{\otimes n} \rar{F \lambda}
        & FUA^{\otimes n} \ar{dd}{ \varepsilon}
        \\
        FUFU(FUB)^{\otimes k} \dar{FU \varepsilon} \drar[xshift=-6]{FUFU f^{\otimes k}}
        & FUF(UA)^{\otimes k} \dar{FUF \lambda}
        & FUA^{\otimes k} \dlar[swap, pos=.4, xshift=6]{F \eta U} \dar[equal]{}
        & & &
        \\
        FU(FUB)^{\otimes k} \ar[shiftarr={yshift=-10pt}, swap]{rr}{FU f^{\otimes k}}
        & FUFUA^{\otimes k} \rar{FU \varepsilon}
        & FUA^{\otimes k} \ar{rrr}{a}
        & & &
        A^{\otimes n}
      \end{tikzcd}
    \]
  \end{tiny}
  \noindent
  To see that $(\#)$ also commutes, note that the counit \(\varepsilon = \id^{+}\) is the
  adjoint transpose of the identity, and transposition is natural. So transposing the two paths
  from \(FUB^{\otimes k}\) to \(A^{\otimes n}\) that form part \((\#)\) yields the inner square of the following commutative diagram in \(\cat D\):
  \[
    \begin{tikzcd}[column sep =2em]
      (UB)^{\otimes k}
      \rar{\lambda}
      \drar[equal]{}
      &
      U B^{\otimes k}
      \rar{{b}^{-}}
      \dar{\lambda^{-1}}
      &
      U B^{\otimes n}
      \rar{\lambda^{-1}}
      &
      (UB)^{\otimes n}
      \rar{(f^{-})^{\otimes n}}
      &
      (UA)^{\otimes n}
      \dar{\lambda}
      \drar[equal]{}
      &
      \\
      &
      (UB)^{\otimes k}
      \rar{({f}^{-})^{\otimes k}}
      &
      (UA)^{\otimes k}
      \rar{\lambda}
      &
      UA^{\otimes k}
      \rar{{a}^{-}}
      &
      UA^{\otimes n}
      \rar{\lambda^{-1}}
      &
      U(A^{\otimes n})
    \end{tikzcd}
  \]
  By pre- and postcomposition of this square with \(\lambda\) and \(\lambda^{-1}\), respectively, and replacing \(a^{-}\) and \(b^{-}\) by their respective conjugates \(\alpha = \lambda^{-1} \cdot a^{-} \cdot \lambda\) and \(\beta = \lambda^{-1} \cdot  b^{-} \cdot \lambda\) this diagram simply becomes the square
  \[
    \begin{tikzcd}
      (UB)^{\otimes k}
      \rar{\beta}
      \dar[swap]{({f}^{-})^{\otimes k}}
      &
      (U B)^{\otimes n}
      \dar{({f}^{-})^{\otimes n}}
      \\
      (UA)^{\otimes k}
      \rar{\alpha}
      &
      (UA)^{\otimes n}
    \end{tikzcd}
  \]
  in \(\cat D\), which is a homomorphism diagram of \((k, n)\)-ary \(U\)-operators.

  We have shown that \(\Phi\) is indeed a functor. Now we verify that it is an equivalence.

  The (natural) isomorphisms
  \[
    \cat{D}(UA, UB)
    \cong
    \cat{C}(FUA, B)
    \cong
    \hatC(\hatB,\hatF \hatU \hatA)
    =
    \hatC(\hatB,T\hatA)
  \]
  given by the duality and the adjunction \(F \dashv U\) ensure that
  this functor is fully faithful.

  To see that it is essentially surjective, pick any \(U\)-operator
  \(c \colon (UC)^{\otimes k} \rightarrow (UC)^{\otimes n}\). We have
  to show that it is isomorphic to an operator of the
  form~\eqref{eq:dual-op}.  Since the original duality
  \(\cat{C} \simeqop \hatC\) is essentially surjective there exists an isomorphism
  \(h \colon C \overset{\sim}\rightarrow \hatX\) for some \(X \in \hatC\).
  Let \(a \colon FU \hatX^{\otimes k} \rightarrow \hatX^{\otimes n}\)
  be the adjoint transpose
  of the \(\cat D\)-morphism
  \[
    U \hatX^{\otimes k}
    \xra{\lambda^{-1}}
    (U\hatX)^{\otimes k}
    \xra{(Uh^{-1})^{\otimes k}}
    (UC)^{\otimes k}
    \xra{c}
    (UC)^{\otimes n}
    \xra{(Uh)^{\otimes n}}
    (U\hatX)^{\otimes n}
    \xra{\lambda}
    U\hatX^{\otimes n}.
  \]
  Then \(Uh \colon (UC, c) \rightarrow (U\hat X, \lambda^{-1} \cdot a^{-} \cdot \lambda)\) is a pure operator isomorphism.
\end{proof}

\begin{rem}\label{rem:heterogeneous-duality}
The definition of $\Phi$ \eqref{eq:Phi} can be slightly generalized to yield a dual correspondence between morphisms $a\colon (UA)^{\otimes k}\to (UB)^{\otimes n}$ of $\cat D$ and morphisms $\rho\colon \hat{B}^{\hatotimesp n}\to \hat TA^{\hatotimesp k}$ of $\hatD$: the dual of $a$ is given by $\rho=\widehat{h^{+}}$, where $h=\lambda \cdot a\cdot \lambda^{-1}\colon UA^{\otimes k}\to UB^{\otimes n}$.
\end{rem}

Our approach to extending dualities is \emph{compositional}
on two levels:
first, compositionality of adjunctions allows us to dualize certain operators more precisely;
second, on the object level, \(U\)-operators themselves admit a compositional structure under morphism composition and tensor product,
leading to \emph{simpler} calculations of dual operators.
We will use the first and second version of compositionality in
\renewcommand{\subsectionautorefname}{Sections}%
\autoref{sec:splitting-adjunction}
\renewcommand{\subsectionautorefname}{Section}
and~\ref{sec:towards-point-free}, respectively.

We elaborate on the first point:
Let \(\cat E\) be a monoidal category with monoidal adjunctions
\[F_{1} \colon \cat E \dashv \cat C \coc U_{1} \qquad\text{and}\qquad
F_{2}\colon \cat D \dashv \cat E \coc U_{2}\]
which \emph{split}
\(F \dashv U\) (i.e.~\(F = F_{1} F_{2}\), \(U = U_{2}U_{1}\)), and
suppose that the monoidal structure of $U$ is given by
\(\lambda = U_{2} \lambda_{1} \cdot \lambda_2U_{1}\).
The compositionality of adjunctions leads to
the following lifting property, applying to both operators
and operator morphisms by setting
\(A=B\) and \(k=n=1\), respectively, in the following proposition.
%
Here we say that, for a monoidal functor \(G\), a morphism \(f \colon (GX)^{\otimes i} \rightarrow (G Y)^{\otimes j}\) \emph{lifts along \(G\)}, if there exists a morphism \(g \colon X^{\otimes i} \rightarrow Y^{\otimes j}\) with \(f = \lambda^{-1} \cdot G g \cdot \lambda\).
The morphism \(g\) is called a \emph{lifting} of \(f\).
\begin{prop}\label{prop:defect}
  Let \(a \colon (UA)^{\otimes k} \rightarrow (UB)^{\otimes n}\)  be a
  morphism in \(\cat D\) dual to \(\rho \colon \hatB^{\hatotimesp n} \rightarrow T \hatA^{\hatotimesp k}\) in \(\hatC\) as in \autoref{rem:heterogeneous-duality}.
  Then the following are equivalent:
\begin{enumerate} \item \(a\) lifts along \(U_{2}\), that is, 
\[a = \lambda_{2}^{-1} \cdot U_{2} b \cdot \lambda_{2} \qquad \text{for some \(b \colon (U_{1}A)^{\otimes k} \rightarrow
      (U_{1}B)^{\otimes n}\)};
\]
\item  \(\rho\) factorizes through the monad morphism
  \(\hatF_{1}\hateps_{2}\hatU_{1} \colon T_{1}
  \rightarrow T\) (where $T_1=\hatF_1 \hatU_1$), that is, 
\[\rho = \hatF_{1}\hateps_{2}\hatU_{1} \cdot \sigma  \qquad\text{for some \( \sigma \colon \hatB^{\hatotimesp n}
      \rightarrow T_{1} \hatA^{\hatotimesp k}\)}.\]
\end{enumerate}
\end{prop}
\begin{proof}
The dual of the morphism \(a \colon (U_{2}U_{1}A)^{\otimes k} \rightarrow (U_{2}U_{1}B)^{\otimes n}\) under the abstract extended duality is given by
  \[\rho=\widehat{h^{+}} \colon \hatB^{\hatotimesp n} \rightarrow T\hatA^{\hatotimesp k} = \hatF_{1}\hatF_{2}\hatU_{2}\hatU_{1}\hatA^{\hatotimesp k}, \]
where $h$ is the unique morphism making the outside of the following diagram commute:
  \[
    \begin{tikzcd}
      U_{2}U_{1}A^{\otimes k}
      \rar{h}
      &
      U_{2}U_{1}B^{\otimes n}
      \dar{U_{2}\lambda_{1}^{-1}}
      \ar[shiftarr={xshift=40pt}]{dd}{\lambda^{-1}}
      \\
      U_{2}(U_{1}A)^{\otimes k}
      \uar{U_{2}\lambda_{1}}
      &
      U_{2}(U_{1}B)^{\otimes n}
      \dar{\lambda_{2}^{-1}}
      \\
      (U_{2}U_{1}A)^{\otimes k}
      \uar{\lambda_{2}}
      \rar{a}
      \ar[shiftarr={xshift=-40pt}]{uu}{\lambda}
      &
      (U_{2}U_{1}B)^{\otimes n}
      \\
    \end{tikzcd}
  \]
  We denote the transposition isomorphisms of the adjunctions \(F_{i} \dashv U_{i}\) by  \[(-)^{\flat_{i}} \colon \cat{C}(F_{i}X, Y) \rightleftarrows \cat{D}(X, U_{i}Y) \coc (-)^{\sharp_{i}}.\]
  Now assume that
  \(\widehat{g^{\sharp_{1}}} \colon \hatB^{\hatotimesp n} \rightarrow T_{1}\hatA^{\hatotimesp k} = \hatF_{1}\hatU_{1}\hatA^{\hatotimesp k}\) is a morphism in \(\hatC\)
  such that we have a factorization \(\widehat{h^{+}} =  \hatF_{1}\hateps_{2}\hatU_{1} \cdot \widehat{g^{\sharp_{1}}}\).
  Under duality this is equivalent to
  \[h^{+} = g^{\sharp_{1}} \cdot F_{1}\varepsilon_{2}U_{1}
   \,\, \LR \,\, h^{\sharp_{2}} = (h^{+})^{\flat_{1}} = g \cdot \varepsilon_{2}U_{1}
   \,\, \LR \,\, h = (h^{\sharp_{2}})^{\flat_{2}} = U_{2}g \cdot (\varepsilon_{2}U_{1})^{\flat_{2}} = U_{2}g \]
  using naturality of the isomorphisms \(\flat_{1}, \flat_{2}\).
  The dual of the \(J_{T_{1}}\)-operator \(\widehat{g^{\sharp_{1}}}\) under the abstract extended duality along the adjunction \(F_{1} \dashv U_{1}\) is the \(U_{1}\)-operator \(b = \lambda_{1}^{-1} \cdot g \cdot \lambda_{1}\).
  Therefore the following diagram commutes:
  \[
    \begin{tikzcd}
      U_{2}U_{1}A^{\otimes k}
      \rar{U_{2}g \,=\, h}
      &
      U_{2}U_{1}B^{\otimes n}
      \dar{U_{2}\lambda_{1}^{-1}}
      \ar[shiftarr={xshift=40pt}]{dd}{\lambda^{-1}}
      \\
      U_{2}(U_{1}A)^{\otimes k}
      \uar{U_{2}\lambda_{1}}
      \rar{U_{2}b}
      &
      U_{2}(U_{1}B)^{\otimes n}
      \dar{\lambda_{2}^{-1}}
      \\
      (U_{2}U_{1}A^{\otimes k})
      \uar{\lambda_{2}}
      \rar{a}
      \ar[shiftarr={xshift=-40pt}]{uu}{\lambda}
      &
      (U_{2}U_{1})B^{\otimes n}
      \\
    \end{tikzcd}
  \]
  This is equivalent to \(a\) admitting the lifting \(b\), that is, \(a = \lambda_{2}^{-1} \cdot U_{2}b \cdot \lambda_{2}\).
\end{proof}

\pagebreak
\begin{rem}\label{rem:defect}\hfill
  \begin{enumerate}
  \item\label{rem:defect:pure} A special case of \autoref{prop:defect}
    proves that extended dualities preserve purity: splitting
    \(F \dashv U\) into \(F_{1} = \Id \dashv \Id = U_{1}\) and
    \(F_{2}= F \dashv U = U_{2}\), we see that a \(U\)-operator (or
    operator morphism) \(a\) is pure iff its dual \(f\) is pure as a
    Kleisli morphism, that is, it factorizes through the unit \(\hat\varepsilon\) of
    the monad \(T\).
      
  \item The right adjoint \(U_{2}\) often is faithful and in this case
    \(\hatF_{1}\hateps_{2}\hatU_{1}\) is monic, that is, \(T_{1}\)
    is a submonad of \(T\): indeed, faithfulness of \(U_{2}\) is equivalent to
    having an epic counit \(\varepsilon_{2} \), hence
    \(\hateps_{2}\hatU_{1}\) is monic, and the right adjoint
    \(\hat{F_{1}}\) preserves monos.  In particular, if \(T\) is
    `powerset-like', then~\(\hatC_{T}\) is a category of relations,
    and we think of $U$-operators (or operator morphisms) of the form
    \(a = \lambda_{2}^{-1} \cdot U_{2}b \cdot \lambda_{2}\) as
    dualizing to `more functional' relations.  This idea is
    illustrated by the examples in \autoref{sec:splitting-adjunction}.
  \end{enumerate}
\end{rem}

The compositionality on the level of \(U\)-operators manifests itself
as follows:

\begin{prop}\label{prop:composite-duality}
  Let \(h, g \colon UA \rightarrow UA\) be \(U\)-operators with respective duals \(\rho, \sigma \colon \hatA \rightarrow T\hatA\).
  On objects the abstract extended duality of \autoref{thm:extended-duality} preserves
  \begin{enumerate}
    \item\label{prop:composite-duality:tensor} \emph{Tensor products of operators:} \[h \otimes g \colon UA \otimes UA \rightarrow UA \otimes UA \quad \text{is dual to} \quad
          \hat{\delta} \cdot (\rho \hatotimes \sigma) \colon \hatA \hatotimes \hatA \rightarrow T \hatA \hatotimes T\hatA \rightarrow T (\hatA \hatotimes \hatA). \]
    \item\label{prop:composite-duality:composition} \emph{Composition of operators:} \[g \cdot h \colon UA \rightarrow UA \rightarrow UA \quad \text{is dual to} \quad  m \cdot T \rho \cdot \sigma \colon \hatA \rightarrow T \hatA \rightarrow TT \hatA \rightarrow T \hatA.\]
    \item\label{prop:composite-duality:identity} \emph{Identity operators:} \[\id_{UA} \colon UA \rightarrow UA \quad \text{is dual to} \quad e_{\hatA} \colon \hatA \rightarrow T \hatA.\]
  \end{enumerate}
\end{prop}
\begin{proof}\hfill
  \begin{enumerate}
    \item Under the extended duality the operator \(h \otimes g\) is
      mapped to \(\widehat{\alpha^{+}}\), the dual of the adjoint
      transpose \(\alpha^{+}\) of the conjugate $\alpha$ of \(h \otimes g\):
      \[
        \alpha = \big(
        U(A\otimes A)
        \xra{\lambda^{-1}}
        UA \otimes UA
        \xra{h \otimes g}
        UA \otimes UA
        \xra{\lambda}
        U(A \otimes A)
        \big).
      \]
      We now prove that the following diagram commutes, where \(\delta\) is the comonoidal structure of \(FU\) dual to the monoidal structure \(\hat{\delta}\) of \(T\) from \autoref{rem:mon-adj-kleisli}.
      \[
          \begin{tikzcd}
            FU(A \otimes A)
            \rar{F \lambda^{-1}}
            \drar{{\delta}}
            \dar{\alpha^{+}}
            &
            F(UA \otimes UA)
            \rar{F(\eta \otimes \eta)}
            \ar[shiftarr={yshift=20pt}]{rrr}{F(h \otimes g)}
            &
            F(UFUA \otimes UFUA)
            \ar{rr}{F(Uh^{+} \otimes Ug^{+})}
            \dar{F \lambda}
            &&
            F(UA \otimes UA)
            \dar{F\lambda}
            \\
            A \otimes A
            &
            FUA \otimes FUA
            \lar[swap]{h^{+} \otimes g^{+}}
            &
            FU(FUA \otimes FUA)
            \lar[swap]{\varepsilon}
            \ar{rr}{FU(h^{+} \otimes g^{+})}
            &&
            FU(A \otimes A)
            \ar[swap, shiftarr={yshift=-20pt}]{llll}{\varepsilon}
          \end{tikzcd}
        \]
        Its outside commutes using the definition of the adjoint transpose $\alpha^+$, and the upper part also commutes by adjoint transposition. The right-hand and lower parts commute by naturality of $\lambda$ and $\varepsilon$, respectively. The middle part commutes trivially by the definition of \(\delta\).
        Thus, the left-hand triangle commutes.
        This shows that, under the duality \(\cat{C} \simeqop \hatC\),  the dual of \(\alpha^{+}\) is equal to the dual of \((h^{+} \otimes g^{+}) \cdot {\delta}\), which in turn is given by \(\hat{\delta} \cdot (\rho \hatotimes \sigma)\).
        
    \item Similarly, the adjoint transpose of \(g \cdot h\) is equal to \(g^{+} \cdot FUh^{+} \cdot F\eta U\) whose dual is given by \(m \cdot T \rho \cdot \sigma\).
    \item The adjoint transpose of \(\id_{UA}\) is \(\varepsilon_{A}\), whose dual is the unit \(\hat{\varepsilon}_{\hat{A}} = e_{\hat{A}} \) of \(T\). \qedhere
  \end{enumerate}
\end{proof}

\section{Example: Extended Priestley Duality}
\label{sec:exampl-extend-priest}
As a first application of our adjoint framework, we investigate the
classical Priestley duality (\autoref{sec:prelim}) and derive a
generalized version of Goldblatt's duality~\cite{goldblatt-89} between
distributive lattices with operators and relational Priestley
spaces. We instantiate \eqref{eq:adjoints} to the following categories
and functors, which we will subsequently explain in detail:
  \begin{equation}\label{eq:asm-gold}
    \begin{tikzcd}[sep = large]
      \cat D
      \rar[phantom]{\simeqop}
      \dar[bend right=30pt, swap]{F}
      &
      \hatD
      \ar[bend right=30pt]{d}[swap,pos=.55]{\hatF}
      \\
      \cat C
      \rar[phantom]{\simeqop}
      \uar[bend right=30pt, swap]{U}
      \uar[phantom]{\dashv}
      &
      \hatC
      \rar[loop right]{T}
      \uar[phantom]{\vdash}
      \ar[bend right=30pt]{u}[swap,pos=.45]{\hatU}
    \end{tikzcd}
\qquad=\qquad\;\;\;
    \begin{tikzcd}[sep = large, row sep=3em]
      \cat \jsl
      \rar[phantom]{\simeqop}
      \dar[bend right=30pt, swap]{F}
      &
      \stonejsl
      \ar[bend right=30pt]{d}[swap,pos=.49]{\hatF}
      \\
      \dl
      \rar[phantom]{\simeqop}
      \uar[bend right=30pt, swap]{U}
      \uar[phantom]{\dashv}
      &
      \priest
      \rar[loop right]{\dviet}
      \uar[phantom]{\vdash}
      \ar[bend right=30pt]{u}[swap,pos=.45]{\dviet}
    \end{tikzcd}
  \end{equation}

\paragraph*{Categories} The upper duality is Hofman-Mislove-Stralka
duality~\cite{hofmann-mislove-stralka-74} between the category of join-semilattices with a
bottom element and the category of Stone semilattices (i.e.~topological join-semilattices with
a bottom element whose underlying topological space is a Stone space) and continuous
semilattice homomorphisms. The duality maps a join-semilattice \(J\) to the Stone semilattice
\(\jsl(J, 2)\) of semilattice homomorphisms into the two-element semilattice, topologized by
pointwise convergence. Equivalently, \(\jsl(J, 2)\) is the space $\idl(J)$ of \emph{ideals}
(downward closed and upward directed subsets) of \(J\), ordered by reverse inclusion, with topology generated by the subbasic open sets \(\sigma(j) = \{I \in \idl(J) \mid j \in I\}\) and their complements for \(j \in J\).
In the other direction, a Stone semilattice \(X\) is mapped to its semilattice \(\stonejsl(X, 2)\) of clopen ideals, ordered by inclusion.

\paragraph*{Functors} The functor
\(U \colon \dl \rightarrow \jsl\) is the obvious forgetful
functor. Its left adjoint \(F \colon \jsl \rightarrow \dl\) maps a
join-semilattice \(J\) to the set \(\mathcal{U}_{\mathrm{fg}}^{\partial}(J)\)
of finitely generated upsets of \(J\), ordered by reverse inclusion.
The dual right adjoint \(\hatF\) of \(F\) is the
forgetful functor mapping a Stone semilattice to its underlying Priestley space.
Indeed, since \(F \dashv U\),  we compute for the carrier set \(|X|\) of a Stone semilattice \(X\) that
\[|\hatF X| = |\widehat{F \hatX}| \cong |\dl(F \hatX, 2)| \cong |\jsl(\hatX, U2)| = |\jsl(\hatX, 2)|  \cong |X|,\]
and this bijection is a homeomorphism.

The left adjoint \(\hatU \colon \priest \rightarrow \stonejsl\) maps a Priestley
space \(X\) to the Stone join-semilattice
\[\widehat{U \hatX} = \jsl(U(\priest(X, 2)), 2) \cong
  \idl(\cl_{\uparrow} X) \] of ideals of clopen upsets
of \(X\).  This space is isomorphic to the \emph{(downset) Vietoris
  hyperspace} \(\dviet X\) of \(X\) carried by the set of
closed downsets of \(X\).  The isomorphism is given by
\begin{align*}
  \idl(\cl_{\uparrow}X) \qquad &\cong \qquad \dviet X \\
  I \qquad &\mapsto \qquad \bigcap_{U \in I} X \setminus U \\
  \{U \in \cl_{\uparrow}X \mid C \subseteq X \setminus U\} \qquad &\mapsfrom \qquad C.
\end{align*}
The topology of pointwise convergence on \(\jsl(U(\priest(X, 2)), 2)\)
translates to the \emph{hit-or-miss topology} on \(\dviet X\)
generated by the subbasic open sets
\[ \{A \subseteq X \text{ closed } \mid A \cap U \ne \emptyset\}
  \q{\text{for}} U \in \cl_{\uparrow}X\] and their complements.
Note that for a Stone space \(X\), the Stone join-semilattice \(\dviet X\) is the \emph{free Stone join-semilattice}
on \(X\),
as observed by Johnstone~\cite[Sec.~4.8]{johnstone-82}.
The monad induced by the adjunction is the \emph{(downset) Vietoris monad};
its unit
\(e \colon X \rightarrow \dviet X\) is given by
\(x \mapsto {\downarrow}x\) and multiplication is given by union~\cite{hofmann-15}.
The monad \(\dviet\) restricts to the Vietoris monad \(\viet\) on the category  $\stone$ of {Stone spaces}.

The duality of modal algebras and coalgebras for the Vietoris construction, going back to Esakia~\cite{esa74}, has often been rediscovered and extended since, see e.g.~\cite{kupke-kurz-venema-04,vj14}; for recent accounts with detailed computations regarding the dualities of $FU$ and $\dviet$ see Bezhanishvili et~al.~\cite{bezhanishvili-harding-morandi-23} or the textbook by Gehrke and van Gool~\cite[Section 6.4]{gg24}.

\begin{rem}[Continuous Relations]\label{rem:cont-rel}
  Continuous maps in \priest of the form \[\rho \colon X \rightarrow \dviet Y\] are known in the literature under a variety of names; we call them as \emph{Priestley relations} as in \cite{cignoli-91,goldblatt-89}, or \emph{Stone relations} if \(X, Y\) are Stone spaces. We write $x\mathbin{\rho} y$ for $y\in \rho(x)$, and sometimes identify $\rho$ with a subset of $X\times Y$. Let us note that some authors (e.g.~\cite{rhodes-09}) call a relation $R\subseteq X\times Y$ between topological spaces \emph{continuous} if it is closed as a subspace of $X\times Y$.
  Every Priestley relation is continuous. The converse generally fails: for instance, consider any non-discrete Stone space $X$ and let $C\subseteq X$ be a subset that is closed but not open. The relation $C\times 1\subseteq X\times 1$ is closed, but the corresponding map $\rho\colon X\to \viet 1$ (given by $\rho(x)=1$ if $x\in C$ and $\rho(x)=\emptyset$ otherwise) is not continuous because $\rho^{-1}[1]=C$ is not open.
\end{rem}

\paragraph*{Monoidal Structure} The category \(\jsl\) of
join-semilattices has a tensor product \(\otimes\) that represents \emph{join-bilinear} maps, that is, maps $J \times J' \to K$ between join-semilattices preserving finite joins in each argument:
\[\mathrm{Bilin}(J\times J', K) \;\cong\; \jsl(J \otimes J', K).\]
Join-bilinear maps $J\times J'\to K$ and their corresponding $\jsl$-morphisms $J\otimes J'\to K$ are often tacitly identified. The tensor product $\otimes$ makes \(\jsl\) a monoidal category with unit \(2\),
i.e.~\(2 \otimes J \cong J\).  The standard presentation of $J  \otimes J'$ as a join-semilattice is given
by the generators \(\{j \otimes j' \mid j \in J, j' \in J'\}\)
with equations
\[j_{1} \otimes 0 = 0 \otimes j'_{1} = 0, \quad (j_{1} \lor j_{2})\otimes j' = j_{1} \otimes j' \lor j_{2} \otimes j' \quad\text{and}\quad j \otimes (j'_{1} \lor j'_{2}) = j \otimes j'_{1} \lor j \otimes j'_{2} \]
ranging over \( j_{1}, j_{2} \in J\) and \(j'_{1}, j'_{2} \in J'\).
We call (the equivalence class of) a generating element \(j \otimes j'\) a \emph{pure tensor}.
If \(D, D'\) are bounded distributive lattices then so is \(UD \otimes UD'\)~\cite{fraser-76},
with meet given on pure tensors by \((d \otimes d') \land (e \otimes e') = (d \land e) \otimes (d' \land e')\).
Moreover, the lattice \(UD \otimes UD'\) is the coproduct of \(D, D'\) in \(\dl\):
the coproduct injections are given by
\[\iota(d) = d \otimes 1'\qquad\text{and}\qquad\iota'(d') = 1 \otimes d'\]
for \(d \in D, d' \in D'\), and the copairing of lattice homomorphisms \(f \colon D \rightarrow E, f' \colon D' \rightarrow E\) is given by the extension of the join-bilinear map \[\land \cdot (f \times f') \colon D \times D' \rightarrow E,\qquad (d, d') \mapsto f(d)\land f(d').\]
Taking coproducts yields a monoidal structure on \(\dl\), and since \(U(D + D') = UD \otimes UD'\), the functor \(U\) is strict monoidal.
\begin{lem}
  The dual monoidal structure \(\hat{\delta}\) of $\dviet$ is given by product:
  \[\hat{\delta} \colon \dviet X \times \dviet Y \rightarrow \dviet (X \times Y), \qquad (C, D) \mapsto C \times D. \]
\end{lem}
\begin{proof}
  Let \(C \in \dviet X, D \in \dviet Y\) be closed downsets.
  We first represent them by their respective ideals \(I_{C} ,I_{D}\) of \(\cl_{\uparrow}X\) and \(\cl_{\uparrow} Y\), which are equivalently $\jsl$-morphisms
  \[c \colon U \cl_{\uparrow}X \rightarrow U 2,\,  d \colon U \cl_{\uparrow} Y \rightarrow U 2. \]
  The map \(\hat{\delta}\) is the given as the dual
  \[\dl(FU \cl_{\uparrow}X + FU \cl_{\uparrow} Y, 2) \rightarrow \dl(FU(\cl_{\uparrow}X + \cl_{\uparrow}Y), 2) \]
  of the comonoidal structure $FU(A + B) \rightarrow FU A + FU B$
  mapping $[c^+, d^+]$ to the prime filter that is the  transpose of
  \begin{align*}
    U[c^{+}, d^{+}] \cdot (\eta \otimes \eta) \colon  &U(\cl_{\uparrow}X + \cl_{\uparrow}Y) \\
    \cong\, &U(\cl_{\uparrow}X) \otimes U(\cl_{\uparrow}Y) \rightarrow UFU\cl_{\uparrow}X \otimes UFU \cl_{\uparrow}Y \\
    \cong\, &U(FU \cl_{\uparrow} X + FU \cl_{\uparrow} Y) \rightarrow U2.
  \end{align*}
 The latter map sends a pure tensor \(A \otimes B \in U(\cl_{\uparrow} X + \cl_{\uparrow} Y)\) to its `product' \(c(A) \land d(B)\).
  Therefore the closed set \(\hat{\delta}(C, D)\) corresponding to
  \(U[c^{+}, d^{+}] \cdot (\eta \otimes \eta)\) contains a pair $(x,y)$ iff \(x \in C\) and \(y \in D\).
\end{proof}

Expanding \autoref{def:operator}, the category \(\op_{J_{\dviet}}^{n,k}(\priest)\) is given as follows:
\begin{defi}
  A (\emph{\((n, k)\)-ary}) \emph{relational Priestley space}  consists of a
  carrier Priestley space~\(X\) and a Priestley relation
  \(\rho \colon X^{n} \rightarrow \dviet X^{k}\).  A \emph{relational
    morphism} from a relational Priestley space \((X, \rho)\) to a
  relational Priestley space
  \((X', \rho')\) is given by a Priestley relation
  \(\beta \colon X \rightarrow \dviet Y\) such that, for all
  $\mathbf{x} \in X^{n}$, $\mathbf{y} \in X^{k}$, and $\mathbf{y}' \in
  X'^{k}$,%
  \[
    \mathbf{x} \mathrel{\rho} \mathbf{y} \land (\forall i \colon y_{i} \mathrel{\beta}
    y_{i}') \implies \exists \mathbf{x}' \colon (\forall i \colon
    x_{i} \mathrel{\beta} x_{i}') \land \mathbf{x}' \mathrel{\rho'} \mathbf{y}',
  \]
  and, for all $\mathbf{x} \in X^{n}$, $\mathbf{x}' \in X'^{n}$, and $\mathbf{y}' \in X'^{k}$,
  \[
    (\forall i \colon x_{i} \mathrel{\beta} x_{i}') \land \mathbf{x}' \mathrel{\rho'}
    \mathbf{y}' \implies \exists \mathbf{y} \colon \mathbf{x} \mathrel{\rho}
    \mathbf{y} \land (\forall i \colon y_{i} \mathrel{\beta} y_{i}').
  \]
We denote by \(\op_{J_{\dviet}}^{n,k}(\priest)\) the category of $(n,k)$-ary relational Priestley spaces and relational morphisms.
\end{defi}
Then \autoref{thm:extended-duality} instantiates to the following result:

\begin{thm}[Extended Priestley duality]\label{thm:extended-priestley-duality}
 The category of \((k, n)\)-ary
  $U$-operators of distributive lattices is dually equivalent to the category of \((n, k)\)-ary relational Priestley spaces and relational
  morphisms:
  \[\op_{U}^{k,n}(\dl) \;\simeqop\; \op_{J_{\dviet}}^{n,k}(\priest).\]
\end{thm}

By setting \(n = 1\) and restricting the operator morphisms on both sides to be pure (\autoref{rem:defect}\ref{rem:defect:pure}),
we recover Goldblatt's duality~\cite{goldblatt-89}. In the latter work,
pure relational morphisms are called \emph{bounded morphisms} and $n$-ary $U$-algebras $(UD)^{\otimes n}\to UD$ in $\jsl$ are called \(n\)-ary \emph{join-hemimorphisms}.
\begin{cor}[Goldblatt duality]
The category of distributive lattices with $n$-ary join-hemimorphisms, and pure morphisms between them, is dually equivalent
to the category of \((1, n)\)-relational Priestley spaces and bounded morphisms.
\end{cor}

\subsection{Deriving Concrete Formulas}
\label{sec:deriv-concr-form}

We proceed to show how an order-enriched extension of our adjoint framework can be used to
methodically derive concrete (i.e.~element-based)
formulas for the dual join operator of a continuous relation and vice versa.
Let us first observe that all involved
categories are \emph{order-enriched} if we equip the homsets
with the usual pointwise order on functions.
Moreover, from the definitions it is clear that the
transposing isomorphisms of the adjunction \(F \dashv U\) and the
duality \(\dl \simeqop \priest\) are order-isomorphisms.

Second, we can represent an element \(\hat{x}\) of a Priestley space \(\hatX\) as a continuous function \(\hat{x} \in \priest(1, \hatX)\);
on the lattice side, elements \(j\) of a join-semilattice \(J\) correspond bijectively to $\jsl$-morphisms $j \in \jsl(2, J)$, using that $2$ is the free semilattice on a single generator. 

For the rest of this subsection let us fix a $U$-algebra $h$ and its dual Priestley relation $\rho$:
\[h \colon (UX)^{\otimes n} \rightarrow U X\qquad\text{and}\qquad \rho \colon \hatX \rightarrow \dviet \hatX^{n}.\]
 We first show how to express $\rho$ in terms of~$h$. Viewing $\rho$ as a relation from $\hatX$ to $\hatX^n$ (\autoref{rem:cont-rel}), two elements \(\hat{x} \in \hatX, \mathbf{\hat{x}} \in \hatX^{n}\) are related by \(\rho\) (i.e.\
\(\hat{x} \mathrel{\rho} \mathbf{\hat{x}}\)) iff the inequality
\(e(\mathbf{\hat{x}}) = {\downarrow}{\mathbf{\hat{x}}} \le
\rho(\hat{x})\) holds,
equivalently, iff the left diagram below
commutes laxly as indicated:
\[
  \begin{tikzcd}
    \hatX
    \rar{\rho}
    \drar[phantom]{\rotatebox{-45}{\(\ge\)}}
    &
    \dviet \hatX^{n}
    &
    &&
    UX
    \dar[swap]{U x}
    \drar[phantom]{\rotatebox{-45}{\(\ge\)}}
    &
    (UX)^{\otimes n}
    \lar[swap]{h}
    \dar{\bigotimes_{i}U x_i}
    \\
    1
    \uar{\hat{x}}
    \rar{\mathbf{\hat{x}}}
    \ar[shiftarr={yshift=-15pt}]{rr}{\Delta}
    &
    \hatX^{n}
    \uar[swap]{e}
    &
    1^n
    \lar[swap]{\prod_{i}\hat{x}_{i}}
    &&
    U 2
    &
    (U2)^{\otimes n}
    \lar[swap]{\nabla }
  \end{tikzcd}
\]
The duals of \(\hat{x}, \hat{x}_{i}\) are $\dl$ morphisms
\(x, x_{i} \colon X \rightarrow 2\).
Under duality and transposition the left diagram
corresponds to the right diagram where \(\nabla\) is the codiagonal
given by \(n\)-fold conjunction, i.e.~it maps
\(\bigotimes_{i=1}^{n} x_{i}\) to \(\A_{i=1}^{n}x_{i}\).
Writing \(F_{z} = z^{-1}(1)\) for the prime filter corresponding to a morphism \(z \in \dl(X, 2)\) the right diagram yields Goldblatt's formula~\cite[p.~186]{goldblatt-89}
for the dual Priestley relation of an algebra~$h$: we have 
\[\hat{x}
\mathrel{\rho} \mathbf{\hat{x}} \qquad\text{iff}\qquad  h[\prod_{i}F_{x_{i}}]
\subseteq F_{x}.\]
Conversely, to express $h$ in terms of $\rho$, it suffices to describe
$h(\mathbf{x})$ for a pure tensor \(\mathbf{x} \in (UX)^{\otimes n}\) by the universal property of the tensor product.
We factorize
\[
  \mathbf{x} = \bigotimes_{i} x_{i} \cdot \nabla^{-1} \colon U 2 \cong (U 2)^{\otimes
    n} \rightarrow (U X)^{\otimes n}
\]
to see that the element \(h(\mathbf{x})\) corresponds to the following morphism
representing an element of the join-semilattice $UX$:
\[h \cdot \bigotimes_{i} x_{i} \cdot \nabla^{-1} \colon U 2 \cong (U 2)^{\otimes n} \rightarrow (U X)^{\otimes n} \rightarrow U X.\]
Its dual is the characteristic function
\[\hatX \xra{\rho} \dviet \hatX^{n} \xra{\dviet(\prod_{i} C_{i})} \dviet (\dviet 1)^{n} \xra{\dviet \hat{\delta}} \dviet \dviet 1^{n} \xra{m} \dviet 1^{n} \xra{\dviet \Delta^{-1}} \dviet 1 = 2,  \]
where \(C_{i}= \widehat{x_{i}^{+}}\) is the clopen upset of \(\hatX\) dual to
\[x_{i} \in \jsl(U 2, U X) \;\cong\; \dl(FU 2, X) \;\cong\; \priest(\hatX, \dviet 1) \;\cong\; \priest(\hatX, 2). \]
This shows that \(h(\mathbf{x}) \in X \;\cong\; \cl_{\uparrow}\hatX\) corresponds to the clopen upset
\[h(\mathbf{x}) =  \{a \in \hatX \mid \exists (b_{1}, \ldots , b_{n}) \in \rho(a) \colon \forall i \colon b_{i} \in C_{i}  = \widehat{x_{i}^{+}} \} \in \cl_{\uparrow}(\hatX),\]
so we derived Goldblatt's formula~\cite[p.\ 184]{goldblatt-89} for the dual algebra of a relation $\rho$.

\subsection{Functional Properties of Priestley Relations.}
\label{sec:splitting-adjunction}
As a further application of the adjoint framework to extended Stone duality, we show how to recover the characterization of those operators on distributive lattices  whose dual Priestley relation is a partial function or a total relation, respectively.
As outlined in \autoref{rem:defect}, we achieve this by considering
suitable splittings of the adjunction \( F \colon \jsl \dashv \dl \coc U\) to obtain submonads of \(\dviet\) on which we instantiate \autoref{prop:defect}.

\paragraph{Partial Functions.} We split the adjunction $F\dashv
U$ into
\[F_1 \colon \dl_{0} \dashv \dl \coc U_1 \qquad \text{and} \qquad F_2 \colon \jsl \dashv \dl_0 \coc U_2,\]
where \(\dl_{0}\) is the
category of distributive lattices that are only bounded from below,
and $U_1,U_2$ are forgetful functors.
The tensor product of objects of \(\dl_{0}\) is given by the tensor product of their underlying join-semilattices.
The left adjoint \(F_1\) adds a top element
to a lattice in \(\dl_{0}\).
\begin{lem}\label{lem:adj-split-partial}
  The submonad \(T_{1}\) on \priest dual to the adjunction \(F_{1} \dashv U_{1}\) is given by the partial function monad  \[T_{1}X = X + \{\emptyset\}.\]
\end{lem}
\begin{proof}
The submonad
\(T_{1} = \hatF_1\hatU_1 \incl \dviet \) on \priest is given by
\[\hatF_1\hatU_1\widehat{D} \;\cong\; \widehat{F_1U_1D} \;\cong\; \dl(F_1U_1D, 2) \;\cong\; \dl_{0}(U_1D, U_12). \]
Every morphism \(f \in \dl_{0}(U_1D, U_12)\) either satisfies \(f(\top) = \top\), in
which case \(f \in \widehat{D}\) is prime; or \(f(\top) = \bot\), but then
\(f\) is the constant zero map \(\bot! \colon U_1D \rightarrow U_12\).
The map \(\bot!\) is the bottom element in the pointwise ordering of
\(\dl_{0}(U_1D, U_12)\), so the monad \(\hatF_1\hatU_1\) just freely adjoins a bottom
element to a Priestley space.
\end{proof}

In particular, the dual category of \(\dl_{0}\) is readily
seen to be equivalent to \(\priest_{0}\), the category of Priestley
spaces with a bottom element, and bottom-preserving continuous
monotone maps.  A Kleisli morphism
\(\rho \colon X \rightarrow T_1 X \incl \dviet X \) is a
\emph{partial continuous function}, and a $U$-operator lifts along $U_2$ iff it preserves non-empty meets.
\autoref{prop:defect} thus recovers the following result (for the unary version see~\cite{halmos-58,hofmann-15}).

\begin{cor}\label{cor:par-fun}
  The dual Priestley relation of a \(U\)-operator is a partial function iff the operator preserves non-empty meets.
\end{cor}

\paragraph{Total Relations.} We split the adjunction $F\dashv
U$ into
\[F'_1 \colon \jsl_1 \dashv \dl \coc U_1' \qquad\text{and}\qquad 
F'_2 \colon \jsl \dashv \jsl_1 \coc U'_2,\]
 where $\jsl_1$ is the category
of join-semilattices with both a bottom and top element with morphisms
preserving joins, bottom, and top. The right adjoints $U'_1,U'_2$ are forgetful functors,
and the tensor product of objects of \(\jsl_{1}\) is given by the tensor product of their underlying join-semilattices.
The left adjoint $F'_1$ maps \(J\in \jsl_1\) to the distributive
lattice \(\mathcal{U}_{\mathrm{fg}+}^{\partial}J\) of \emph{non-empty}
finitely generated upsets of \(J\), ordered by reverse inclusion.
\begin{lem}\label{lem:adj-split-ne}
  The submonad \(T'_{1}\) on \priest dual to the adjunction \(F'_{1} \dashv U'_{1}\) is given by the non-empty Vietoris monad
  \[T'_{1} X = \dvietplus X.  \]
\end{lem}
\begin{proof}
We have
\[\hatF'_1\hatU'_1\widehat{D} \;\cong\; \dl(F'_1U'_1 D, 2) \;\cong\; \jsl_{1}(U'_1D, U'_12).\]
The only ideal \(f \in \jsl(U D, U 2)\) that is not an
element of \(\jsl_1(U'_{1} D, U'_{1} 2)\)
is the trivial ideal \(\bot! \colon d \longmapsto \bot\): if $f(\top)\neq \top$ then $f(\top)=\bot$, so $f$ maps all elements to $\bot$ by monotonicity.. The trivial ideal corresponds to the empty set \(\emptyset \in \dviet \widehat{D}\).
Thus \(T'_{1} X = \dviet X \setminus \{\emptyset\}\).
\end{proof}
Kleisli morphisms  \(X \rightarrow T_{1}' Y\) therefore are simply \emph{total} Priestley relations. Moreover, a $U$-operator lifts along $U_2'$ iff it preserves the top element. \autoref{prop:defect} then yields the following result (for the unary version see~\cite{halmos-58,hofmann-15}).
\begin{cor}\label{cor:dual-total}
  The dual Priestley relation of a \(U\)-operator is a total relation iff the operator preserves the top element.
\end{cor}

\subsection{Equational Properties of Operators}
\label{sec:towards-point-free}

In \autoref{sec:splitting-adjunction} we have shown how to use factorizations of the adjunction \(F \dashv U\) to obtain more precise dualities for operators with additional properties.
However, spelling out and computing the factorizing adjunctions and their dual monads for the desired operator property individually is tedious, and there often is a simpler method:
If an operator property can be described \emph{equationally},
we can use the fact that inequations of
operators translate to inequations of their dual relations to obtain additional information under dualization.

The results we recover in this section are usually associated with modal correspondence theory~\cite{ben77,gt75,brv01}. Note that our proofs neither use first-order logic nor canonical frames: after having found a suitable encoding of order-theoretic properties as operator (in-)equations they are mere instantiations of the abstract results from \autoref{sec:extending-dualities}.

\begin{prop}\label{prop:operator-props}
  Let \(h \colon UA \rightarrow UA\) be an operator on a bounded distributive lattice \(A\) with dual Priestley relation \(\rho \colon \hatA \rightarrow \dviet \hatA\).
  The following correspondences hold:
  \begin{enumerate}
    \item  \(\rho\) is reflexive iff \(\forall a \in A \colon a \le h(a)\).
    \item  \(\rho\) is symmetric iff \(\forall a, b \in A \colon a \land h(b) \le h (h(a) \land b)\).
    \item  \(\rho\) is euclidean iff \(\forall a, b \in A \colon h(a) \land h(b) \le h (a \land h(b))\).
    \item  \(\rho\) is transitive iff \(\forall a, b \in A \colon h(a) \land h(b) \ge h (a \land h(b))\).
    \item  \(\rho\) is total iff \(h(\top) = \top \).
\item $\rho$ is empty iff \(h(\top) = \bot\).
  \end{enumerate}
\end{prop}
\noindent
The proof follows a simple pattern: (a)~translate the inequation into an inequation between operators, (b)~dualize the inequation between operators to an inequation between Priestley relations and (c) convert the relational inequality into the corresponding first-order property.
\begin{proof}
  First we dualize some canonical operators whose combinations we can then dualize using \autoref{prop:composite-duality}.
  Recall from \iref{prop:composite-duality}{identity} that the identity operator \(\id_{UA} \colon UA \rightarrow UA\) dualizes to the unit \(e \colon \hatA \rightarrow \dviet \hatA\).
  The conjunction operator \(\land \colon UA \otimes UA \rightarrow UA\) is equal to the conjugate \[ \lambda  \cdot U \nabla \cdot \lambda^{-1} \colon UA \otimes UA \cong U(A + A) \rightarrow UA \cong UA \]
  and thus dualizes to the diagonal function \(\Delta \colon \hatA \rightarrow \hatA \times \hatA\).

  Now let \(h \colon UA \rightarrow UA\) be a unary operator with dual Priestley relation \(\rho \colon \hatA \rightarrow \dviet \hatA\).
  \begin{enumerate}
    \item The condition \(\forall a \in A \colon a \le h(a)\) is equivalent to the operator inequality  \(\id_{UA} \le h\).
          The dual inequality is given by \(e \le \rho \colon \hatA \rightarrow \dviet \hatA\), that is, \(\forall x \in \hatA \colon \{x\} \subseteq \rho(x)\), which states precisely that \(\rho\) is reflexive.
    \item The condition \(\forall a, b \in A \colon a \land h(b) \le h (h(a) \land b)\) is equivalent to \(\land \cdot (\id \otimes h) \le h \cdot \land \cdot (h \otimes \id)\).
          The left side dualizes to \[\hat{\delta} \cdot (e \times \rho) \cdot \Delta \colon \hatA \rightarrow \hatA \times \hatA \rightarrow \dviet \hatA \times \dviet \hatA \rightarrow \dviet (\hatA \times \hatA) \]
          which is given by \(x \mapsto \{(x, y) \mid y \in \rho(x)\}\),
          while the right side dualizes to
          \[m \cdot \hat{\delta} \cdot \dviet (\rho \times e) \cdot  \dviet \Delta \cdot \rho \colon \hatA \rightarrow \dviet \hatA \rightarrow \dviet (\hatA \times \hatA) \rightarrow \dviet (\dviet \hatA \times \dviet \hatA) \rightarrow \dviet \dviet (\hatA \times \hatA) \rightarrow \dviet (\hatA \times \hatA)\]
          which is given by \(x \mapsto \{(z, y) \mid y \in \rho(x), z \in \rho(y)\}\).
          On elements this inequality of relations reads
          \[ \forall x, y, z \in \hatA \colon y \in \rho(x) \Rightarrow x \in \rho(y), \]
          which states precisely that \(\rho\) is symmetric.
    \item[(3), (4)] The proof is analogous to part (2) and left as an exercise to the reader.
    \item[(5)] Here we have to use two auxiliary operators that exist on every Boolean algebra:
          the `bottom' operator \[z \colon UA \rightarrow UA,\quad x \mapsto \bot \qquad \text{whose dual is the empty relation} \qquad \zeta \colon \hatA \rightarrow \dviet \hatA,\quad x \mapsto \emptyset.\]
          and the `top' operator \[t \colon UA \rightarrow UA,\quad x \mapsto
          \begin{cases}
            \bot & \text{if } x = \bot \\
            \top & \text{otherwise,}
          \end{cases}
          \qquad \text{whose dual is} \qquad \tau \colon \hatA \rightarrow \dviet \hatA,\quad x \mapsto \hatA.\]
          The equation \(h(\top) = \top \) holds iff the equation \(h \cdot t = t \colon UA \rightarrow UA\) of operators holds.
          The left side dualizes to%
          \[m \cdot \dviet \tau \cdot \rho \colon \hatA \rightarrow \dviet \hatA \rightarrow \dviet \dviet \hatA \rightarrow \dviet \hatA,\quad x \mapsto \bigcup_{y \in \rho(x)} \tau(y) =
            \begin{cases}
              \hatA & \text{if $\rho(x)$ is non-empty}, \\
              \emptyset & \text{else}.
            \end{cases}
            \]
          The right side is simply \(\tau \colon \hatA \rightarrow \dviet \hatA, x \mapsto \hatA\),
          and these relations are equal iff \(\rho\) is total, that is, \(\rho(x)\) is non-empty for every \(x\).
   \item  One has \(h(\top) = \bot\) iff the operator equation \(h = z\) holds.
          Its dual equation \(\rho = \zeta\) simply states that \(\rho(x)\) is empty for every \(x\). \qedhere
  \end{enumerate}
\end{proof}

\begin{rem}
  The axioms (1)--(4) considered in \autoref{prop:operator-props} correspond to the classical
  modal axioms (T), (B), (D) and~(5).  Usually, axioms (2)--(4) are phrased differently by
  using \emph{Boolean} modal logic: (2) is equivalent to
  \(\forall a \colon a \le \neg h (\neg h(a))\), (3) to
  $\forall a \colon h(a) \le \neg h(\neg h(a))$ and (4) to \(h(h(a)) \le h(a)\).  We show that
  axiom (2) is equivalent to (B) and leave the rest as an easy exercise for the reader.

  Suppose that~(2) holds in a Boolean algebra \(B\), that is, it satisfies
  \(\forall a, b \in B \colon a \land h(b) \le h(h(a) \land b)\). We prove that it 
  satisfies \(\forall a \in B \colon a \le \neg h (\neg h(a))\). Note that in a Boolean algebra
  $p \le q $ is equivalent to $p \land \neg q \le \bot$. Hence, we have
  \[a \land \neg \neg h(\neg h(a)) = a \land h(\neg h(a)) \le h(h(a) \land \neg h(a)) \le h(\bot) = \bot, \]
  where we use~(2) instantiated with \(b = \neg h(a)\) in the middle inequality. So \(B\) satisfies~(B).
  For the converse, suppose that \(B\) satisfies \(\forall a \colon a \le \neg h(\neg
  h(a))\). We show that it satisfies \(\forall a, b \colon a \land h(b) \le h(h(a) \land b)\):
  \begin{align*}
    \,& (a \land h(b)) \land \neg h(h(a) \land b) \le \neg h(\neg h(a)) \land h(b) \land \neg h(h(a) \land b) \\
    = \,& h(b) \land \neg (h(\neg h(a)) \lor h(h(a) \land b)) = h(b) \land \neg h(\neg h(a) \lor (h(a) \land b)) \\
    = \,& h(b) \land \neg h(\neg h(a) \lor b) = h(b) \land \neg h(\neg h(a)) \land \neg h(b) = \bot,
  \end{align*}
  where we used (B) in the first step. Thus, \(B\) satisfies (2).

  Our phrasings of (B) and (D) as (2) and (3) have the clear benefit of not using negation, and therefore being applicable to all (and not just Boolean) bounded distributive lattices, which are models of \emph{positive} modal logic. As for the phrasing of transitivity as (4), we simply found the symmetry to the euclidean property appealing.
\end{rem}

\autoref{prop:operator-props} allows us, for example, to generalize a classic result of algebraic (Boolean) modal logic due to Halmos~\cite{halmos-58} to the positive setting:
an operator \(h \colon UA \rightarrow UA\) is a \emph{quantifier} if it satisfies the inequations from (1), (3) and (5):

\begin{cor}[Halmos]
An operator on a bounded distributive lattice algebra is a quantifier if and only if its dual relation is an equivalence relation.
\end{cor}

\section{Monoids, Comonoids and Residuation Algebras}
\label{sec:residuation-algebras}
In this section we investigate \emph{residuation algebras}, as introduced by Gehrke~\cite{gehrke-16}, which are ordered structures with residual operators similar to language derivatives.
After recalling some foundations, we divide our study into two steps:
First, we start with the simpler case of \emph{complete} ordered structures, for which we prove a duality between certain complete residuation algebras and ordered monoids.
This result will then serve as the foundation for both major applications in Sections~\ref{sec:duality-theory} and~\ref{sec:duality-cat}: the duality for complete structures in particular restricts to finite structures, since finite lattices are complete, which can then be extended to a duality for more general structures by forming appropriate completions.
On the other hand, the discrete duality for monoids is also the basis for the duality of the category of all categories.

\subsection{Foundations: Discrete Duality}
\label{sec:adjunct-distr-latt}

We recall some facts and notation for distributive lattices and their complete counterparts. In particular, we describe the instantiation of the following discrete version of the extended duality setting from Diagram~\eqref{eq:asm-gold}:
\begin{equation}\label{eq:asm-cat}
    \begin{tikzcd}[sep = large]
      \cat D
      \rar[phantom]{\simeqop}
      \dar[bend right=30pt, swap]{F}
      &
      \hatD
      \ar[bend right=30pt]{d}[swap,pos=.55]{\hatF}
      \dar[bend left=30pt, leftarrow, pos=.55]{\hatU}
      \\
      \cat C
      \rar[phantom]{\simeqop}
      \uar[bend right=30pt, swap]{U}
      \uar[phantom]{\dashv}
      &
      \hatC
      \rar[loop right]{T}
      \uar[phantom]{\vdash}
    \end{tikzcd}
\qquad=\qquad\;\;\;
    \begin{tikzcd}[sep = large, row sep=3em]
      \cat{CSL_{\bigwedge}}
      \rar[phantom]{\simeqop}
      \dar[bend right=30pt, swap]{\mathcal{D}}
      \dar[bend left=30pt,leftarrow]{\Vm}
      &
      \cat{CSL_{\bigvee}}
      \ar[bend right=30pt, swap]{d}{|-|}
      \\
      \cat{ACDL}
      \rar[phantom]{\simeqop}
      \uar[phantom]{\dashv}
      &
      \pos
      \rar[loop right]{\mathcal{D}}
      \uar[phantom]{\vdash}
      \ar[bend right=30pt, swap]{u}{\mathcal{D}}
    \end{tikzcd}
  \end{equation}

\paragraph{(Complete) Semilattices}
We denote the  category of meet-semilattices with a top element, which is isomorphic to \jsl, by
\(\msl\). It has a monoidal structure
given by tensor product \(\msor\) of meet-semilattices.
The category \msl is dual to the category of Stone
meet-semilattices~\cite{hofmann-mislove-stralka-74}. Henceforth, we
denote the forgetful functors from \dl to \(\jsl\) and \(\msl\) by
\(\Uj\) and \(\Um\), respectively, to avoid ambiguity.
  The monad on \priest induced by the dual of \(F_{\land} \dashv \Um\)
  maps a Priestley space \(X\) to its hyperspace \({\uviet} X\) of closed
  \emph{upsets}~\cite{bezhanishvili-harding-morandi-23}.  The respective comonads on \dl
  for the adjunctions \(F_{\land} \dashv \Um\) and
  \(F_{\lor} \dashv \Uj\) are not isomorphic but
  \emph{conjugate}:
  \(F_{\land}\Um \cong (F_{\lor}\Uj(-)^{\partial})^{\partial}\), where \(X^{\partial}\) denotes the \emph{order-dual} of \(X\).
  Their restrictions to the category of Boolean algebras are
  isomorphic since their dual monads on \stone satisfy
  \(\dviet = \viet = \uviet\).

Similarly, the category \cslj (\cslm) consists of complete join- (meet-) semilattices with morphisms preserving all joins (meets).
Note that every complete join- or meet-semilattice~\(X\) also has all meets and joins, respectively, given by
\[\bigwedge A = \bigvee \{x \mid \forall a \in A \colon x \le a\}, \qquad \bigvee A = \bigwedge \{x \mid \forall a \in A \colon x \ge a\}.\]
The categories $\cslj$ and $\cslm$ are easily seen to be dual to each other by swapping joins for meets and taking right, respectively left adjoints of morphisms. This duality is often stated equivalently as \(\cslj \simeqop \cslj\), if the duality also reverses the order on objects.
Analogously to its finitary counterpart, \(\cslj\) has a monoidal structure given by the complete tensor product \(\jsorc\) representing $\bigvee$-bilinear maps, making the left adjoints \(\mathcal{D} \colon (\pos, \times) \rightarrow (\cslj, \jsorc)\) and \(\pow \colon (\set, \times ) \rightarrow (\cslj, \jsorc)\) strong monoidal (the anlogous structure for \cslm is denoted \(\msorc\)).
Analogously to the finitary setting, the tensor product of complete join-semilattices \(J, J'\) is a quotient of \(\mathcal{D}(|J| \times |J'|)\), that is, it is presented by generators \((j, j') \in |J| \times |J'\), which we also write as \(j \jsor j'\), modulo the equations \[(\bigvee_{j \in A}) \jsor (\bigvee_{j' \in A'}) = \bigvee_{j \in A, j' \in A'} j \jsor j' \qquad \text{ for all } A \subseteq J, A' \subseteq J'. \]

The monad on \pos induced by this adjunction is the \emph{downset monad} \(\mathcal{D}\). It
maps a set $X$ to the set of all downward closed subsets of $X$. Its unit and multiplication are given by \({\downarrow}(-) \colon X \rightarrow \mathcal{D} X\) and union, respectively, and its \(\times\)-monoidal structure simply takes products of downsets:
  \[
    \hat{\delta} \colon \mathcal{D} X \times \mathcal{D} Y \rightarrow \mathcal{D} (X \times Y), \quad (A, B ) \mapsto A \times B.
  \]
  The Kleisli category of $\mathcal{D}$ is the category \ordrel of posets with \emph{order
    relations} as morphisms, that is, those relations $R \subseteq X \times Y$ between posets satisfying \(x' \ge x R y \ge y' \Longrightarrow x' R y'\).
  We denote the lifting of the cartesian structure of \pos to \ordrel by \(\bar{\times}\):
  on objects of \ordrel we have  \(X \bar{\times} Y = X \times Y \), and the tensor \(r
  \bar{\times} r'\) of order relations \(r \colon X \rightarrow \mathcal{D} Y\) and \(r' \colon X' \rightarrow \mathcal{D} Y'\) is defined as
  \[ \hat{\delta} \cdot (r \times r') \colon X \times X' \rightarrow \mathcal{D} Y \times \mathcal{D} Y' \rightarrow \mathcal{D} (Y \times Y'), \quad (x, x') \mapsto \{(y, y') \mid y \in r(x), y' \in r(y)\}.\]

  The order-discrete restriction \(\pow \colon \set \dashv \cslj \coc |-|\) induces the classical product monoidal structure on the powerset monad, and its Kleisli category \(\set_{\pow}\) is well-known to be the category of sets and relations.

\paragraph{Completely Distributive Lattices.}
We denote the completely join-prime elements of an ACDL \(D\) by \(\mathcal{J} D \cong \acdl(D, 2)\) and the atoms of a CABA \(B\) by \(\mathcal{A} B \cong \caba(B, 2)\).
The left adjoint  to the forgetful functor  \(\Vm \colon \acdl \rightarrow \cslm\) takes downsets, i.e.\ \(\mathcal{D} \dashv \Vm\), and its unit \(M \rightarrow \mathcal{D} \Vm\) is given by
\[ \eta_{M} \colon M \rightarrow \mathcal{D} M,\quad x \mapsto \{y \mid y \le x\}.\]
Note that the left adjoint of the forgetful functor \(\caba \rightarrow \cslm\) is given by \(\pow\), with the same unit \(\eta_{M}\).
If we respectively equip \acdl and \cslm with coproduct \(+\) and tensor product \(\msorc\) of complete meet-semilattices as monoidal structures, then the forgetful functor \((\acdl, +) \rightarrow (\cslm, \msorc)\) is strong monoidal, analogously to the forgetful functor \((\dl, +) \rightarrow (\cslm, \msor)\); the same holds for the forgetful functor \((\caba, +) \rightarrow (\msl, \msorc)\).
A straightforward verification shows that the adjunctions \(\mathcal{D} \dashv \Vm \) and \(\mathcal{D} \dashv |-|\) are dual with respect to the dualities \(\acdl \simeqop \pos\) and \(\cslm \simeqop \cslj\).
For detailed verifications on the order-discrete setting we refer the reader to Bezhanishvili~et al.~\cite{bezhanishvili-carai-morandi-22}.

\begin{rem}
  Restricting to finite carriers, the diagrams~\eqref{eq:asm-gold} and~\eqref{eq:asm-cat} coincide, since
  \[\acdl_{\mathrm{f}} = \dl_{\mathrm{f}}, \qquad \jsl_{\mathrm{f}} \cong \cat{CSL}_{\bigvee, {\mathrm{f}}}, \qquad \priest_{\mathrm{f}} \cong \pos_{\mathrm{f}},\]
  and so the results we establish in the following for complete structures will in particular hold for finite ones.
\end{rem}

\begin{nota}
We tacitly omit the
forgetful functors \(\Um\) and \(\Uj\) for notational brevity, whenever they are
clear from the context, and just write the join- and meet-semilattice tensor products
of the underlying
semilattices of distributive lattices lattices \(D, D'\) as \(D \jsor D'\) and \(D \msor D'\), respectively.
The same holds for the complete versions: we omit the functors \(\Vm, \Vj\)
and write \(D \msorc D'\) and \(D \jsorc D'\) for the complete tensor products of ACDLs \(D\)
and \(D'\).
\end{nota}

\begin{rem}[Adjunctions on Lattices]\label{rem:adjunctions-lattices}
  \sloppypar
  It is well known that a
  monotone map \(f \colon D \rightarrow D'\) between complete
  lattices preserves all joins if and only if it has a right adjoint
  \(f_{*} \colon D' \rightarrow D\), which is then given by
  \(f_{*}(d') = \V_{f(d) \le d'} d \); dually, it preserves all meets
  iff it has a left adjoint \(f^{*} \colon D' \rightarrow D\), given
  by \(f^{*}(d') = \A_{d' \le f(d)} d \). The
  join-primes \(\mathcal{J} D\) of a bounded distributive lattice \(D\) are
  precisely those elements $p\in D$ whose characteristic function
  \(\chi_{p} \colon D \rightarrow 2\) (mapping \(x \in D\) to \(1\)
  iff \(p \le x\)) is a homomorphism.  The left adjoint of
  \(\chi_{p}\), denoted \(p \colon 2 \rightarrow D\), is the join-semilattice morphism
  defined by \(1 \mapsto p\).
\end{rem}

\begin{lem}\label{lem:tensor}\hfill
  \begin{enumerate}
    \item\label{lem:tensor:iso} The join- and meet-semilattice tensor products of bounded distributive lattices \(D, E\) yield isomorphic lattices, i.e.\ there exists an isomorphism
          \[\omega \colon \Uj D \jsor \Uj E \xra{\cong} \Um D \msor \Um E\]
          satisfying \(\omega(d \jsor 1) = d \msor 0 \) and \(\omega(1 \jsor e) = 0 \msor e\).
    \item \label{lem:tensor:adj} Adjunctions on bounded distributive lattices `compose horizontally':
          Given adjunctions \( f \colon D \dashv E \coc g\) and \(f' \colon D' \dashv E' \coc g'\) between distributive lattices we get adjunctions:\\
            \begin{center}
              \begin{tikzcd}[sep=19]
                E \msor E'
                \rar[yshift=4pt]{g \msor g'}
                \rar[phantom]{\rotatebox{90}{\(\dashv\)}}
                &
                D \msor D'
                \dar{\mtoj}
                \lar[yshift=-4pt]{}
                \\
                E \jsor E'
                \uar{\jtom}
                &
                D \jsor D'
                \lar{f \jsor f'}
              \end{tikzcd}
              \hspace{-5pt}
              \begin{tikzcd}[sep=19]
                E \msor E'
                \rar{g \msor g'}
                \drar[yshift=3pt,xshift=3pt]
                \drar[phantom]{\rotatebox{62}{\(\dashv\)}}
                &
                D \msor D'
                \dar{\mtoj}
                \\
                E \jsor E'
                \uar{\jtom}
                &
                D \jsor D'
                \lar{f \jsor f'}
                \ular[yshift=-3pt,xshift=-3pt]{}
              \end{tikzcd}
              \hspace{-5pt}
              \begin{tikzcd}[sep=19]
                E \msor E'
                \rar{g \msor g'}
                &
                D \msor D'
                \dar{\mtoj}
                \dlar[yshift=-3pt,xshift=3pt]{}
                \\
                E \jsor E'
                \uar{\jtom}
                \urar[phantom]{\rotatebox{118}{\(\dashv\)}}
                \urar[yshift=3pt,xshift=-3pt]
                &
                D \jsor D'
                \lar{f \jsor f'}
              \end{tikzcd}
              \hspace{-5pt}
              \begin{tikzcd}[sep=19]
                E \msor E'
                \rar{g \msor g'}
                &
                D \msor D'
                \dar{\mtoj}
                \\
                E \jsor E'
                \uar{\jtom}
                \rar[phantom]{\rotatebox{90}{\(\dashv\)}}
                \rar[yshift=4pt]{}
                &
                D \jsor D'
                \lar[yshift=-4pt]{f \jsor f'}
              \end{tikzcd}
            \end{center}

          \noindent If the right adjoints $g$ and $g'$ preserve finite joins, then this simplifies to
          \[f \jsor f' \dashv g \jsor g' = \omega^{-1} (g \msor g') \omega.\]
          Dually, \(\omega (f \jsor f') \omega^{-1} = f \msor f'\), if $f$ and $f'$ preserve finite meets.
  \end{enumerate}
\end{lem}
\begin{proof}\hfill
  \begin{enumerate}
    \item We have already seen that for bounded distributive lattices \(D, E\) their \jsl tensor product \(\Uj D \jsor \Uj E\) is their coproduct in \dl.
          But by order-duality, the meet-semilattice tensor product \(\Um D \msor \Um E\) also gives a representation of the coproduct \(D + E\) in \dl:
          its inclusions \(\hat{\iota}_{1}, \hat{\iota}_{2}\) map \(d \in D, e \in E\) to \(\hat{\iota}_{1}(d) = d \msor 0\) and \(\hat{\iota}_{2}(e) = 0 \msor e\), respectively.
          By the universal property of the coproduct we obtain a unique isomorphism
          \[\jtom \colon \Uj D \jsor \Uj E \rightarrow \Um D \msor \Um E \qquad \text{such that} \qquad \text{$\jtom \cdot \iota_{i} = \hat{\iota}_{i}$ for $i=1,2$.}\]
          We now show that $\omega$ is given by the following concrete formula:
          \[\V_{i \in I} d_{i} \jsor e_{i} \mapsto \A_{A \in \pow I} \big(\V_{i \in A} d_{i}\big) \msor \big(\V_{i \not \in A} e_{i}\big). \]
          By definition, the canonical isomorphism \(\jtom\) is the coparing \(\jtom = [\hat{\iota}_{1}, \hat{\iota}_{2}]\) of the inclusions of the meet-semilattice tensor product.
          Therefore on pure tensors \(\jtom\) maps \(d \jsor e \mapsto d \msor 0 \land 0 \msor e\), which extends to general elements of \(D \jsor E\) via distributivity as
          \begin{align*}
            \jtom\big(\V_{i \in I} d_{i} \jsor e_{i}\big) &= \V_{i \in I} \jtom(d_{i} \jsor e_{i}) \\
                                                &= \V_{i \in I} d_{i} \msor 0 \land 0 \msor e_{i} \\
                                                &= \A_{A \in \pow I} \V_{i \in A} d_{i} \msor 0 \lor \V_{i \not \in A} 0 \msor e_{i} \\
                                                &= \A_{A \in \pow I} \big(\V_{i \in A} d_{i}\big) \msor 0 \lor 0 \msor \big(\V_{i \not \in A} e_{i}\big) \\
                                                &= \A_{A \in \pow I} \big(\V_{i \in A} d_{i}\big) \msor \big(\V_{i \not \in A} e_{i}\big),
          \end{align*}
          where we use in the last two steps that joins in \(D \msor E\) satisfy the equation
          \[(a \msor b) \lor (c \msor d) = (a \lor c) \msor (b \lor d).\]
          Note that by order-duality the inverse \(\omega^{-1}\) is given by
          \[\A_{i \in I} d_{i} \msor e_{i} \mapsto \V_{A \in \pow I} \big(\A_{i \in A} d_{i}\big) \jsor \big(\A_{i \not \in A} e_{i}\big). \]
    \item We only need to prove that one of the squares is an adjunction, since we obtain all others by suitable composition with  \(\jtom\) and its inverse. We show this for the third diagram, that is, we show that there is an adjunction
          \[(f \jsor f') \cdot \jtom^{-1} \dashv (g \msor g') \cdot \jtom\]
          by verifying the unit and counit inequalities
          \[ \id\leq (g \msor g') \cdot \jtom \cdot (f \jsor f') \cdot \mtoj \qquad\text{and}\qquad  (f \jsor f') \cdot \mtoj \cdot (g \msor g') \cdot \jtom \leq \id.  \]
          We only prove the counit inequality; the proof of the unit inequality is dual. Recall that the right adjoint $g$ preserves meets.
          Given a finite index set \(A\) we write \(x_{A} = \V_{i \in A} x_{i}\) and compute for every element \(\bigvee_{i} x_{i} \jsor y_{i} \in D \jsor E\)
          \begin{align*}
            &  (f \jsor f') \mtoj (g \msor g') \jtom (\V_{i} x_{i} \jsor y_{i}) \\
            =~& (f \jsor f') \mtoj (g \msor g')(\A_{A \in \pow I } x_{A} \msor y_{A^{c}} ) && \text{def.\ } \jtom \\
            =~& (f \jsor f') \mtoj (\A_{A \in \pow I} g(x_{A}) \msor g'(y_{A^{c}}) ) && \text{def.\ } g \msor g' \\
            =~& (f \jsor f')  (\V_{B \in \pow \pow I} (\A_{A \in B} g(x_{A})) \jsor (\A_{A \in B^{c}} g'(y_{A^{c}})) ) && \text{def.\ } \mtoj  \\
            =~& (f \jsor f')  (\V_{B \in \pow \pow I} g(\A_{A \in B} x_{A}) \jsor g'(\A_{A \in B^{c}} y_{A^{c}}) ) && g \text{ preserves meets} \\
            =~&  (\V_{B \in \pow \pow I} fg(\A_{A \in B} x_{A}) \jsor f'g'(\A_{A \in B^{c}} y_{A^{c}}) ) && \text{def.\ } f \jsor f' \\
            \le~&  (\V_{B \in \pow \pow I} \id(\A_{A \in B} x_{A}) \jsor \id(\A_{A \in B^{c}} y_{A^{c}}) ) && \text{counits } f \dashv g, f' \dashv g' \\
            =~& \mtoj \jtom (\V_{i} x_{i} \jsor y_{i}) = \V_{i} x_{i} \jsor y_{i}.
          \end{align*}

          As for the last statement, if $g$ and $g'$ preserve finite joins, then \(g \jsor g'\) is defined (otherwise it would not be!), and it is clear that \(f \jsor f' \dashv g \jsor g'\).
          By uniqueness of adjoints this implies \(g \jsor g' =
          \omega^{-1} (g \msor g') \omega\).\qedhere
  \end{enumerate}
\end{proof}

\begin{rem}
  \autoref{lem:tensor} holds analogously for complete tensor products of ACDLs: there exists a unique isomorphism \(\omega \colon \Vj D \jsorc \Vj E \simeq \Vm D \msorc \Vm E\) commuting with the coproduct injections.
\end{rem}


\begin{prop}\label{prop:mon-impl}
Let \(D\) be an ACDL.
  \begin{enumerate}
    \item\label{prop:mon-impl:def}
          For every \(x \in D\) the \cslj-morphism
          \[x \jsor (-) \colon D \rightarrow D \jsorc D, \qquad y \mapsto x \jsor y,\] has a right adjoint
          \begin{align*}
            x \moni (-) \colon D \jsorc D \quad&\rightarrow \quad D \\
            T \quad \quad &\mapsto \quad  \V_{x \jsor y \le T} y
          \end{align*}
          called \emph{tensor implication}.
          It can be extended to a binary function  \[(-) \moni (-) \colon D^{\partial} \msorc (D \jsorc D) \rightarrow D. \]
Analogously,  $(-)\otimes x\colon D \rightarrow D \jsorc D$ has a right adjoint $(-)\imon x\colon D \jsorc D\to D$.
    \item\label{prop:mon-impl:p} If \(p \in \mathcal{J} D\) is completely join-prime then \(p \moni (-)\) is a lattice homomorphism given by
          \[\lambda \cdot (\chi_{p} + \id) \colon D + D \rightarrow 2 + D \cong D, \qquad \V_{i \in I} p_{i} \jsor q_{i} \mapsto \V_{p \le p_{i}} q_{i}.\]
    \item\label{prop:mon-impl:preserve}
    Every adjunction \( l \colon E \dashv D \coc r\) between ACDLs satisfies
          \[x \moni \omega^{-1}(r \msor r)\omega (T) = r(l(x) \moni T)\]
          as well as
          \[l(x \moni T) \le l(x) \moni (l \jsor l)(T) \qq{\text{and}} r(x \moni T) \le r(x) \moni \omega^{-1}(r \msor r)\omega(T), \]
          where the latter inequality holds with equality if \(r\) is order-reflecting.
  \end{enumerate}
\end{prop}
\begin{proof}\hfill
  \begin{enumerate}
    \item The function \( x \jsor (-)\) preserves joins by definition, so its right adjoint \(x \moni (-)\) exists and is given by \(T \mapsto \V_{x \jsor y \le T} y\).
          We can write \(x \jsor (-)\) as
          \[x \jsor (-) = \lambda^{-1} \cdot (x \jsor \id) \colon D \cong 2 \jsor D \rightarrow D \jsor D.\]
          By \iref{lem:tensor}{adj}, this has the right adjoint  \(\lambda \cdot \omega^{-1} \cdot (\chi_{x} \msor \id) \cdot \omega\).
          The extension to a binary function \((-) \moni (-) \colon D^{\partial} \msor (D \jsor D) \rightarrow D \) is an instance of an \emph{adjunction with a parameter} (cf.~\cite[Ch.~IV.7]{maclane-98}).
          
    \item If \(x = p \in \mathcal{J} D\) is join-prime, then \(\chi_{p}\) preserves joins. Hence, \(p \moni (-)\) simplifies to
          \[\lambda^{-1} \cdot (\chi_{p} \jsor \id) \colon D \jsor D \rightarrow 2 \jsor D \cong 2, \quad \V_{i} p_{i} \jsor q_{i} \mapsto \V_{p \le p_{i}} q_{i},\]
          which is an ACDL morphism since \(\lambda^{-}, \chi_{p}\) and \(\id\) are.
    \item Let \(T \in D \jsor D\) and \(x \in E\). Then for all \(y \in E\) we have
          \begin{align*}
            y \le x \moni \omega^{-1}(r \msor r)\omega (T) &\LR x \jsor y \le \omega^{-1}(r \msor r)\omega (T) \\
                                               &\LR l(x) \jsor l(y) \le T \\
                                               &\LR l(y) \le l(x) \moni T \\
                                               &\LR y \le r(l(x) \moni T).
          \end{align*}
          So the first statement follows. Using this we compute
          \[x \moni T \le x \moni \omega^{-1}(r \msor r)\omega (l \jsor l)(T) = r(l(x) \moni (l \jsor l)(T)),\]
          which by adjoint transposition is equivalent to
          \[l(x \moni T) \le l(x) \moni (l \jsor l)(T).\]
          Similarly,
          \[r(x \moni T) \le r(l(r(x)) \moni T) = r(x) \moni \omega^{-1}(r \msor r)\omega(T),\]
          and the first step is an equality if \(l \cdot r = \id\), which is equivalent to \(r\) being order-reflecting.\qedhere
  \end{enumerate}
\end{proof}

\subsection{Residuation Algebras}
We proceed with recalling the definition of residuation algebras~\cite{gehrke-16}, which are distributive lattices equipped with two additional operations to be thought of as abstractions of language derivatives (Example~\ref{ex:res-alg}.\ref{ex:res-alg:reg}).
We furthermore introduce an obvious extension of residuation algebras for the complete setting.

\begin{defi}\label{D:bra}\hfill
  \begin{enumerate}
  \item A \emph{(Boolean) residuation algebra} consists of a (Boolean)
    lattice \(R \in \dl\) equipped with \msl-morphisms
    \(\lres \colon R^{\partial} \msor R \rightarrow R\) and
    \(\rres \colon R \msor R^{\partial} \rightarrow R\), the
    \emph{left} and \emph{right residual}, satisfying the
    \emph{residuation property}:
    \(b \le a \lres c \LR a \le c \rres b\).

  \item A \emph{residuation ACDL} (\emph{CABA}) \(R\) is an ACDL (CABA) whose residuals are complete morphisms \(R^{\partial} \msorc R \rightarrow R \leftarrow R \msorc R^{\partial}\).
  \item A residuation algebra \(R\) is  \emph{associative} if it satisfies
    \[
      x \lres (z \rres y) = (x \lres z) \rres y
      \qquad\text{for all \(x, y, z \in R\)},
    \]
      and it is \emph{(prime-)unital}
  if there exists a (join-prime) element
  \(e \in R\)  which is a \emph{unit}:
  \[
    e \lres z = z = z \rres e.
  \]
\end{enumerate}
\end{defi}

\begin{rem}
 Units are unique: If
  \(e, e'\) are units, then \(e = e' \lres e \) since \(e'\) is a unit,
  so by the residuation property \(e' \le e \rres e = e\).
  Analogously we have \(e \le e'\), so \(e = e'\).
\end{rem}

As indicated above, residuals serve as algebraic generalizations of language
derivatives, but as the following examples indicate they are not limited to this
interpretation.

\begin{exas}\label{ex:res-alg}\hfill
  \begin{enumerate}
  \item\label{ex:res-alg:heyt} Every Heyting algebra is
    an associative residuation algebra with residuals
    \(a \lres c = a \rightarrow c \) and \(c \rres b = b \rightarrow c\).

  \item\label{ex:res-alg:neg} Every Boolean algebra \(B\) is a
    non-associative residuation algebra with \(x \lres 1 = 1\) and \(x \lres z = \neg x\) for \(z \ne 1\). If \(B\) is non-trivial, then it is not prime-unital by the first equality.
    \item\label{ex:res-alg:stone}
          Every continuous binary function \(f \colon X \times X \rightarrow X\) on
    a Stone space \(X\) induces a
    residuation algebra on its dual Boolean algebra \(\hatX\) of clopens: given $A,B,C\in \hatX$, put
    \begin{align*}
      A \lres C &= \{x\in X \mid \forall a \in A\colon 
      f(a, x) \in C\},\\
C \rres B &= \{x\in X \mid \forall b \in B\colon f(x, b) \in C\}.
    \end{align*}     

  \item\label{ex:res-alg:reg} The set \(\reg \Sigma\) of all regular
    languages over a finite alphabet \(\Sigma\) forms an associative
    Boolean residuation algebra with residuals given by left and right
    \emph{extended derivatives}:
    \[K \lres L = \{v \in \Sigma^{*} \mid Kv \subseteq L\}, \qquad L \rres K = \{v \in \Sigma^{*} \mid vK \subseteq L\}.\]
          The unit is the singleton \(\{\varepsilon\}\), where $\varepsilon$ is the empty word.
          This example is a special case
    of~\ref{ex:res-alg:stone}: take as $(X,f)$
    the \emph{free profinite monoid} on $\Sigma$, which is the Stone dual of \(\reg \Sigma\).
  \end{enumerate}
\end{exas}

We now introduce the notion of a \emph{residuation
  morphism} between residuation algebras and also its \emph{relational} generalization.

\begin{defi}\label{D:mor}\hfill
  \begin{enumerate}
  \item A lattice morphism \(f \colon R \rightarrow S\) between prime-unital
    residuation algebras is a \emph{(pure) residuation morphism} if it
    satisfies the conditions
    \begin{align}
      & \tag{Forth}
      \label{eq:forth}
      \forall x, z \in R \colon f(x \lres z) \le f(x) \lres f(z) \text{ and } f(z \rres x) \le f(z) \rres f(x)  \\
      & \tag{Back}
      \label{eq:back}
      \forall(y,z) \in S \times R\colon \exists x_{y,z} \in R\colon y \le f(x_{y,z}) \text{ and }  y \lres f(z) = f(x_{y,z} \lres z)\\
      &\tag{Back'} \forall(y,z) \in S \times R\colon \exists x_{y,z} \in R\colon y \le f(x_{y,z}) \text{ and }  f(z) \rres y = f(z \rres x_{y, z})\\
& \tag{Unit}
      \label{eq:unit}
      \forall x \in R \colon e \le x \Leftrightarrow e' \le f(x)
    \end{align}
where $e$ and $e'$ are the units of $R$ and $S$. The morphism \(f\) is \emph{open} if, additionally, it has a left adjoint.
    Prime-unital residuation algebras and residuation morphisms
    form a category \(\res\).

  \item A \emph{corelational residuation morphism} from a
    prime-unital residuation algebra \(R\) to a prime-unital residuation algebra \(S\) is a
    morphism \(\rho \in \jsl_{1}(R, S)\) satisfying
    \[\rho(x \lres z) \le \rho(x) \lres \rho(z) \q{\text{and}} e' \le
      \rho(e).\] Prime-unital residuation algebras with relational
    morphisms  form a category \(\relres\).
  \end{enumerate}
\end{defi}
\pagebreak
\begin{nota}\hfill
  \begin{enumerate}
    \item We use the convention that for a subcategory \(\cat C\) of \(\res\) or \(\relres\) we denote the full subcategory of \(\cat C\) with Boolean carriers by \(\cat{BC}\).
    \item All categories above have obvious counterparts for residuation ACDLs and residuation CABAs with complete morphisms, which we denote by \(\resacdl, \relresacdl\), etc.
  \end{enumerate}
\end{nota}

\begin{rem}
  Let us provide some intuition behind \autoref{D:mor}.

  \begin{enumerate}
    \item The notion of residuation morphism is derived from a result by Gehrke
          \cite[Thm.~3.19]{gehrke-16}, where it is shown to capture precisely
          the conditions satisfied by the duals of morphisms of binary Stone
          algebras.

    \item We speak about \emph{corelational} morphisms of residuation algebras
          since for these will dualize precisely to relational
          morphisms of monoids. Recall that a \emph{relational morphism}
          a monoid \(M\) to a monoid \(N\) is a total relation
          \(\rho \colon M \rightarrow \pow^{+} N\) satisfying
          \begin{equation}
            \label{eq:rel-mor}
            \rho(x)\rho(y) \subseteq \rho(xy) \qquad \text{and} \qquad 1_{N} \in \rho(1_{M}).
          \end{equation}
          Relational morphisms represent inverses of surjective monoid
          homomorphisms~\cite[p.~38]{rhodes-09}.  More precisely, the inverse relation \(h^{-1}\) of a surjective monoid
          homomorphism \(h \colon N \epi M\) is a relational morphism; conversely, if a relational morphism
          \(h^{-1} \colon M \rightarrow \pow^{+}N\) is the inverse of a function
          \(h \colon N \rightarrow M\), then \(h\) is a
          surjective monoid homomorphism.

          Categorically, we can consider an
          inverse relation \(e^{-1} \colon M \rightarrow \pow N \) of a surjective map
          \( e \colon N \epi M\) as is its \emph{right adjoint} in the
          order-enriched category \(\cat{Rel} \simeq \set_{\pow} \) of sets and
          relations: as relations they satisfy \(\id_{N} \le e^{-1} \cdot e\) and
          \(e \cdot e^{-1} \le \id_{M}\). Under duality the composition is reversed,
          so an inverse relation \(e^{-1}\) dualizes to a \emph{left adjoint}
          \(\widehat{e^{-1}} \dashv \widehat{e}\). Since left adjoints between
          lattices are precisely the join-preserving functions, this justifies
          our choice that corelational morphisms of residuation algebras preserve
          (finite) joins (and not necessarily meets). Note also that totality of
          \(e^{-1}\) is equivalent to surjectivity of \(e\), which by
          \autoref{cor:dual-total} dualizes to the property that
          \(\widehat{e^{-1}}\) preserves the top element.

    \item This is also the rationale behind the naming for \emph{open}
          residuation morphisms: if \(e \colon M \epi N\) is a continuous
          surjection between Stone monoids then \(e^{-1} \colon N \rightarrow \viet M\) is
          continuous precisely iff \(e\) is an open map.
  \end{enumerate}
\end{rem}

For open residuation morphisms the three conditions \eqref{eq:back},
\eqref{eq:forth}, \eqref{eq:unit} can be replaced by two equivalent, yet much
simpler conditions. Over complete residuation algebras this is particularly
convenient, since every residuation morphism is open.
\begin{lem}\label{lem:loc-fin-mor-open}
  Let \(R, S\) be prime-unital residuation algebras. A lattice
  morphism \(f \colon R \rightarrow S\) with a left adjoint \(f^{*} \colon S \rightarrow R\) is an open residuation
  morphism iff it satisfies the equations
  \begin{equation}
    \label{eq:open}
   f^{*}(e') = e, \qquad \qquad        \forall y \in S, z
   \in R \colon y \lres f(z) = f(f^{*}(y) \lres z),
   \tag{Open}
 \end{equation}
 and the same equation for the right residual $\rres$.
\end{lem}
\noindent
In the proof below we omit mentioning the right residual because the
arguments for it are completely analogous. This will be the case in
most of the subsequent proofs involving properties of residuals. 
\begin{proof}
  We first show that if \(f\) satisfies \eqref{eq:open}, then it is an open residuation morphism.
  The \eqref{eq:forth} condition follows from \(f^{*} \cdot f \le \id\) and contravariance of $\lres$ in the first argument:
  \[f(x \lres z) \le f(f^{*}(f(x)) \lres z) = f(x) \lres f(z).\]
  For the \eqref{eq:back} condition, given $y \in S$ and $z \in R$ we choose the element \(x_{y,z} = f^{*}(y) \in R\) (independently of \(z\)).
  The unit of the adjunction yields \(y \le f(f^{*}(y)) = f(x_{y,z})\), and using \eqref{eq:open} we obtain
  \[y \lres f(z) = f(f^{*}(y) \lres z) = f(x_{y,z} \lres z).\]

  For the other direction, we prove that every open residuation morphism satisfies the condition \eqref{eq:open}.
  Let \((y, z) \in S \times R\). By the \eqref{eq:back} condition, there exists \(x_{y,z} \in R\) such that \(y \le f(x_{y,z})\) and \(y \lres f(z) = f(x_{y,z} \lres z)\).
  This implies \(f^{*}(y) \le x_{y,z}\), and using \eqref{eq:back} and contravariance of $\lres$ in the first argument, we obtain
  \[y \lres f(z) = f(x_{y,z} \lres z) \le f(f^{*}(y) \lres z).\]
  On the other hand, the adjunction unit \(y \le f(f^{*}(y))\), \eqref{eq:forth} and contravariance of $\lres$ yield
  \[f(f^{*}(y) \lres z) \le f(f^{*}(y)) \lres f(z) \le y \lres f(z).\]
  This proves that \(f\) indeed satisfies \eqref{eq:open}.

  For the respective unitality conditions we have by \(f^{*} \dashv f\) that
  \[\forall x \colon e \le x \Leftrightarrow e' \le f(x) \Leftrightarrow f^{*}(e') \le x, \]
  which is equivalent to \(e = f^{*}(e')\).
\end{proof}

%

\begin{exa}
  Let \(\Sigma\) and \(\Delta\) be finite alphabets. Every substitution
  \(f_{0} \colon \Sigma \rightarrow \Delta^{*}\) can be extended to a
  monoid homomorphism \(f \colon \Sigma^{*} \rightarrow \Delta^{*}\),
  and for regular languages \(L \in \reg \Sigma\) and \(K \in \reg \Delta\), both \(f[L]\) and \(f^{-1}[K]\) are also regular.
  Then
  \(f^{-1} \colon \reg \Delta \rightarrow \reg \Sigma\) is an open
  residuation morphism. Indeed, its left adjoint is given by the direct image map \(f[-] \colon \reg \Sigma \rightarrow \reg \Delta\), satisfying \(f[\{\varepsilon\}] = \{f(\varepsilon)\} = \{\varepsilon\}\) and
  \[K \lres f^{-1}[L] = \{w \mid K w \subseteq f^{-1}[L]\} = \{w \mid f[K]f(w) \subseteq L\} = f^{-1}(f[K] \lres L).\]
\end{exa}

\subsection{Residuation ACDLs}
\label{sec:finite-resid-algebr}

We start by investigating complete residuation algebras, whose characterization
(\autoref{thm:fin-dual}) in terms of coalgebras forms not only the backbone of  the
classification and duality theory of \emph{locally finite} residuation algebras in
\renewcommand{\subsectionautorefname}{Sections}
\autoref{sec:locally-finite-resid} and~\ref{sec:duality-theory},
\renewcommand{\subsectionautorefname}{Section}
but also of  the duality for the category of small categories in \autoref{sec:duality-cat}.
Concretely, we use the tensor implication operator introduced in the last section to associate a \emph{comultiplication} to the residuals and investigate its properties.

\begin{constr}\label{con:op-translation}\hfill
\begin{enumerate}
\item\label{con:op-translation:1} Every $\Vj$-algebra
  \(\Vj D \jsorc \Vj D \rightarrow \Vj D\)
  on an
  ACDL \(D\) has a right adjoint
  \(\gamma \colon \Vm D \rightarrow \Vm(D \jsorc D) \) that can be extended, by
  using the isomorphism \(\omega\) from \autoref{lem:tensor},
  to a \(\Vm\)-coalgebra
  \[ \bar{\gamma} = \Vm \omega \cdot \gamma \colon \Vm D \rightarrow \Vm(D \jsorc D) \cong \Vm(D \msorc D) = \Vm D \msorc \Vm D.  \]
  We refer to both versions \(\gamma\) and \(\bar{\gamma}\) as \emph{comultiplication} or \emph{coalgebra structure}.
  Conversely, we obtain a \(\Vj\)-algebra from a comultiplication \(\gamma \colon \Vm D \rightarrow \Vm(D \jsorc D)\) by taking its left adjoint.
\item  In a residuation ACDL \(R\) the partially applied
  residuals have respective left adjoints
  \(\mu(x,-) \dashv (x \lres -) \) and \( \mu(-, y) \dashv (- \rres y)\) that
  can be combined
  into a \(\Vj\)-algebra structure
  \(\mu \colon \Vj R \jsorc \Vj R \rightarrow \Vj R\)
  that we call \emph{multiplication}. By part \ref{con:op-translation:1}, the multiplication $\mu$ induces the comultiplication \(\gamma \colon \Vm R \rightarrow \Vm(R \jsorc R) \), or $\bar{\gamma} \colon \Vm R \rightarrow \Vm R \msorc \Vm R$.
\end{enumerate}

\end{constr}

  Each of the operators $\rres$, $\lres$, $\mu$, $\gamma$ determines the others uniquely due to the equivalences
  \[x \le z \rres y \iff y \le x \lres z \iff \mu(x \jsor y) \le z \iff x \jsor y \le \gamma(z). \]
  The following lemma provides the concrete formulas.


\begin{lem}\label{lem:gamma-res}
  Let \(R\) be a residuation ACDL.%
  \begin{enumerate}
  \item\label{lem:gamma-res:1} The residuals can be calculated from the comultiplication: 
          \[x \lres z = x \moni \gamma(z) \quad \text{ and } \quad z \rres y = \gamma(z) \imon y,\]
where $\moni$ and $\imon$ are the tensor implications given by \autoref{prop:mon-impl}.
    \item\label{lem:gamma-res:2} The comultiplication can be calculated from the residuals:
          \[\gamma(z) = \V_{x \in R} x \jsor (x \lres z) = \V_{p \in \mathcal{J}R} p \jsor (p \lres z). \]
  \end{enumerate}
\end{lem}
\begin{proof}\hfill
  \begin{enumerate}
    \item For all \(x, y, z \in R\) we have
          \[y \le x \lres z \LR \mu(x \jsor y) \le z \LR x \jsor y \le \gamma(z) \LR y \le x \moni \gamma(z),\]
          and analogously \(x \le z \rres y = \gamma(z) \imon y\).
    \item We compute
          \begin{align*}
            \gamma(z) &= \V \{\V_{i} x_{i} \jsor y_{i} \mid x_{i}, y_{i} \in R, \mu(\V_{i} x_{i} \jsor y_{i}) \le z\} && \text{formula for right adjoint} \\
                 &= \V \{x \jsor y \mid x, y \in R,\mu(x \jsor y) \le z\} && \text{\(\mu\) preserves joins} \\
                 &= \V \{x \jsor y \mid x, y \in R,y \le x \lres z\} && \mu(x \jsor -) \dashv (x \lres -) \\
                 &= \V_{x \in R} x \jsor x \lres z  && \text{simplification.} \\
          \end{align*}
          It is clear that \(\V_{p \in \mathcal{J} R} p \jsor p \lres z \le \V_{x \in R} x \jsor x \lres z\),
          for the reverse inclusion we compute
          \[x \jsor x \lres z = (\V_{p \le x} p) \jsor x \lres z = \V_{p \le x} p \jsor x \lres z \le \V_{p \le x} p \jsor p \lres z,\]
          where $p$ ranges over $\mathcal{J} R$ and we use contravariance of \((- \lres z)\) in the last step. \qedhere
  \end{enumerate}
\end{proof}

We now investigate when the comultiplication is \emph{pure},
that is, lifts to a complete lattice morphism \(R \rightarrow R + R\) and thus corresponds to a pure multiplication.

\begin{lem}\label{lem:fin-op-equiv}
 For a residuation ACDL \(R\), the following are equivalent:
  \begin{enumerate}
    \item\label{lem:fin-op-equiv:gamma} The comultiplication is pure: for all \(A \subseteq R\) we have
      \(\gamma(\bigvee_{x \in A} x) = \bigvee_{a \in A}\gamma(x)\).
      
    \item\label{lem:fin-op-equiv:res} For all $p \in \mathcal{J} R$ and \(A \subseteq R\)
      we have
      \[
        p \lres (\bigvee_{x \in A} x) = \bigvee_{x \in A} p \lres x
        \qquad\text{and}\qquad
        (\bigvee_{x \in A}) \rres p = \bigvee_{x \in A} x \rres p
      \]
      
    \item\label{lem:fin-op-equiv:mu} For all $x, y \in R \colon \mu(x \jsor y) = 0 \iff x = 0 \lor y = 0$,
          and  \(\mu[\mathcal{J}(R + R)] \subseteq \mathcal{J} R\).
  \end{enumerate}
\end{lem}
\begin{proof}
 
  \ref{lem:fin-op-equiv:gamma} $\LR$ \ref{lem:fin-op-equiv:mu}. First, we have
  \begin{align*}
    \gamma(0)= 0 &\LR \gamma(0) \le 0 \\
            &\LR \forall T \colon T \le \gamma(0) \Rightarrow T \le 0 \\
            &\LR \forall T = \V_{i} x_{i} \jsor y_{i} \colon \mu(T) = \V_{i} \mu(x_{i} \jsor y_{i}) \le 0 \Rightarrow \forall i \colon x_{i} \jsor y_{i} \le 0\\
            &\LR \forall x, y\colon \mu(x \jsor y) = 0 \Rightarrow x \jsor y = 0 \\
            &\LR \forall x, y \colon \mu(x \jsor y) = 0 \Leftrightarrow x = 0 \lor  y = 0,
  \end{align*}
  where we use in the penultimate equivalence that \(\mu(0) = 0\), and in the last equivalence that \(x \otimes y = 0\) iff \(x = 0\) or \(y = 0\).
  To show that $\gamma$ preserves joins, note that the join-primes of \(R\jsor R\) are given by pure tensors of \(p \jsor q\) of join-primes \(p, q \in \mathcal{J} R\) and that in a distributive lattice an element \(j\) is join-prime iff it is is \emph{join-prime}: for \(A \subseteq R\) if  $j\leq \bigvee A$ then \(j \le x\) for some \(x \in R\).
  Given \(A \subseteq R\), we compute:
  \begin{align*}
    &\gamma(\bigvee_{x \in A} x) = \bigvee_{x \in A} \gamma(x) \\
    \LR & \forall a, b \in \mathcal{J} R : a \jsor b \le \bigvee_{x \in A} \gamma(x)  \Rightarrow a \jsor b \le \bigvee_{x \in A}\gamma(x)  \\
    \LR & \forall a, b \in \mathcal{J} R : a \jsor b \le \gamma(\bigvee_{x \in A} x) \Rightarrow [\exists x \in A \colon a \jsor b \le \gamma(x)] \\
    \LR &\forall a, b \in \mathcal{J} R : \mu(a \jsor b) \le \bigvee_{x \in A} x \Rightarrow [\exists x \in A \colon \mu(a \jsor b) \le x] \\
    \LR & \forall a \jsor b \in \mathcal{J}[R \jsor R] : \mu(a \jsor b) \in \mathcal{J} R.
  \end{align*}

  For the equivalence \ref{lem:fin-op-equiv:gamma} \(\LR\) \ref{lem:fin-op-equiv:res}, we combine \autoref{lem:gamma-res} with the preservation properties of \(x \moni (-)\) from \iref{prop:mon-impl}{def}:
  If \(\gamma\) is pure, then it preserves joins and so does \((p \lres -) = p \moni \gamma(-)\) for \(p \in \mathcal{J} R\);
  and if every \((p \lres -)\) preserves joins, then so does \(\gamma = \V_{p \in \mathcal{J} R}p \jsor p \lres (-)\).
\end{proof}

Next we show how structural identities like (co-)associativity or
unitality translate between $\gamma$, $\mu$ and the residuals. Note that while the
statements are to be expected, the proof is non-trivial due to the
complication introduced by the seemingly innocent isomorphism
\(\omega \colon R \jsorc R \cong R \msorc R\). Recall that a coalgebra
$c\colon  R \to  R \msorc  R$ is
\emph{coassociative} if $(c \msor \id) \cdot c = (\id
\msor c) \cdot c$ and \emph{(prime-)counital} if it is equipped with a (prime) \emph{counit}
\(\eps \in \cslm(R, 2)\) ($\eps \in \acdl(R, 2)$) satisfying \((\eps \msor \id) \cdot c = \id = (\id \msor \eps) \cdot c\).
Diagrammatically, these equations are dual to the well-known monoid equations:
\[
  \begin{tikzcd}[column sep=30]
     C \rar{c} \dar{c} &  C \msorc C \dar{\id \msorc c} \\
     C \msorc  C \rar{c \msorc \id} &  C \msorc  C \msorc  C 
   \end{tikzcd}
   \qquad
   \begin{tikzcd}[column sep=small]
     C \msorc  C \dar{\id \msorc \eps} &  C \lar[swap]{c} \rar{c}  &  C \msorc  C \dar{\eps
       \msorc \id} \\
     C \msorc 2 \rar[phantom]{\cong} &  C \rar[phantom]{\cong} & 2 \msorc  C
   \end{tikzcd}
 \]

\begin{lem}\label{lem:op-props}
  The following are equivalent for a residuation ACDL \(R\):
  \begin{enumerate}
  \item\label{lem:op-props:gamma} The comultiplication on \(R\) is
    coassociative and has a (prime) counit \(\eps\).
    
  \item\label{lem:op-props:res} The residuals are associative and
    \(R\) has  a (prime) unit \(e \in R\).
    
  \item\label{lem:op-props:mu} The multiplication \(\mu\) is
    associative and has a (prime) unit, that is, there exists
    \(e \in R\) satisfying
    \(\mu(e \jsor -) = \id = \mu(- \jsor e)\).
  \end{enumerate}
\end{lem}
\begin{proof}
  For the proof we only use `adjunctional calculus'. The equivalence \ref{lem:op-props:mu} \(\LR\) \ref{lem:op-props:res} follows from uniqueness of adjoints; to see this, we write associativity of \(\mu\) as
  \[\forall x, y \colon \mu(- \jsor y) \cdot \mu(x \jsor -) = \mu(x \jsor -) \cdot \mu(- \jsor y)\]
  and associativity of the residuals as
  \[\forall x, y \colon (x \lres -) \cdot (- \rres y) = (- \rres y) \cdot (x \lres -).\]
  Since the respective left and right sides of the above equalities are adjoint, and adjoints
  are unique, it is clear that one of the equations holds iff the other one does. The unit of
  the residuals is the left adjoint of the comultiplication $\gamma$, which implies the 
  equivalence of the (co-)unit properties.

  The equivalence \ref{lem:op-props:gamma} \(\LR\) \ref{lem:op-props:mu} is shown similarly, but we have to be careful, since \(\mu\) and \(\bar{\gamma}\) are adjoint only up to composition with the isomorphism \(\omega \colon R \jsorc R \rightarrow R \msorc R\).
  By \iref{lem:tensor}{adj}, we have the following diagram of adjunctions:
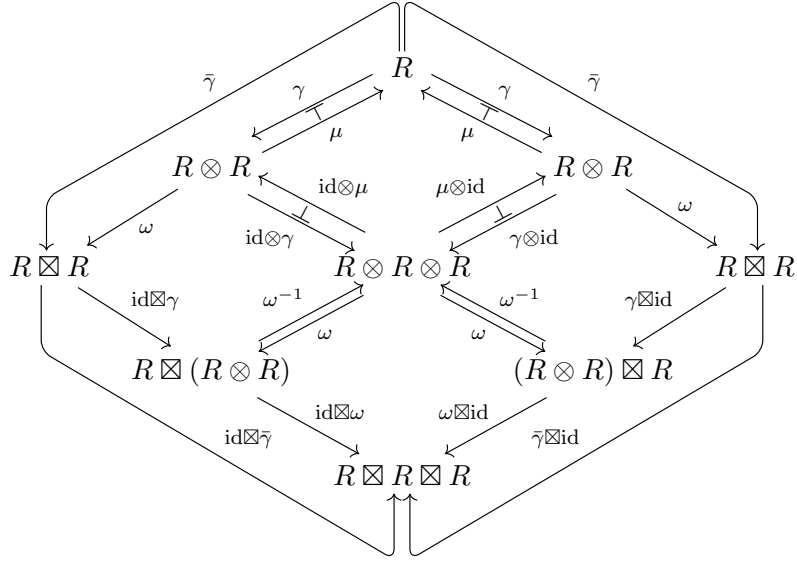
\[
    \begin{tikzcd}[column sep=tiny, row sep=2em]
      & &
      R
      \dlar[yshift=3pt, xshift=-2pt, swap]{\gamma}
      \dlar[phantom]{\rotatebox{117}{\(\dashv\)}}
      \dlar[yshift=-3pt, xshift=2pt, leftarrow]{\mu}
      \drar[yshift=3pt, xshift=2pt]{\gamma}
      \drar[phantom]{\rotatebox{64}{\(\dashv\)}}
      \drar[yshift=-3pt, xshift=-2pt, swap, leftarrow]{\mu}
      \ar[xshift=-1pt, yshift=-13pt, shiftarr={yshift=20pt}, swap]{ddll}{\bar \gamma}
      \ar[xshift=1pt,  yshift=-13pt, shiftarr={yshift=20pt}]{ddrr}{\bar \gamma}
      & &
      \\
      &
      R \jsorc R
      \dlar{\jtom}
      \drar[yshift=-2pt, xshift=-1pt, swap]{\id \jsorc \gamma}
      \drar[yshift=5pt, xshift=3pt, leftarrow]{\id \jsorc \mu}
      \drar[phantom, yshift=2pt]{\rotatebox{242}{\(\dashv\)}}
      & &
      R \jsorc R
      \drar{\jtom}
      \dlar[yshift=-2pt, xshift=1pt]{\gamma \jsorc \id}
      \dlar[yshift=5pt, xshift=-3pt,swap, leftarrow]{\mu \jsorc \id}
      \dlar[phantom, yshift=2pt]{\rotatebox{298}{\(\dashv\)}}
      &
      \\
      R \msorc R
      \drar{\id \msorc \gamma}
      \ar[xshift=-3pt, yshift=19pt, shiftarr={yshift=-25pt}]{ddrr}{\id \msorc \bar \gamma}
      & &
      R \jsorc R \jsorc R
      \dlar[yshift=-2pt]{\jtom}
      \drar[swap, yshift=-2pt]{\jtom}
      & &
      R \msorc R
      \dlar[swap]{\gamma \msorc \id}
      \ar[xshift=3pt, yshift=19pt, shiftarr={yshift=-25pt}, swap]{ddll}{\bar \gamma \msorc \id}
      \\
      &
      R \msorc (R \jsorc R)
      \urar[yshift=2pt]{\mtoj}
      \drar{\id \msorc \jtom}
      & &
      (R \jsorc R) \msorc R
      \ular[swap, yshift=2pt]{\mtoj}
      \dlar[swap]{\jtom \msorc \id}
      &
      \\
      & &
      R \msorc R \msorc R
      & &
    \end{tikzcd}
  \]
  The left and right diamonds come from the horizontal composition of adjunctions under the respective tensor products.
  The bottom diamond is easily seen to commute.
  If \(\mu\) is associative the top inner diamond commutes, and so by uniqueness the outer big diamond commutes by uniqueness of adjoints, proving \(\bar{\gamma}\) coassociative.
  Dually, if \(\bar{\gamma}\) is coassociative, then \(\mu\) is associative.
  The unit of \(\mu\) is the left adjoint of the counit of $\gamma$, so one is prime iff the other one is.
\end{proof}

These lemmas suggest the following definitions:

\begin{defi}\label{def:res-alg}\hfill
  \begin{enumerate}
    \item A residuation ACDL \(R\) is \emph{pure} if it satisfies one of the equivalent conditions of \autoref{lem:fin-op-equiv}.
    \item\label{def:res-alg:der-alg} A \emph{
          derivation ACDL} is a pure residuation ACDL that satisfies the equivalent conditions of \autoref{lem:op-props}. We denote the respective full subcategories by
         \[ \deracdl \hookrightarrow \resacdl  \qquad\text{and}\qquad  \relderacdl \hookrightarrow \relresacdl.  \]
    \item A  \(\Vm\)-coalgebra
      \(\bar{\gamma} \colon \Vm C \rightarrow \Vm C \msorc
      \Vm C\) is a \emph{\(\Vm\)-comonoid} if it is
          coassociative and prime-counital, and a
          \emph{comonoid} if \(\bar{\gamma}\) is pure. We analogously define \emph{\(\Um\)-comonoids} and (pure)  \emph{comonoids} in \dl.
  \end{enumerate}
\end{defi}

In order to extend the correspondence of residuation ACDLs and
coalgebras to a categorical equivalence we introduce appropriate morphisms for coalgebras, which we also define for the general case of \(\Um\)-coalgebras.

\begin{defi}\label{def:fin-res-mor}\hfill
  \begin{enumerate}
  \item\label{def:fin-res-mor:pure} A \emph{pure morphism} from a
    prime-counital \(\Vm\)-coalgebra \((C, \bar{\gamma},\epsilon)\) to
    \((C', \bar{\gamma}',\epsilon')\) is a morphism
    \(f \in \acdl(C, D)\) satisfying
    \((f \msorc f) \cdot \bar{\gamma} = \bar{\gamma}' \cdot f\) and
    \(\epsilon = \epsilon' \cdot f\).
    \[
      \begin{tikzcd}[column sep=40]
        \Vm C
        \rar{\Vm f}
        \dar{\bar{\gamma}}
        &
        \Vm C'
        \dar{\bar{\gamma}'}
        \\
        \Vm C \msorc \Vm C
        \rar{\Vm\! f \,\msorc\, \Vm\! f}
        &
        \Vm C' \msorc \Vm C'
      \end{tikzcd}
      \qquad\qquad
      \begin{tikzcd}
        \Vm C
        \rar{\Vm f}
        \drar[swap]{\Vm \epsilon}
        &
        \Vm C'
        \dar{\Vm \epsilon'}
        \\
        &
        \Vm 2
      \end{tikzcd}
    \]
    The category of prime-counital  \(\Vm\)-coalgebras with pure morphisms
    is denoted by \(\coalg(\Vm)\) and its full
    subcategory  of \(\Vm\)-comonoids by \(\comon(\Vm)\), again with the full subcategory \(\comon\) of comonoids.

  \item\label{def:fin-res-mor:lax} Let \(C\) and \(C'\) be comonoids in \acdl.  A \emph{corelational
      morphism} from \(C\) to \(C'\) is a morphism
    \(\rho \in \cslj(C, C')\) satisfying
    \(\rho(1) = 1\), \((\rho \jsorc \rho) \cdot \gamma \le \gamma' \cdot \rho\)
    and \(\epsilon \le \epsilon' \cdot \rho\),
    that is, the following diagrams in $\cslj$ commute laxly as indicated:
    \[
      \begin{tikzcd}
        \Vj C
        \rar{\rho}
        \dar{\Vj\gamma}
        &
        \Vj C'
        \dar{\Vj \gamma'}
        &
        &
        \Vj C
        \rar{\rho}
        \drar[swap]{\Vj\epsilon}
        &
        \Vj C'
        \dar{\Vj\epsilon'}
        \\
        \Vj C \jsorc \Vj C
        \rar{\rho \jsorc \rho}
        \urar[phantom]{\rotatebox{45}{\(\le\)}}
        &
        \Vj C' \jsorc \Vj C'
        &
        &
        \phantom{X}
        \urar[phantom, near end]{\rotatebox{45}{\(\le\)}}
        &
        \Vj 2
      \end{tikzcd}
    \]
    Comonoids with corelational morphisms form a category \(\relcomon\).
   
    \item Analogously, we define the category \(\coalg(\Um)\) with its subcategories
      \(\comon(\Um)\) and \(\comon\) of (pure) \(\Um\)-comonoids. We also denote by
      \(\relcomon\) the category of comonoids with corelational morphisms. So we overload
      notation; whether we mean comonoids in \(\msl\) or \(\cslm\) will be clear from context.
  \end{enumerate}
\end{defi}

Recall from \autoref{lem:loc-fin-mor-open} that every morphism of residuation ACDLs is open.

\begin{prop}\label{prop:fin-mor}
  Let \(R\) and \(R'\) be prime-unital residuation ACDLs.
  \begin{enumerate}
    \item A  morphism \(f \in \acdl(R, R')\)  is a pure coalgebra morphism iff it is a residuation morphism.
    \item If \(R\) and \( R'\) are comonoids, then a morphism \(\rho \in \cslj(R, R')\) preserving the top element is a corelational comonoid morphism iff it is a  corelational residuation morphism.
  \end{enumerate}
\end{prop}
\pagebreak
\begin{proof}\hfill
  \begin{enumerate}
    \item First, let \(f\) be a pure coalgebra morphism.
          Then
          \begin{align*}
            x \lres f(z) &= x \moni \gamma' f (z) && \text{\autoref{lem:gamma-res}\ref{lem:gamma-res:1}} \\
                         &= x \moni \omega^{-1}\bar{\gamma}' f (z) && \bar{\gamma}' = \omega \cdot \gamma' \\
                         &= x \moni \omega^{-1} (f \msorc f) \omega
                         \gamma (z) && f \text{ is a coalgebra morphism}\\
                         &= f(f^{*}(x) \moni \gamma(z)) && \text{\autoref{prop:mon-impl}\ref{prop:mon-impl:preserve}}\\
            &= f(f^{*}(x) \lres z) && \text{\autoref{lem:gamma-res}},
          \end{align*}
          which by \autoref{lem:loc-fin-mor-open} shows that \(f\) is
          a  residuation morphism.
          
          Conversely, if \(f\) is a  residuation morphism, then for every \(z \in R\) we compute
          \begin{align*}
            \gamma' f (z)       &= \V_{x' \in R'} x' \jsor x' \lres f(z) && \text{\autoref{lem:gamma-res}\ref{lem:gamma-res:2}} \\
                           &= \V_{x' \in R'} x' \jsor f(f^{*}(x') \lres z) && \text{\(f\) residuation morphism} \\
                           &\le \V_{x' \in R'} f f^{*}(x') \jsor f(f^{*}(x') \lres z) && \id \le f \cdot f^{*} \\
                           &= (f \jsorc f)(\V_{x' \in R'} f^{*}(x') \jsor f^{*}(x') \lres z) && f \jsor f \text{ preserves joins} \\
                           &\le (f \jsorc f) (\V_{x \in R} x \jsor x \lres z) && f^{*}[R'] \subseteq R \\
                           &= (f \jsorc f)\gamma (z) && \text{\autoref{lem:gamma-res}\ref{lem:gamma-res:2}}. \\
          \end{align*}
          Now (order-isomorphic) postcomposition with \(\omega\) gives
          \[ \bar{\gamma}' f  = \omega \gamma' f \le \omega (f \jsorc f) \gamma = (f \msorc f) \omega \gamma = (f \msorc f) \bar{\gamma}.\]
          Conversely,                   
          \begin{align*}
            (f \jsorc f) \gamma (z) &= \V_{x \in R} f(x) \jsor f(x \lres z) && \text{\autoref{lem:gamma-res}\ref{lem:gamma-res:2}} \\
                              &\le \V_{x \in R}  f(x) \jsor f(f^{*}f(x) \lres z) && \text{$f^*f\leq\id$ and contravariance} \\
                              &= \V_{x \in R} f(x) \jsor f(x) \lres f(z) && \text{\(f\) is open and \autoref{lem:loc-fin-mor-open}} \\
                              &\le \V_{x' \in R'} x' \jsor x' \lres f(z) && \text{$f[R] \subseteq R'$} \\
                              &= \gamma'(f(z))&& \text{\autoref{lem:gamma-res}\ref{lem:gamma-res:2}}. \\
          \end{align*}
          Postcomposition with \(\omega\) again yields \(\bar{\gamma}' f \ge (f \msorc f) \bar{\gamma}\).
          Hence, we obtain \(\bar{\gamma}' f = (f \msorc f) \bar{\gamma}\), so $f$ is a pure coalgebra morphism. Moreover,
          it is clear that the counit condition from \iref{def:fin-res-mor}{pure} is equivalent to the unit conditions from \eqref{eq:open}, since (1) \(e \dashv \epsilon\) and \(f^{*}(e) \dashv \epsilon' \cdot f \), and (2) adjoints are unique, so either equations holds iff the other one does.%
         
        \item If \(\rho \colon R \rightarrow R'\) is a corelational morphism of pure coalgebras, then
          \begin{align*}
            \rho(x \lres z) &= \rho(x \moni \gamma(z)) && \text{\iref{lem:gamma-res}{1}} \\
                         &\le \rho(x) \moni (\rho \jsorc \rho)(\gamma(z)) && \text{\iref{prop:mon-impl}{preserve}} \\
                         &\le \rho(x) \moni \gamma'(\rho(z)) && \text{\(\rho\) corelational morphism} \\
            &= \rho(x) \lres \rho(z) && \text{\iref{lem:gamma-res}{1}}.
          \end{align*}
          Conversely, if $\rho\colon R\to R'$ is a corelational morphism of residuation algebras, then using \iref{lem:gamma-res}{2} twice, and that $\rho[R] \subseteq R'$, we obtain
          \[(\rho \jsorc \rho)\gamma(z) = \V_{x \in R} \rho(x) \jsor \rho(x \lres z) \le \V_{x \in R}\rho(x) \jsor \rho(x) \lres \rho(z) \le \gamma'(\rho(z)).\]
          For the respective counits we identify the neutral element \(e \in R\) with the \jsl-morphism \(e \colon 2 \rightarrow R\) to compute
          \begin{align*}
            && \epsilon &\le \epsilon' \cdot \rho  && \\
            \LR && \forall x \colon \epsilon(x) &\le \epsilon'(\rho(x)) && \\
            \LR && \forall x, y \colon y \le \epsilon(x) & \Rightarrow y \le \epsilon'(\rho(x)) && \\
            \LR && \forall x, y \colon e(y)\le x & \Rightarrow e'(y) \le \rho(x) && e \dashv \epsilon, e' \dashv \epsilon' \\
            \LR && \forall x \colon e \le x & \Rightarrow e' \le \rho(x) && y \in \{0, 1\} \text{ and } e(0) = e'(0) = 0 \\
            \LR && e' \le \rho(e),
          \end{align*}
          where in the last step we set \(x = e\) for the downward implication, and the upward direction is simply monotonicity of \(\rho\). \qedhere
  \end{enumerate}
\end{proof}

\begin{thm}\label{thm:fin-iso}
  The following categories are isomorphic
  \[\coalg(\Vm) \cong \resacdl, \quad \comon \cong \deracdl, \quad \relcomon \cong \relderacdl. \]
\end{thm}
\begin{proof}
  The three isomorphisms are given on objects
 by swapping between residuals and
  comultiplication (note that the residual unit is a left adjoint of
  the counit of the comultiplication), and act as identity on morphisms.
  \newcommand{\lemmaautorefname}{Lemmas}%
  \autoref{lem:fin-op-equiv} and \ref{lem:op-props}
  \renewcommand{\lemmaautorefname}{Lemma}%
and \autoref{prop:fin-mor} show that they are well-defined.
\end{proof}

From \autoref{thm:fin-iso} we obtain a dual characterization of ordered monoids; it restricts to a duality between ordinary monoids and derivation CABAs.
Recall that a \emph{relational
  morphism} of ordered monoids $M$ and $N$
  is a total order-relation \(\rho \colon M \rightarrow \mathcal{D}^+ N\) (where \(\mathcal{D}\) is the downset monad) making the
  following diagrams commute laxly:
  \begin{equation}
    \label{eq:rel-mon-mor}
    \begin{tikzcd}
      M \times M
      \ar{rr}{\cdot_{M}}
      \dar{\rho \times \rho}
      &
      &
      M
      \dar{\rho}
      &&
      1
      \rar{1_{M}}
      \dar{1_{N}}
      &
      M
      \dar{\rho}
      \\
      \mathcal{D} N \times \mathcal{D} N
      \rar{\hat{\delta}}
      \ar[phantom]{urr}{\rotatebox{45}{\(\le\)}}
      &
      \mathcal{D}(N \times N)
      \rar{\mathcal{D}(\cdot_N)}
      &
      \mathcal{D} N
      &&
      \urar[phantom]{\rotatebox{45}{\(\le\)}}
      N
      \rar{\eta}
      &
      \mathcal{D} N
    \end{tikzcd}
  \end{equation}

\begin{thm}\label{thm:fin-dual}\hfill
  \begin{enumerate}
  \item\label{thm:fin-dual:1} The category of ordered monoids is dually equivalent to the
    category of derivation ACDLs (or
    ACDL-comonoids):
    \[\ordmon \;\simeqop\; \comon \;\cong\; \deracdl. \]

  \item\label{thm:fin-dual:2} The category of ordered monoids and relational morphisms is
    dually equivalent to the category of derivation
    ACDLs (or comonoids) and corelational
    morphisms:
    \[
      \relordmon \;\simeqop\; \relcomon \;\cong\;
      \relderacdl.
    \]
  \end{enumerate}
\end{thm}
\begin{proof}
  Both statements follow by extending the equivalences from \autoref{thm:fin-iso}
  with the extended duality applied to the setting \autoref{eq:asm-cat} established in \autoref{sec:adjunct-distr-latt}: For \autoref{thm:fin-dual:1} we get that  ordered monoids \(M \times M \rightarrow M\) are dual to comonoids \(C \rightarrow C + C\) in \acdl, and for \autoref{thm:fin-dual:2}, this means that relational monoid morphisms \(M \rightarrow \mathcal{D}^{+} N\) as in \ref{eq:rel-mon-mor} dualize precisely to corelational morphisms \(\Vj N \rightarrow \Vj M\) of comonoids.
\end{proof}

\begin{rem}\label{rem:fin-dual}
  Both \autoref{thm:fin-iso} and \autoref{thm:fin-dual} restrict to finite carriers.
  All finite lattices are complete and \(\res_{\mathrm{f}} = \resacdl_{\mathrm{f}}\), so, writing \(\der_{\mathrm{f}}\) for \(\deracdl_{\mathrm{f}}\), we get the equivalences
  \[\ordmon_{\mathrm{f}} \simeqop \comon_{\mathrm{f}} \cong \der_{\mathrm{f}} \quad \text{ and } \quad \relordmon_{\mathrm{f}} \simeqop \relcomon_{\mathrm{f}} \cong \relder_{\mathrm{f}}. \]
\end{rem}

\subsection{Locally Finite Residuation Algebras}
\label{sec:locally-finite-resid}

We now extend the correspondence between residuation algebras and coalgebras from \autoref{sec:finite-resid-algebr} from complete  to non-complete carriers.
The main challenge arises from the reliance on adjoints for the constructions in \autoref{sec:finite-resid-algebr}, whose existence is of course not ensured for arbitrary distributive lattices as carriers.
We tackle this problem by considering \emph{locally finite} structures, allowing us to extend the comultiplication~\autoref{con:op-translation} from finite subalgebras to the whole lattice.

We start with the motivating example from automata theory:
the residuals of regular languages from Example~\ref{ex:res-alg}.\ref{ex:res-alg:reg}.

\begin{exa}\label{ex:reg-coalg}
  It is well known that the Boolean algebra \(\reg \Sigma\) of regular
  languages dualizes under Stone duality to the Stone space $\overline{\Sigma^*}$ of \emph{profinite words}\footnote{This space is commonly denoted $\widehat{\Sigma^*}$ in the literature; we use the notation $\overline{\Sigma^{*}}$ to avoid a clash with notation $\widehat{(-)}$ for the dual equivalence.} (see e.g.~Pippenger~\cite{pippenger-97}). The space $\overline{\Sigma^*}$ can be constructed as the limit in the category of Stone spaces of the diagram of all finite quotient monoids of $\Sigma^*$, regarded as discrete spaces. It is equipped with a continuous monoid structure
  \(\mu \colon \pfm \times \pfm \rightarrow \pfm\)
  extending the concatenation of words.
  We calculate below that its dual comultiplication
  on regular languages under Stone duality is given as follows:
  \begin{equation}\label{eq:comul-reg}
    \begin{aligned}
      \gamma \colon \reg \Sigma \quad &\rightarrow \quad \reg \Sigma + \reg \Sigma \\
      L \quad \quad &\mapsto \quad \V_{[v] \in \mathrm{Syn}_{L}} [v] \jsor [v] \lres
      L,
    \end{aligned}
  \end{equation}
  where \(\mathrm{Syn}_{L}\) is the \emph{syntactic monoid} of
  $L$; its elements are the  equivalence classes of the \emph{syntactic congruence relation} relation on $\Sigma^*$ defined by
  \[ [v] = [w] \qquad \text{iff} \qquad \forall K,K'\subseteq\Sigma^*\colon  v \in K \lres L \rres K' \iff w \in K \lres L \rres K'. \]
  The Stone monoid \(\overline{\Sigma^{*}}\) is \emph{profinite}, that is, it is the (cofiltered) limit of the diagram of its finite continuous monoid quotients.
  Therefore, by duality, \(\reg \Sigma\) is the filtered colimit of its finite sub-coalgebras \(\pow M \cong \hatM \hookrightarrow \widehat{\overline{\Sigma^{*}}} \cong \reg \Sigma\) dual to finite monoid quotients \(\overline{\Sigma^{*}} \epi M\).
  This means that, given a regular language \(L\), we can compute the
  value of \(\gamma(L)\) using $\mathrm{Syn}(L)$ as in the following diagram:
    \[
      \begin{tikzcd}
        \reg \Sigma \rar[description]{\cong}
        \ar[shiftarr={yshift=20pt}]{rrrr}{\gamma}
        &
        \cl(\overline{\Sigma^{*}})
        \ar{r}{\mu^{-1}}
        &
        \cl(\overline{\Sigma^{*}} \times \overline{\Sigma^{*}})
        \rar[description]{\cong}
        &
        \cl(\overline{\Sigma^{*}}) \!+\! \cl(\overline{\Sigma^{*}})
        \rar[description]{\cong}
        &
        \reg \Sigma\! +\! \reg \Sigma
        \\
        &
        \pow \mathrm{Syn}_L
        \uar[hook]{h^{-1}}
        \rar{\mu^{-1}}
        &
        \pow(\mathrm{Syn}_L \times \mathrm{Syn}_L)
        \rar[description]{\cong}
        \uar[hook]{(h \times h)^{-1}}
        &
        \pow \mathrm{Syn}_L \!+\! \pow \mathrm{Syn}_L
        \uar[hook]{h^{-1} + h^{-1}}
        &
      \end{tikzcd}
    \]
  We denote by \([w]\) the syntactic congruence class of \(w\) with respect to \(L\) to compute
\allowdisplaybreaks
  \begin{align*}
    \mu^{-1}(L) &= \mu^{-1}(h^{-1}( h[L] )) && \text{ \(L\) recognized by \(M\)} \\
              &= (h \times h)^{-1}(\mu^{-1}(h[L])) && \text{
                \(h^{-1}\) is a coalgebra homomorphism} \\
              &= (h \times h)^{-1}(\{(m, n) \mid mn \in h[L]\}) &&  \\
              &= (h \times h)^{-1}(\bigcup_{mn \in h[L]} \{(m, n)\} ) &&  \\
              &\mapsto (h^{-1} + h^{-1})(\V_{mn \in h[L]}\{m\} \otimes \{n\}) &&  \pow (M \times M) \cong \pow M + \pow M \\
              &= \V_{mn \in h[L]}h^{-1}[\{m\}] \otimes h^{-1}[\{n\}] &&\\
              &= \V_{h(v)h(w) \in h[L]}h^{-1}h[ v ] \otimes h^{-1}h[ w ] && \text{ \(h\) surjective\ }\\
              &= \V_{vw \in L}[v] \otimes [w] && \text{ syntactic equivalence classes }\\
              &= \V_{v \in \Sigma^*}\V_{w \in v^{-1} L}[v] \otimes [w] && \text{ definition of \(v^{-1}L\)}\\
              &= \V_{v \in \Sigma^*}[v] \otimes (\V_{w \in v^{-1}L} [w]) && (\star)\\
              &= \V_{v \in \Sigma^*}[v] \otimes v^{-1}L && (\diamond) \\
              &= \V_{[v] \in \mathrm{Syn}_{L}}[v] \otimes [v] \lres L. && (\#) \\
  \end{align*}
  Step \((\star)\) uses that \(L\) is regular and so the join
  \(\bigvee_{w \in v^{-1}L} [w]\) is finite; at step \((\diamond)\) we insert the equality \(v^{-1}L = \bigcup_{w \in v^{-1}L}[w]\), which holds by the definition of syntactic congruence; and at step \((\#)\) we use the definition of residuals in \(\reg \Sigma\) as derivatives, and that $L$ is regular.
\end{exa}

The key to extending the duality of \autoref{rem:fin-dual} is the notion of \emph{residuation ideal} of a residuation algebra. It was introduced by Gehrke~\cite{gehrke-16} to
characterize quotients of Priestley topological algebras.
In particular, she has shown~\cite[Thm.~15]{gehrke-16-proc} that
  the Stone monoid quotient \(\mathrm{Syn}_{L}\) of \(\pfm\) dualizes to the {residuation ideal}
  generated by \(L \in \reg \Sigma\)

\begin{defi}
  A \emph{residuation ideal} of a residuation algebra \(R\) is a
  sublattice \(I \incl R\) closed under derivatives w.r.t.\ arbitrary
  elements of $R$:%
  \[
    \forall z \in R, x \in I \colon x \lres z \in I \text{ and } z
    \rres x \in I.
  \]
  In particular, every residuation ideal is a residuation subalgebra
  of $R$.
  We denote the residuation ideal generated by a subset \(X \subseteq R\) by \(\genres X\).
\end{defi}

\autoref{ex:reg-coalg} suggests a path to extending our constructions of \autoref{sec:finite-resid-algebr} from complete to more general residuation algebras: extend them \emph{locally}, that is, by considering suitable finite substructures.
Note that in
the formula~\eqref{eq:comul-reg} for the comultiplication on regular languages it is
crucial that the residuation ideal \(\genres{\{L\}}\) generated
by a single regular language \(L\) is \emph{finite},
so that the join \((\#)\) is defined.

\begin{defi}\label{D:locally-fin}\hfill
  \begin{enumerate}
  \item A residuation algebra \(R\) is \emph{locally finite} if every
    finite subset of \(R\) is contained in a finite residuation ideal
    of \(R\).
    
  \item A \(\Um\)-coalgebra \(C\) is \emph{locally finite} if
    every finite subset of \(C\) is contained in a finite subcoalgebra
          of \(C\).
          The category of locally finite comonoids is denoted \(\comon_{\mathrm{lf}}\).
  \end{enumerate}
\end{defi}

Note that not every residuation algebra is locally finite, consider
for example an infinite Boolean algebra in Example~\ref{ex:res-alg}.\ref{ex:res-alg:neg}.


\begin{prop}\label{prop:loc-fin-equiv}\hfill
  \begin{enumerate}
  \item Every locally finite residuation algebra \(R\) induces a locally
    finite \(\Um\)-coalgebra
    \(\gamma_{\lres} \colon \Um R \rightarrow \Um(R \jsor
    R)\) where
    \[
      \gamma_{\lres}(z) = (\iota_{A} \jsor \iota_{A})(\gamma_{A}(z)) =
      \V_{x \in A} \iota_{A}(x)  \jsor \iota_{A}(x \lres z) = \V_{p
        \in \mathcal{J} A} \iota_{A}(p) \jsor \iota_{A}(p \lres z),
    \]
    for every finite residuation ideal \(\iota_{A} \colon A \incl R\)
    containing \(z\) (here \(\gamma_{A}\) is the comultiplication on~\(A\) as in \autoref{con:op-translation}).  

  \item Every locally finite \(\Um\)-coalgebra \((C, \gamma)\)
    induces a locally finite residuation algebra where
    \[
      x \lres_{\gamma} z
      \;=\;
      \iota_{A}(x \moni \gamma(z) ), \qquad z \rres_{\gamma} x \;=\; \iota_{A}(\gamma(z) \imon x), \] for every finite subcoalgebra
    \(\iota_{A} \colon A \incl C\) containing \(x, z\) (here
    \(\lres_{A}\) is the residual on \(A\) as given by
    \autoref{con:op-translation}).  The residuals have a canonical
    presentation as
    \[ x \lres_{\gamma} z \;=\; \iota_{z}(\iota_{z}^{*}(x) \lres z) \qquad\text{and}\qquad 
      z \rres_{\gamma} x \;=\;  \iota_{z}(z\rres \iota_{z}^{*}(x)) \]
    where \(\iota_{z} \colon \langle z \rangle \rightarrow C\) is the
          smallest (finite) subcoalgebra containing \(z\).
    \item These constructions are mutually inverse: \[
\gamma_{\lres_{\gamma}} \;=\; \gamma \qquad\text{and}\qquad \lres_{\gamma_{\lres}} \;=\; \lres.\]
  \end{enumerate}
\end{prop}
\begin{proof}\hfill
  \begin{enumerate}
    \item[(1a)] We first show that the comultiplication
          \(\gamma_{\lres} = (\iota_{A} \jsor \iota_{A})(\gamma(z))\) is well-defined, that is, it
          does not depend on the residuation ideal \(A\) containing \(z\). We
          first prove an auxiliary statement:
          \begin{lem}
            If \(\iota_I = \iota_{K} \cdot \iota \colon I \incl K \incl R\) are finite
            residuation ideals containing \(z\), then
            \begin{equation}
              \label{eq:gammaiota}
              (\iota_{I} \jsor \iota_{I})(\gamma_{I}(z)) = (\iota_{K} \jsor \iota_{K})(\gamma_{K}(\iota(z))).
            \end{equation}
          \end{lem}
          \begin{proof}
            Since \(I \subseteq K\), it is clear that
            \begin{align*}
              (\iota_{I} \jsor \iota_{I})(\gamma_{I}(z))  &= (\iota_{I} \jsor \iota_{I})(\V_{x \in I} x \jsor x \lres z) \\
                                             &\le (\iota_{K} \jsor \iota_{K})(\V_{x \in K} x \jsor x \lres z) \\
                                             &= (\iota_{K} \jsor \iota_{K})(\gamma_{K}(\iota(z))).
            \end{align*}
            For the reverse inequality, note that for every join-prime
            \(q \in \mathcal{J}K\) we find \(p \in \mathcal{J} I\) such that \(q \le \iota(p)\), since
            \[q \le \top = \iota(\top) = \iota(\V_{p \in \mathcal{J} I} p) = \V_{p \in \mathcal{J} I}\iota(p)\] and \(q\) is
            join-prime. We use this to calculate
            \begin{align*}
              (\iota_{K} \jsor \iota_{K})(\gamma_{K}(\iota(z))) &= \V_{q \in \mathcal{J} K} \iota_{K}(q) \jsor \iota_{K}(q \lres \iota(z) ) \\
                                               &= \V_{p \in \mathcal{J} I} \V_{q \le \iota(p)} \iota_{K}(q) \jsor \iota_{K}(q \lres \iota(z)) && \text{idempotence of join} \\
                                               &= \V_{p \in \mathcal{J} I} \V_{q \le \iota(p)} \iota_{K}(q) \jsor \iota_{I}(\iota^{*}(q) \lres z) && {(*)}  \\
                                               &= \V_{p \in \mathcal{J} I} \V_{q \le \iota(p)} \iota_{I}\iota^{*}(q) \jsor \iota_{I}(\iota^{*}(q) \lres z) && {(**)}  \\
                                               &\le \V_{p \in \mathcal{J} I} \iota_{I}(p) \jsor \iota_{I}(p \lres z) && {({**}*)} \\
                                               &= (\iota_{I} \jsor \iota_{I})(\gamma_{I}(z)).
            \end{align*}
            For step $(*)$, we use that for \(p \in \mathcal{J} I, q \in \mathcal{J} K\) such that
            \(q \le \iota(p)\), the following holds:
            \begin{align*}
              \iota_{K}(q \lres \iota(z)) &= \iota_{K}(q) \lres \iota_{K}(\iota(z)) \\
                                  &= \iota_{K}(q) \lres \iota_{I}(z) \\
                                  &= \iota_{I}(\iota_{I}^{*}(\iota_{K}(q)) \lres z) && \iota_{I} \text{ finite (open) residuation morphism} \\
                                  &= \iota_{I}(\iota^{*}(\iota_{K}^{*}(\iota_{K}(q))) \lres z ) && \iota_{I} = \iota_{K} \iota \\
                                  &= \iota_{I}(\iota^{*}(q) \lres z ) && \iota_{K} \text{ embedding.} \\
            \end{align*}
            Simlarly, for step $(**)$ we use have
            \[\iota_{K} = \iota_{K} \cdot \id \le \iota_{K} \cdot \iota \cdot \iota^{*} = \iota_{I} \cdot \iota^{*}.\]
            Lastly, for step $({**}*)$ note that for every \(q \in \mathcal{J} K\) we have
            $\iota^*(q)\in \mathcal{J}(I)$: indeed, $\iota^*(q)\leq x\vee y$ in $I$ implies
            $q\leq \iota(x)\vee \iota(y)$ in $K$, hence $q\leq \iota(x)$ or $q\leq \iota(y)$, so $\iota^*(q)\leq x$
            or $\iota^*(q)\leq y$. In particular for all \(q \in \mathcal{J} K\) we find some
            \(p \in \mathcal{J} I\) with
            \[\iota_{I}\iota^{*}(q) \jsor \iota_{I}(\iota^{*}(q) \lres z) = \iota_{I}(p) \jsor (\iota_{I}(p) \lres z), \]
            which implies \(({**}*)\).
          \end{proof}
          Now, for well-definedness of \(\gamma_{\lres}\), if \(I, I'\) are finite
          residuation ideals containing \(z\) they are both contained in a
          finite residuation ideal $K$, since that $R$ is locally finite; we
          write \(\iota \colon I \incl K \hookleftarrow I' \coc \iota'\) for the inclusion maps. Now
          we have
          \begin{align*}
            (\iota_{I} \jsor \iota_{I})(\gamma_{I}(z)) &= (\iota_{K} \jsor \iota_{K})(\iota \jsor \iota)(\gamma(z)) && \text{\eqref{eq:gammaiota}} \\
                                          &= (\iota_{K} \jsor \iota_{K})(\gamma_{K}(\iota(z))) && \iota \text{ coalgebra morphism} \\
                                          &= (\iota_{K} \jsor \iota_{K})(\gamma_{K}(\iota'(z))) && \iota, \iota' \text{ subcoalgebras of \(K\)} \\
                                          &= (\iota_{I'} \jsor \iota_{I'})(\gamma_{I'}(z)) && \text{backwards.}
          \end{align*}
          This shows that the mapping%
          \[\gamma_{\lres} \colon R \mapsto R \jsor R, \qquad z \mapsto (\iota_{I} \jsor \iota_{I})(\gamma(z))\]
          does not depend on the choice of the residuation ideal \(I\).

    \item[(1b)]\label{item:1b} We show that the mapping \(\gamma_{\lres}\) indeed
          yields a \(\Um\)-coalgebra structure, that is, it preserves
          finite meets. Let \(F \subseteq R\) be a finite subset. By local finiteness we
          find a residuation ideal \(I\) containing \(F\). Now we simply use
          that both the comultiplication on \(I\) and \(\iota_{I} \jsor \iota_{I}\)
          preserve finite meets:
          \[\gamma_{\lres}(\A_{x \in F} x) = (\iota_{I} \jsor \iota_{I})(\gamma(\A_{x \in F} x)) = \A_{x \in F}(\iota_{I} \jsor \iota_{I})(\gamma(x)) = \A_{x \in F}\gamma_{\lres}(x). \]
    \item[(1c)] The coalgebra is easily seen to be locally finite, since for
          every finite subset \(X \subseteq R\) we find a finite residuation ideal \(I\)
          containing \(X\), and the corresponding coalgebra structure on~\(I\)
          is by definition a subcoalgebra of \((R, \gamma_{\lres})\).

    \item[(2a)] We first show that for finite subcoalgebras \(A, A'\) of \(R\)
          containing both \(x, z\) we have
          \[\iota_{A}(x \lres z) = \iota_{A'}(x \lres z).\]
          First, let \(\iota_{B} \cdot \iota \colon A \incl B \incl R\) be finite
          subcoalgebras. Then
          \begin{align*}
            \iota_{A}(x \lres z) &= \iota_{A}(x \moni \gamma(z)) \\
                             &= \iota_{B}(\iota(x \moni \gamma (z))) \\
                             &= \iota_{B}(\iota(x) \moni (\iota \jsor \iota)(\gamma(z))) && \text{embeddings preserve} \moni \\
                             &= \iota_{B}(\iota(x) \moni \gamma(\iota(z))) && \iota \text{ coalgebra morphism}\\
                             &= \iota_{B}(\iota(x) \lres \iota(z)).
          \end{align*}
          From this it follows that
          \begin{equation}
            \label{eq:loc-fin}
            \iota_{B}(x \lres \iota(z)) = \iota_{B}(x \moni (\iota \jsor \iota)(\gamma(z))) = \iota_{B}(\iota(\iota^{*}(x) \moni \gamma(z))) = \iota_{A}(\iota^{*}(x) \lres z).
          \end{equation}
          Now let $A \hookrightarrow C$ be a finite subcoalgebra containing $x$ and $z$. Then
          $A$ certainly contains the (finite) subcoalgebra \( \langle z \rangle\) generated
          by \(z\); we write \(\iota \colon \langle z \rangle \incl A \) for its inclusion into
          \(A\). We now obtain the canonical expression, where the last step
          uses~\eqref{eq:loc-fin}
          \[x \lres_{\gamma} z = \iota_{A}(x \lres z) = \iota_{A}(x \lres \iota(z)) = \iota_{\langle z \rangle }(\iota^{*}(x) \lres z) \]
          For general finite subcoalgebras \(A, A' \incl R\) containing \(x, z\)
          we find an upper bound \(\iota \colon A \incl B \hookleftarrow A' \coc \iota' \) and
          compute
          \[\iota_{A}(x \lres z) = \iota_{B}(\iota(x) \lres \iota(z)) = \iota_{B}(\iota'(x) \lres \iota'(z)) = \iota_{A'}(x \lres z).\]

    \item[(2b)] The proof that the residuals preserve finite meets in the
          covariant component is analogous to the proof for the
          comultiplication.

    \item[(2c)] The residuation algebra structure induced by \(\gamma\) is locally
          finite: Since the coalgebra \((C, \gamma)\) is locally finite, every finite
          subset \(F \subseteq C\) is contained in a finite subcoalgebra \(A \incl C\).
          The finite residuation algebra structure on \(A\) given by
          \iref{lem:gamma-res}{2} makes \(A\) a residuation subalgebra of
          \((C, \lres_{\gamma})\) containing \(F\).
    \item[(2d)] It remains to verify the residuation property:
          \begin{align*}
            y \le x \lres_{\gamma} z &\LR y \le \iota_{z}(\iota_{z}^{*}(x) \lres_{z} z) && \text{definition } \lres_{\gamma} \\
                              &\LR y \le \iota_{z}( \iota_{z}^{*} x \moni \gamma z) && \text{definition } \lres_{z} \\
                              &\LR \iota_{z}^{*} y \le \iota_{z}^{*} x \moni \gamma z  && \iota_{z}^{*} \dashv \iota_{z} \\
                              &\LR \iota_{z}^{*} x \jsor \iota_{z}^{*} y \le \gamma z  && \iota_{z}^{*} x \! \jsor \! (-) \, \dashv \, \iota_{z}^{*} x \! \moni \! (-) \\
                              &\LR \iota_{z}^{*} x \le \gamma z \imon \iota_{z}^{*} y   && (-) \! \jsor \! \iota_{z}^{*} y  \, \dashv \, (-) \imon \iota_{z}^{*} y \!  \\
                              &\LR x \le \iota_{z}( \gamma z \imon \iota_{z}^{*} y ) && \iota_{z}^{*} \dashv \iota_{z} \\
                              &\LR x \le \iota_{z}(z \rres_{\!z} \iota_{z}^{*} y) && \text{definition } \rres_{\! z}\\
                              &\LR x \le z \rres_{\! \gamma} y && \text{definition } \rres_{\! \gamma}
          \end{align*}
    \item[(3)] The translations are inverse since they are liftings of the
          translations between the operators on the finite substructures: We
          first show that \(\gamma_{\lres_{\gamma}} = \gamma\).
          For every \(z \in R\), the
          subcoalgebra \(\iota_{z} \colon \langle z \rangle \incl R\) generated by \(z\) is a
          residuation ideal of \(\lres_{\gamma}\): For all \(x \in R\), we have
          \[x \lres_{\gamma} z = \iota_{z}(\iota_{z}^{*}(x) \moni \gamma(z)) \in \langle z \rangle .\] We can
          therefore choose it as a residuation ideal containing \(z\) in the
          definition of \(\gamma_{\lres_{\gamma}}\) to get
          \[\gamma_{\lres_{\gamma}}(z) = (\iota_{z} \jsor \iota_{z})(\gamma_{\lres_{\gamma}}(z)) = (\iota_{z} \jsor \iota_{z})(\gamma(z)) = \gamma(z).\]
          An analogous argument proves
          \(\lres_{\gamma_{\lres}} = \lres\).\qedhere
  \end{enumerate}
\end{proof}

\autoref{prop:loc-fin-equiv} shows that every locally finite
residuation algebra carries a unique \(\Um\)-coalgebra structure and
vice versa. We may thus translate at will between the residuals and
comultiplication as in the complete case and omit the subscripts.  We
extend Lemmas \ref{lem:fin-op-equiv} and~\ref{lem:op-props} to locally finite structures:

\begin{lem}\label{lem:loc-fin-op-equiv}
  Let \(R\) be a locally finite residuation algebra.
  \begin{enumerate}
    \item\label{lem:loc-fin-op-equiv:sub} Finite residuation ideals correspond to finite subcoalgebras.
    \item The residuals  are associative iff the comultiplication is coassociative.
    \item It is prime-unital iff the comultiplication is prime-counital.
    \item\label{lem:loc-fin-op-equiv:pure} The comultiplication is pure  iff every finite residuation ideal is pure (see \autoref{def:res-alg}).
  \end{enumerate}
\end{lem}
\begin{proof}\hfill
  \begin{enumerate}
    \item If \(\iota_{I} \colon I \incl R\) is a finite residuation
      ideal, then by definition its comultiplication makes \(I\) a
      subcoalgebra of \((R, \gamma)\).
      
          In the reserve direction, let \(A \incl R\) be a
          finite subcoalgebra.
          The residuation algebra structure on \(A\) due to
          \iref{lem:gamma-res}{1} together with  the definition
          \(x \lres z = \iota_{A}(x \lres z)\) of the residuation algebra structure on \(R\)
          show that \(A\) is a
          residuation subalgebra of \(R\).
          To show that \(A\) is a residuation ideal, let \(x \in R\) and \(z \in A\).
          There exists a finite subcoalgebra \(B\) containing \(x,
          z\); we denote the inclusion map by \(\iota \colon A \incl B\).
          By \eqref{eq:loc-fin} we then have
          \[x \lres z = \iota_{B}(x \lres z) = \iota_{A}(\iota^{*}(x)
            \lres z),
          \]
          which states that $x \lres z$ lies in $A$.

    \item First, let \(\gamma\) be the coassociative comultiplication, and let \(x, y, z \in R\).
          By local finiteness these elements are contained in a finite coassociative subcoalgebra \(\iota_{A} \colon A \incl R\).
          Then by \iref{lem:loc-fin-op-equiv}{sub}, \(\iota_{A}\), is an associative finite residuation ideal of \(R\), whence
          \[x \lres (z \rres y) = \iota_{A}(x \lres (z \rres y)) = \iota_{A}((x \lres z) \rres y) = (x \lres z) \rres y.\]
          The other direction works analogously: if the residuals are associative and \(z \in R\), then
          it is contained in a finite associative residuation ideal \(\iota_{I} \colon I \incl R\).
          So \(I\) is a finite associative subcoalgebra of \(R\) and we have
          \begin{align*}
            (\bar{\gamma} \msor \id)(\bar{\gamma}(\iota_{I}(z))) &= (\iota_{I} \msor \iota_{I} \msor \iota_{I})((\bar{\gamma} \msor \id)(\bar{\gamma}(z)))  \\
            &= (\iota_{I} \msor \iota_{I} \msor \iota_{I})((\id \msor \bar{\gamma})(\bar{\gamma}(z)))  \\
            &= (\id \msor \bar{\gamma})(\bar{\gamma}(\iota_{I}(z))), \\
          \end{align*}
          proving that \(\bar{\gamma}\) is associative.
    \item Let \(\epsilon \colon R \rightarrow 2\) be a counit for the comultiplication
          \[\bar{\gamma} = \omega \cdot \gamma \colon U_{\land } R \rightarrow \Um(R \jsor R) \cong \Um R \msor \Um R \]
          with left adjoint element \(e \in \jsl(2, R) \cong R\).
          Note that \(\epsilon \cdot \iota_{A}\) is a counit for every subcoalgebra \(\iota_{A} \colon A \incl R\):
          Let \(z \in A\), then, since \(\epsilon\) is a unit for \(\bar{\gamma}\), we have
          \[\iota_{A} (((\epsilon \cdot \iota_{A}) \msor \id_{A})(\bar{\gamma}(z))) = (\epsilon \msor \id)(\iota_{A} \msor \iota_{A})(\bar{\gamma}(z)) = (\epsilon \msor \id)(\bar{\gamma}(\iota_{A}(z))) = \iota_{A}(z).\]
          But by the implication \ref{lem:op-props:gamma} \(\Rightarrow\) \ref{lem:op-props:res} of \autoref{lem:op-props} this means that the left adjoint \(e_{A} := \iota_{A}^{*}(e) \in A\) of \(\epsilon \cdot \iota_{A}\) is a unit for the finite residuation ideal \(A\) of \(R\).
          So for all \(z \in R\) we have
          \[e \lres z = \iota_{z}(\iota_{z}^{*}(e) \lres z) = \iota_{z}(e_{z} \lres z) = \iota_{z}(z) = z, \]
          so \(e\) is a unit for the residuals.
          The case for the right residual \(\rres\) is dual.

          For the other direction let \(e \in R\) be a unit for the residuals with right adjoint \(\epsilon \colon R \rightarrow 2\).
          For every residuation ideal \(\iota_{I} \colon I \incl R\) the element \(\iota_{I}^{*}(e) \in I\) is the unit of \(I\):
          The embedding trivially  is an open residuation morphism,
          and we have
          \[
            \iota_{I}(\iota_{I}^{*}(e) \lres z) = e \lres \iota_{I}(z)
            = \iota_{I}(z)
            \qquad\text{for every \(z \in I\)}.
          \]
          So the subcoalgebra structure \(I\) has the counit
          \(\epsilon \cdot \iota_{I}\). Now given $z \in R$, we pick a
          residuation ideal $I$ containing $z$ using that $R$ is
          locally finite. Then we have
          \begin{align*}
            (\epsilon \msor \id)(\bar{\gamma}(z)) &= (\epsilon \msor \id)(\bar{\gamma}(\iota_{I}(z))) \\
                                      &= (\epsilon \msor \id)((\iota_{I} \msor \iota_{I})(\bar{\gamma}(z))) \\
                                      &= \iota_{I}((\epsilon \iota_{I} \msor \id)(\bar{\gamma}(z))) \\
                                      &= \iota_{I}(z) = z.
          \end{align*}
          This shows that \(\epsilon\) is a counit for \(\bar{\gamma}\).
    \item If the comultiplication is pure, associative and has a counit, then this holds for every finite subcoalgebra.
          Every finite residuation ideal \(I\) of \(R\) is a finite pure subcoalgebra, which therefore is a derivation algebra.

          Conversely, if every finite residuation ideal is a derivation algebra, then we only have to show that the comultiplication preserves finite joins, since it already is coassociative and has a counit.
          The join of finitely many elements
          is taken in  some finite subcoalgebra \(A\).
          But~\(A\) is a finite residuation ideal and therefore by assumption a derivation algebra, and the comultiplication preserves finite joins.\qedhere
  \end{enumerate}
\end{proof}

\begin{rem}\label{rem:join-pres-at-prim}
  \iref{lem:loc-fin-op-equiv}{pure} characterizes locally finite
  residuation algebras with a pure comultiplication. By extended
  duality, its dual Priestley relation is functional. 
  Gehrke~\cite[Prop.~3.15]{gehrke-16} has presented a necessary and
  sufficient condition for a general residuation algebra~\(R\) to have
  a functional dual relation, namely \emph{join-preservation at
    primes}:
  \[\forall F \in \dl(R, 2) \colon \forall a \in F\colon \forall b, c \in R \colon \exists a' \in F \colon a \lres (b \lor c) \le (a' \lres b) \lor (a' \lres c). \]
  One can indeed  show that every locally finite residuation algebra with a pure comultiplication (\iref{lem:loc-fin-op-equiv}{pure}) is join-preserving at primes:
    If \(F\) is a prime filter on \(C\) and \(a \in F\), then we obtain by local finiteness for all \(b, c \in C\) a finite \emph{pure} residuation ideal \(\iota \colon I \incl C\) containing \(a, b, c\).
    In this residuation ideal we have \(a = \V K \) for join-primes \(K \subseteq \mathcal{J} I \subseteq C \).
    Since \(F\) is prime, some \(a' \in K\) lies in \(F\), and it satisfies
    \begin{align*}
      a \lres (b \lor c) &= \iota(a \lres (b \lor c)) \\
                         &= \iota((\V K) \lres (b \lor c)) \\
                         &\le \iota(a' \lres (b \lor c)) && \text{\((-) \lres x\) antimonotone} \\
                         &= \iota(a' \lres b \lor a' \lres c) && \text{\(I\) is pure} \\
                         &= \iota(a' \lres b) \lor \iota(a' \lres c) \\
                         &= (a' \lres b) \lor (a' \lres c) && a', b, c \in I.
    \end{align*}
    This shows that the residuals are join-preserving at primes.
\end{rem}

\begin{defi}
  A  residuation algebra \(R\) is a \emph{derivation
    algebra} if it is locally finite, associative, prime-unital and
  every finite residuation ideal~\(I\) is pure.
 This yields full subcategories
\[ \der \hookrightarrow \res \qquad\text{and}\qquad \relres \hookrightarrow \relder.\]
\end{defi}

Note that a derivation algebra with a finite carrier is precisely a finite derivation ACDL as defined in \iref{def:res-alg}{der-alg}. Recall from \autoref{def:fin-res-mor} the definition morphism for \(\Um\)-coalgebras and comonoids.

\begin{prop}\label{prop:loc-fin-mor}
  Let \(R, R'\) be locally finite residuation algebras with units.
  \begin{enumerate}
    \item A lattice morphism \(f \in \dl(R, R')\) is a residuation morphism iff it is a prime-counital morphism of \(\Um\)-coalgebras.
    \item If \(R, R'\) are comonoids, a finite join-preserving function \(\rho \in \jsl(R, R')\) is a corelational residuation morphism iff it is a corelational comonoid morphism.
  \end{enumerate}
\end{prop}
\begin{proof}\hfill
  \begin{enumerate}
    \item First, let \(f \colon R \rightarrow R'\) be a residuation morphism.
          For \(z \in R\) we choose a finite residuation ideal \(I' \incl R'\) containing \(f(z)\).
          Since \(f\) is a residuation morphism we have by the \eqref{eq:back} condition that for every \(y \in I'\) there exists some \(x_{y,z}\) with \(y \le f(x_{y,z})\) and \(y \lres f(z) = f(x_{y,z} \lres z)\).
          We now choose a finite ideal \(I \incl R\) containing \(z\) and all \(x_{y,z}\) for \(y \in I'\).
          We therefore have
          \begin{align*}
            \gamma_{\lres}(f(z)) &= \V_{y \in I'} y \jsor y \lres f(z) \\
                            &\le \V_{y \in I', y \le f(x_{y,z})}  f(x_{y,z}) \jsor f(x_{y,z} \lres z)  \\
                            &= (f \jsor f)(\V_{y \in I', y \le f(x_{y,z})} x_{y,z} \jsor x_{y,z} \lres z) \\
                            &\le (f \jsor f)(\V_{x \in I} x \jsor x \lres z) \\
                            &= (f \jsor f)(\gamma_{\lres}(z)).
          \end{align*}
          For the reverse inequality let \(z \in R\) be contained in the finite residuation ideal \(I\).
          We choose a finite residuation ideal \(J \incl R'\) containing \(f(z)\) and all \(f(x), x \in I\), and use \eqref{eq:forth}:
          \begin{align*}
            (f \jsor f)(\gamma_{\lres}(z)) &= \V_{x \in I}f(x) \jsor f(x \lres z) \\
                                      &\le \V_{x \in I}f(x) \jsor f(x) \lres f(z) \\
                                      &\le \V_{y \in J} y \jsor y \lres f(z) \\
            &= \gamma_{\lres}(f(z)),
          \end{align*}
This proves that $f$ is a morphism of \(\Um\)-coalgebras.

Conversely, let \(f \colon R \rightarrow R'\) be a morphism of \(\Um\)-coalgebras.
          For every \(z \in R\) the morphism \(f\) restricts to the  finite subcoalgebras generated by \(z, f(z)\) as  \(f_{z} \colon \langle z \rangle \rightarrow \langle f(z) \rangle\).
          If we denote the respective inclusions by \(\iota_{z} \colon \langle z \rangle \incl R\) and \(\iota_{f(z)} \colon \langle f(z) \rangle \incl R'\), then this is equivalent to saying that \(f \cdot \iota_{z} = \iota_{f(z)} \cdot f_{z}\).
          From the unit of \(\iota_{z}\) we thus get
          \[f \le f \cdot \iota_{z} \cdot \iota_{z}^{*} = \iota_{f(z)} \cdot f_{z} \cdot \iota_{z}^{*},\]
          which under transposition is equivalent to
          \begin{equation}
            \label{eq:iota}
            \iota_{f(z)}^{*} \cdot f \le f_{z} \cdot \iota_{z}^{*}.
          \end{equation}
          This entails the \eqref{eq:forth} condition:
          \begin{align*}
            f(x \lres_{\gamma} z) &= f(\iota_{z}(\iota_{z}^{*}(x) \moni \gamma(z))) && \text{def. } \lres_{\gamma} \\
                             &= \iota_{f(z)}(f_{z}(\iota_{z}^{*}(x) \moni \gamma(z))) && f_{z} \text{ restriction} \\
                             &\le \iota_{f(z)}(f_{z}(\iota_{z}^{*}(x)) \moni (f_{z} \jsor f_{z}) \gamma(z))) && \moni \text{\autoref{prop:mon-impl}\ref{prop:mon-impl:preserve}} \\
                             &= \iota_{(z)}(f_{z}(\iota_{z}^{*}(x)) \moni \gamma(f_{z}(z))) && f_{z} \text{ coalgbra morphism} \\
                             &\le \iota_{f(z)}(\iota^{*}_{f(z)}(f(x)) \moni \gamma(f_{z}(z))) && \text{\eqref{eq:iota} + contravariance} \\
                             &\le \iota_{f(z)}(\iota^{*}_{f(z)}(f(x)) \moni \gamma(f(z))) && f_{z}(z) = f(z) \\
                             &= f(x) \lres_{\gamma} f(z).
          \end{align*}
          To verify the \eqref{eq:back} condition, let \(y \in R', z \in R\) and put \[x_{y,z} = \iota_{z}(f_{z}^{*}(\iota_{f(z)}^{*}(y))).\] Then
 \[y \le \iota_{f(z)}\iota_{f(z)}^{*}(y) \le \iota_{f(z)}(f_{z}(f_{z}^{*}(\iota_{f(z)}^{*}(y)))) = f(\iota_{z}(f_{z}^{*}(\iota_{f(z)}^{*}(y)))) = f(x_{y,z})\]
and
          \begin{align*}
            y \lres_{\gamma} f(z) &= \iota_{f(z)}(\iota^{*}_{f(z)}(y) \moni \gamma(f_{z}(z))) && \text{def. } \lres_{\gamma} \\
                             &= \iota_{f(z)}(\iota_{f(z)}^{*}(y) \moni (f_{z} \jsor f_{z})(\gamma(z))) && f_{z} \text{ coalgebra morphism} \\
                             &= \iota_{f(z)}(f_z(f_{z}^{*}(\iota_{f(z)}^{*}(y)) \moni \gamma(z))) && \text{\autoref{prop:mon-impl}\ref{prop:mon-impl:preserve}} \\
                             &= f(\iota_{z}(f^{*}_{z}(\iota_{f(z)}^{*}) \moni \gamma(z))) && \iota_{f(z)} \cdot f_{z} = f \cdot \iota_{z} \\
                             &= f(\iota_{z}(\iota_{z}^{*}(\iota_{z}(f_{z}^{*}(\iota_{f(z)}^{*}(y)))) \moni \gamma(z))) && \iota_{z}^{*}\iota_{z} = \id \\
            &= f_{z}(x_{y,z} \lres_{\gamma} z)  && \text{def. } x_{y,z}.
          \end{align*}
          For the prime-(co-)unitality conditions we split the pointwise equality \(\forall x \colon \epsilon'(f(x)) = \epsilon(x)\) into \(\epsilon'(f(x)) \le \epsilon(x)\) and \(\epsilon(x) \le \epsilon'(f(x))\).
          These are equivalent to \(e' \le f(x) \Rightarrow e \le x\) and \(e \le x \Rightarrow e' \le f(x)\), respectively, combining to the desired condition \(\forall x \colon e' \le f(x) \Leftrightarrow e \le x\).
    \item Now let \(R, R'\) be comonoids and let \(f \in \jsl(R, R')\) be a corelational comonoid morphism.
          For \(x, z \in R\) we choose a finite subcomonoid \(\iota \colon I \incl R\) that contains \(x, z\) and a finite subcomonoid \(\iota' \colon I' \incl R'\) (with comultiplication $\gamma'$) containing \(\rho[I]\).
          Then \(\rho\) restricts to a corelational morphism \(\rho \colon I \rightarrow I'\) of finite comonoids. By \autoref{thm:fin-iso} \(\rho\) is a corelational morphism of finite residuation algebras, so it satisfies \(\rho(a \lres_{\gamma} b) \le \rho(a) \lres_{\gamma'} \rho(b)\) for all \(a, b \in I\).
          We therefore get
          \[\rho(x \lres_{\gamma} z) = \rho(\iota(x \lres_{\gamma} z)) = \iota'(\rho(x \lres_{\gamma} z)) \le \iota'(\rho(x) \lres_{\gamma'} \rho(z)) = \rho(x) \lres_{\gamma'} (z).  \]
          which proves that \(f\) is a corelational residuation morphism.
          To verify \(e' \le \rho(e)\) we again choose finite subcoalgebras \(A, A'\) with \(e \in A\) and  \(\rho[A] \cup \{e'\} \subseteq A'\).
          Since \(\rho\) is prime-counital it satisfies \(\epsilon \le \epsilon' \cdot \rho\) and therefore also \(\epsilon \cdot \iota_{A} \le \epsilon' \cdot \rho \cdot \iota_{A} = \epsilon' \cdot \iota_{A'} \cdot \rho \).
          As \(\epsilon \cdot \iota_{A}\) and \(\epsilon' \cdot \iota_{A'}\) are the counits for  \(A\) and \(A'\), respectively, its restriction \(\rho \colon A \rightarrow A'\) is thus also prime-counital and whence satisfies \(\iota_{A}^{*}(e) \le \rho(\iota_{A'}^{*}(e'))\) for the corresponding units of the residuals on \(A, A'\).
          But \(e \in A\) and \(e' \in A'\), so this equation simplifies to the desired \(e \le \rho(e')\).

          Conversely, if \(f\) is a corelational residuation morphism choose for \(z \in R\)
          a finite residuation ideal \(\iota \colon I \incl R\) containing \(z\) and a finite residuation ideal \(\iota' \colon I' \incl R\) containing \(\rho[I]\).
          Then we have
          \begin{align*}
            (\rho \jsor \rho)(\gamma(z)) &= (\rho \jsor \rho)(\gamma(\iota(z))) \\
                              &= (\rho \jsor \rho)(\iota \jsor \iota)(\gamma(z)) \\
                              &= (\iota' \jsor \iota')(\rho \jsor \rho)(\gamma(z)) \\
                              &\le (\iota' \jsor \iota') \gamma(\rho(z)) \\
                              &= \gamma(\rho(z)),
          \end{align*}
          so \(\rho\) is a corelational residuation morphism.
          To show that \(\epsilon \le \epsilon' \cdot \rho\), take \(z \in R\) with ideals \(I, I'\) chosen as before.
          Recall that \(\iota^{*}(e)\) is a unit of the residuation ideal \(I\) and \(\epsilon \cdot \iota\) is a counit for the corresponding subcoalgebra.
          Since \(\rho\) is prime-unital it satisfies
          \[e' \le \rho(e) \le \rho(\iota(\iota^{*}(e))) = \iota'(\rho(\iota^{*}(e)))\]
          which is equivalent to \(\iota'^{*}(e') \le \rho(\iota^{*}(e))\).
          Since \(\rho \colon I \rightarrow I'\) is a corelational residuation morphism it is a corelational morphism of the coalgebra structures on \(I, I'\) and we thus get
          \[\epsilon(z) = \epsilon(\iota(z)) \le \epsilon'(\iota'(\rho(z))) = \epsilon'(\rho(\iota(z))) = \epsilon'(\rho(z)).\tag*{\qedhere}\]
  \end{enumerate}
\end{proof}

\begin{rem}
  The proof of \autoref{prop:loc-fin-mor} gives an alternative formulation of the \eqref{eq:back} condition for locally finite residuation algebras as
  \[
    y \lres f(z) = f((\iota_{z} \cdot f_{z}^{*} \cdot
    \iota_{f(z)}^{*})(y) \lres z).
  \]
  Here we choose the existentially quantified \(x_{y,z}\) in \eqref{eq:back} via \(f_{z}^{*}\) \emph{locally}:
\[
    \begin{tikzcd}[sep=1.2em]
      x_{y,z}
      \drar[phantom]{\rotatebox{-45}{\(\in\)}}
      \ar[mapsto]{rrr}{}
      &
      &
      &
      f(x_{y, z})
      \dar[phantom]{\rotatebox{-90}{\(\in\)}}
      \ar[phantom]{r}{\scalebox{.8}{\(\ge\)}}
      &
      y
      \dlar[phantom]{\rotatebox{225}{\(\in\)}}
      \ar[mapsto]{ddddd}{}
      \\
      &
      R
      \ar{rr}{f}
      \ar[shift right=2,swap]{ddd}{\iota_{z}^*}
      \ar[phantom]{ddd}{\dashv}
      &
      &
      S
      \ar[shift right=2,swap]{ddd}{\iota_{f(z)}^*}
      \ar[phantom]{ddd}{\dashv}
      &
      \\
      & & & &
      \\
      & & & &
      \\
      &
      \genres z
      \ar[yshift=-4, swap]{rr}{f_{z}}
      \ar[phantom]{rr}{\rotatebox{270}{\(\dashv\)}}
      \ar[hook, shift right=2,swap]{uuu}{\iota_{z}}
      &
      &
      \genres{f(z)}
      \ar[yshift=4, swap]{ll}{f_{z}^*}
      \ar[hook, shift right=2,swap]{uuu}{\iota_{f(z)}}
      &
      \\
      f_{z}^* (\iota_{f(z)}^*(y))
      \urar[phantom]{\rotatebox{45}{\(\in\)}}
      \ar[mapsto]{uuuuu}{}
      &
      &
      &
      &
      \iota_{f(z)}^*(y)
      \ular[phantom]{\rotatebox{135}{\(\in\)}}
      \ar[mapsto]{llll}{}
    \end{tikzcd}
  \]
  Compare this with Equation \eqref{eq:open} satisfied by open residuation morphisms,
  where the existence of a global left adjoint \(f^{*}\) allows one to
  choose \(x_{y,z} = f^{*}(y)\) \emph{globally}, that is, independently of \(z\).
\end{rem}
We then have the following extension of \autoref{thm:fin-iso} to locally finite structures:

\begin{thm}\label{thm:loc-fin-dual}\hfill
  \begin{enumerate}
  \item\label{thm:loc-fin-dual:res} The category of
   locally finite residuation algebras and resi\-du\-ation morphisms is isomorphic
    to the category of locally finite prime-unital \(\Um\)-coalgebras and
    pure coalgebra morphisms.

  \item\label{thm:loc-fin-dual:der} This isomorphism restricts to an isomorphism between the
    full subcategories of derivation algebras and locally finite comonoids.

  \item\label{thm:loc-fin-dual:rel} The category of
    derivation algebras and corelational residuation morphisms is isomorphic to the
    category of locally finite comonoids with relational
    morphisms.
  \end{enumerate}
\end{thm}
\begin{proof}
Immediate from  \autoref{lem:loc-fin-op-equiv} and \autoref{prop:loc-fin-mor}.
\end{proof}

\subsection{Duality Theory for Locally Finite Residuation Algebras}
\label{sec:duality-theory}

We now gathered the ingredients to present the first application of our abstract extended duality (\autoref{thm:extended-duality}): a categorical duality between \emph{profinite ordered monoids} and
\emph{derivation algebras}.
Recall that a profinite ordered monoid is a codirected limit of finite ordered monoids; like in the order-discrete setting, they are equivalent to \emph{Priestley monoids}, viz.\ monoids in the cartesian category \priest (\autoref{prop:every-priest-mon-prof}).
This result is a non-trivial
restriction of Gehrke's duality~\cite{gehrke-16-proc,gehrke-16}
between Priestley-topological algebras and residuation algebras.

Conceptually, this general duality is an extension of the finite duality \(\ordmon_{\mathrm{f}}
\simeqop \comon_{\mathrm{f}} \cong \der_{\mathrm{f}}\) by forming suitable completions. We
start by investigating the Ind- and Pro-completions of the categories involved in the finite
duality \autoref{rem:fin-dual}.%

\begin{rem}\label{rem:ind-pro}
  The \emph{Ind-completion} (or \emph{free completion under filtered colimits}) of a small category $\cat C$ is given by a category $\ind(\cat C)$ with filtered (equivalently directed) colimits and a full embedding $I\colon \cat C\hookrightarrow \ind(\cat C)$ such that every functor $F\colon \cat C\to \cat D$ into a category $\cat D$ with filtered colimits extends to a functor $\overline{F}\colon \ind(\cat C)\to \cat D$, unique up to natural isomorphism, such that $F=\overline{F}\cdot I$. To show that a category $\cat D$ is the Ind-completion of a full subcategory $\cat C$, it suffices to prove the following (see e.g.~\cite[Thm.~A.4]{adamek-chen-milius-urbat-21}):
  \begin{enumerate}
  \item $\cat D$ has filtered colimits,

  \item every object of $\cat D$ is a filtered colimit of objects of $\cat C$, and

  \item every object $C$ of $\cat C$ is finitely presentable in $\cat D$, that is, the functor $\cat D(C,-)\colon \cat D\to \set$ preserves filtered colimits.  
  \end{enumerate}
  The \emph{Pro-completion} (or \emph{free completion under cofiltered limits}) of $\cat C$ is defined dually.
\end{rem}

Profinite ordered monoids form the {Pro-completion} of the category of finite ordered monoids.
Dually, lattice comonoids (and therefore also derivation algebras by \iref{thm:loc-fin-dual}{der}) form {Ind-completions}  of their respective subcategories of finite objects:


\begin{prop}\label{prop:loc-fin-comon-ind-compl}
  The category of locally finite comonoids forms the Ind-completion of the category of finite comonoids:
  \[\comon_{\mathrm{lf}} \simeq \ind(\comon_{\mathrm{f}}).\]
\end{prop}
\begin{proof}
  (a) We first show that filtered colimits of lattice comonoids are formed in \set.
  First, since \dl is a category of algebras over a finitary signature, filtered colimits in \dl are formed in \set.
  Second, since \(+ \colon \dl \times \dl \rightarrow \dl\) is a finitary functor (colimits commute with colimits) filtered colimits in the category of \(+\)-coalgebras are also formed in \set.
  As comonoids are a full subcategory of \(+\)-coalgebras it suffices to show that the filtered colimit in \(+\)-coalgebras of lattice comonoids is again a comonoid, which is a straightforward verification:
  Let \(d_{i} \colon D_{i} \rightarrow D_{i} + D_{i}, i \in I\) be a cofiltered diagram of comonoids, and let \(d \colon D \rightarrow D + D\)  be their colimit in the category of \(+\)-coalgebras with colimit injections \(\kappa_{i} \colon D_{i} \rightarrow D\).
  Since the colimit \(D\) is formed in $\set$, there exists for every \(x \in D\) some \(i \in I\) and \(x_{i} \in D_{i}\) with \(\kappa_{i}(x_{i}) = x\).
  But then
  \begin{align*}
    (d + \id)(d(x)) &= (d + \id)d(\kappa_{i}(x_{i})) && \\
                    &= (d + \id)(\kappa_{i} + \kappa_{i})(d_{i}(x_{i})) && \text{\(\kappa_{i}\) comonoid morphism} \\
                    &= ((d \cdot \kappa_{i}) + \kappa_{i})(d_{i}(x_{i})) && \\
                    &= (((\kappa_{i} + \kappa_{i}) \cdot d_{i}) + \kappa_{i})(d_{i}(x_{i})) && \text{\(\kappa_{i}\) comonoid morphism} \\
                    &= (\kappa_{i} + \kappa_{i} + \kappa_{i})(d_{i} + \id)(d_{i}(x_{i})) && \\
                    &= (\kappa_{i} + \kappa_{i} + \kappa_{i})(\id + d_{i})(d_{i}(x_{i})) && d_{i} \text{ comonoid}\\
                    &= \ \cdots && \text{backwards} \\
                    &= (\id + d)(d(x)), &&
  \end{align*}
  so \(D\) is coassociative.
  Counitality works similarly for \(\epsilon \colon D \rightarrow 2\) with \(\epsilon(x) = \epsilon(\kappa_{i}(x_{i})) = \epsilon_{i}(x_{i})\).
  This proves that filtered colimits of comonoids are formed in \set.

  (b) To prove that the category $\comon_{\mathrm{lf}}$ is  the Ind-completion of its full subcategory $\comon_{\mathrm{f}}$, we verify the conditions of \autoref{rem:ind-pro}: 
\begin{enumerate}
  \item $\comon_{\mathrm{lf}}$ has filtered colimits: filtered colimits of comonoids are formed in \set by (a). Moreover, a filtered colimit $c_i\colon C_i\to C$ ($i\in I$) of locally finite comonoids is locally finite: Given $x\in C$ one has $x=c_i(x_i)$ for some $i\in I$ and $x_i\in C_i$. Thus $x_i\in C_i'$ for some finite subcomonoid $C_i'$ of $C_i$, and so $x_i$ lies in the finite subcomonoid $c_i[C_i']\subseteq C$.
\item Every locally finite comonoid is the directed union of the diagram of all its finite subcomonoids. This follows again from (a), since this is clearly a directed union in $\set$.
\item Every finite comonoid is finitely presentable:  A finite comonoid can be regarded as a coalgebra $C\to 2\times (C+C)$ for the functor $FX=2\times (X+X)$ on $\dl$ by pairing its counit and comultiplication. Every finite $F$-coalgebra is finitely presentable in the category of all $F$-coalgebras~\cite[Lemma~3.2]{ap04}. This implies the corresponding statement for finite comonoids since they form a full subcategory of the category of $F$-coalgebras.\qedhere
\end{enumerate}
\end{proof}

On the dual side, we require a property of Priestley monoids (\autoref{prop:every-priest-mon-prof} below), namely that every Priestley monoid is profinite, that is, it is a cofiltered limit of finite ordered monoids with the discrete topology. The corresponding result for (unordered) Stone monoids is well known~\cite[Thm.~VI.2.9]{johnstone-82}; its ordered version is analogous, and we give a full proof for the convenience of the reader.

We need some auxiliary results:
First, we recall a well-known characterization for quotients of ordered algebras, see e.g.~\cite[Sec.~B.2]{milius-urbat-19}.
Recall that a preorder \(\sqsubseteq\) on an ordered monoid \(A\) is
\emph{stable} if \(\sqsubseteq\) refines the order on \(A\)
(i.e.~\(\le\) is included in \(\sqsubseteq\)) and the multiplication
is monotone: \(a \sqsubseteq b, a' \sqsubseteq b'$ implies $aa' \sqsubseteq bb'\).

\begin{lem}\label{lem:stable-preorder}
  Let \(X\) be an ordered monoid.
  Then ordered monoid quotients of \(X\) are in bijective correspondence with stable preorders on \(X\).
\end{lem}
\begin{proof}
  A quotient \(e \colon A \epi B\) induces a stable preorder by putting \(a \sqsubseteq a'\) iff \(e(a) \le e(a')\) for \(a, a' \in A\).
  Conversely, from a stable preorder \(\sqsubseteq \) on \(A \times A\) we obtain an equivalence relation on \(A\) by identifying \(a \sim a'\) iff \(a \sqsubseteq a'\) and \(a' \sqsubseteq a\).
  Monotonicity of the multiplication then ensures that multiplication on equivalence classes is well-defined.
  The refinement property ensures that the canonical projection, which maps an element to its equivalence class, is order-preserving.
  It is a routine verification to check that these constructions are inverses.
\end{proof}

\begin{lem}\label{lem:priest-sep}
  Let \(X\) be a Priestley monoid. If \(x \not \le y\), then there exists a finite Priestley monoid quotient \(f\colon X \epi M \) such that \(fx \not \le fy\).
\end{lem}
  We modify the proof of a corresponding statement for Stone algebras given by Johnstone~\cite[Ch.~VI, Sec.~2.7]{johnstone-82}.
\begin{proof}
  Let \(x \not \le y\) be elements of a Priestley monoid~\(X\), we show that there exists a quotient of \(f \colon X \epi M\) into some finite discretely-topologized ordered monoid \(M\) satisfying \(f(x) \not \le f(y)\).
  Since the underlying Priestley space of \(X\) is profinite, there exists a continuous surjection \(e \in \priest(X, A)\) such that \(A\) is a finite discretely topologized poset, satisfying \(e(x) \not \le e(y)\).
  We denote the preorder on \(X\) induced by \(e\) by \(\sqsubseteq\) (i.e.~\(x \sqsubseteq y\) iff \(e (x) \le e (y)\)).
  We define \(\preceq\) by \[x \preceq y \quad \text{iff} \quad \forall u, v \in X \colon uxv \sqsubseteq uyv,\] and denote its corresponding equivalence relation by \(\approx\).

  (1) We prove that \(\preceq\) is a stable preorder.
  First we show that \(\preceq\) refines \(\le\): if \(x \le y\), then \(\forall u, v \in X \colon uxv \le uyv\) since the multiplication is monotone. Moreover, we have \(uxv \sqsubseteq uyv\) for all $u,v \in X$, since \(\le\) is contained in \(\sqsubseteq\).
  It follows that multiplication is \(\preceq\)-monotone: if \(x \preceq y\) and \(x' \preceq y'\), then for all \(u, v \in X\) we have
  \[ u(xx')v= ux(x'v) \sqsubseteq uy(x'v) = (uy)x'v \sqsubseteq (uy)y'v = u(yy')v,\]
  whence \(xx' \preceq yy'\).

  (2a) We prove that the equivalence relation \({\approx} \subseteq X \times X\) is open.
  The equivalence relation induced by $\sqsubseteq$ is denoted by \(\equiv\), that is, \(x \equiv y\) iff \(x \sqsubseteq y \sqsubseteq x\)  iff \(e(x) = e(y)\).
  Note that  \(\equiv\) is open since it is the preimage of the (open) diagonal of \(A\) under \(e \times e\).
  Now let \(m \approx n\). By openness of \(\equiv\) and continuity of the multiplication we obtain for all \(u, v \in X\) open neighbourhoods \(U_{u}, U_{v}, V_{u,v}\) and \(W_{u,v}\) of \(u, v, m, n\), respectively, such that
  \begin{align}
    \label{eq:env}
    \forall u' \in U_{u}, v' \in U_{v}, m' \in V_{u, v}, n' \in W_{u, v} \colon u'm'v' \equiv u'n'v'.
  \end{align}
  By compactness of \(X\) we have \(X \times X = \bigcup_{u, v} U_{u} \times U_{v} = \bigcup_{i=1}^{n} U_{u_{i}} \times U_{v_{i}}\) for some $n$ and $u_i,v_i$.
  Now set \(V = \bigcap_{i=1}^{n} V_{u_{i}, v_{i}}\) and \(W = \bigcap_{i=1}^{n} W_{u_{i}, v_{i}}\). Then
  \(V \times W\) is an open neighbourhood of \((m, n)\) satisfying \(V \times W \subseteq {\approx}\):
  for \((m', n') \in V \times W\) we have for all \(u', v' \in X\) some \(i\) with \((u', v') \in U_{u_{i}} \times U_{v_{i}}\). Using \eqref{eq:env} we have \(u'm'v' \equiv u'n'v' \), proving \(m' \approx n'\).

  (2b) We prove that every equivalence class \([x]_{\approx}\) is open.
  For every  \(y \in [x]_{\approx}\), there exists, since \(\approx\) is open by (2a), a basic open \(U \times V \subseteq \,\, \approx \) with \((x, y) \in U \times V\).
  For every \(y' \in V\) we have \((x, y') \in U \times V \subseteq { \approx } \), so \(V\) is an open neighbourhood of \(y\) that is contained in \([x]_{\approx}\), proving that \(\approx\) is open.

  (3) It follows that the quotient \(f\) induced by $\preceq$ has the desired properties:
  (3a) It is a homomorphism of ordered monoids by \autoref{lem:stable-preorder}, since \(\preceq\) is a stable preorder.
  (3c) The codomain \(X / {\approx}\) is finite and discrete, since equivalence relations with open equivalence classes are discrete, so \(X / {\approx}\) is finite by compactness.
  (3d) Finally we show that \(f\) order-separates \(x\) and \(y\):
  For the sake of contradiction suppose \(f(x) \le f(y)\).
  By definition we have \(x \preceq y\), implying  \(x = 1x1
  \sqsubseteq 1y1 = y\),
  a contradiction.
 \end{proof}

\begin{prop}\label{prop:every-priest-mon-prof}
  Every Priestley monoid is profinite: \[\cat{PriestMon} \cong \cat{ProfOrdMon}.\]
\end{prop}
\begin{proof}
  For every $X\in \cat{PriestMon}$ we prove that \(X \cong \lim D_{X}\), where \(D_{X}\) is the canonical codirected diagram for \(X\) over all finite discretely-topologized ordered monoid quotients of~\(X\); in particular, this shows that $X$ is profinite.

  First note that the limit of a diagram \(D\) in \cat{PriestMon} is  formed by taking the limit \(L=\lim |D|\) of the underlying diagram \(| - | \cdot D\) in \stonemon and equipping it with the product order. Then all projections are monotone, and
  this order makes \(L\) totally order-disconnected:
  if \((x_{e})_{e}, (y_{e})_{e} \in L\) with \((x_{e})_{e} \not
  \le (y_{e})_{e}\) then there exists by definition of the order some
  quotient \(e \colon X \epi X_{e}\)
  with \(x_{e} \not \le y_{e}\).
  Hence \(p_{e}^{-1}[{\uparrow} x_{e}]\) is a clopen upset of \( L\) containing \((x_{e})_{e}\) but not \((y_{e})_{e}\).
  This shows that \(L\) is a Priestley monoid, and it is easy to verify that it satisfies the universal property in \cat{PriestMon}.

  We have to prove that the canonical homomorphism \(\phi \colon X \rightarrow \lim D_{X}\) is an isomorphism of Priestley monoids, or equivalently, that it is a continuous surjective order-embedding. Continuity and order-preservation are given since \(\phi\) is a morphism in \priest.

  That \(\phi\) is an order-embedding follows from it  being  order-reflecting: \(\phi(x) \le \phi(y) \Rightarrow x \le y\).
  By contraposition it suffices to show \(x \not \le y \Rightarrow \phi(x) \not \le \phi(y)\), but the right side of the implication is equivalent to finding an ordered monoid quotient \(f \colon X \epi A\) into a finite Priestley monoid with \(f(x) \not \le f(y)\).
  This is precisely the content of \autoref{lem:priest-sep}.

  To show that $\phi$ is surjective, we use a general property of codirected limits of compact Hausdorff spaces \cite[Lemma 1.1.5]{rz10}: if \(f_{i} \colon X \epi X_{i}\) is a compatible family of surjections into a codirected diagram \(X_{i}\), then the induced mapping \(f \colon X \rightarrow \lim X_{i}\) is surjective.
\end{proof}

By duality we immediately obtain the following result:

\begin{cor}
  Every comonoid is locally finite: \(\comon \;\cong\; \comon_{\mathrm{lf}}\).
\end{cor}

\begin{defi}
  Let \(X\) and \(Y\) be Priestley monoids.
  A \emph{Priestley relational morphism} \(X \rightarrow Y\) is a total Priestley relation
  \(\rho \colon X \rightarrow \dviet Y\) such that
  \[\rho(x) \rho(x') \subseteq \rho(xx') \qquad \text{ and } \qquad 1_{N} \in \rho(1_{M}).\]
\end{defi}

\begin{thm}\label{thm:dual}\hfill
  \begin{enumerate}
  \item\label{thm:dual:monoids} The category of derivation algebras is
    dually equivalent to the category of Priestley monoids
    \[
      \der \;\cong\; \comon \;\simeqop\; \cat{PriestMon}.
    \]

    \item The duality from \autoref{thm:dual:monoids} extends to Priestley relational morphisms:
          \[
      \relder \;\cong\; \relcomon \;\simeqop\; \cat{RelPriestMon}.
    \]
  \end{enumerate}
\end{thm}
\begin{proof}\hfill
\begin{enumerate}
  \item\label{item:1}  We assemble all the steps of the (dual) equivalences:
        The category  of profinite ordered monoids is the Pro-completion of the category of finite ordered monoids~\cite[Prop.~2.10]{adamek-chen-milius-urbat-21}.
        Since the category of finite ordered monoids is dual to the category of finite comonoids (\autoref{thm:fin-dual} restricted to finite structures), the Pro-completion of the former is dual to the Ind-completion of the latter. By \autoref{prop:loc-fin-comon-ind-compl},
        the latter is equivalent to the category of locally finite comonoids -- but this category is, by \iref{thm:loc-fin-dual}{der}, equivalent to the category of derivation algebras:
        \[\cat{PriestMon} \;\simeq\; \pro(\ordmon_{\mathrm{f}}) \;\simeqop\; \ind(\comon_{\mathrm{f}}) \;\simeq\; \comon_{\mathrm{lf}} \;\simeq\; \der. \]

  \item A Priestley relational morphism from \(M\) to \(N\) is precisely a total Priestley relation \(\rho \colon M \rightarrow \dviet N\) such that the following diagrams commute laxly as indicated.
\[
    \begin{tikzcd}
      M \times M
      \ar{rr}{\cdot_{M}}
      \dar{\rho \times \rho}
      &
      &
      M
      \dar{\rho}
      &&
      1
      \rar{1_{M}}
      \dar{1_{N}}
      &
      M
      \dar{\rho}
      \\
      \dviet N \times \dviet N
      \rar{\hat{\delta}}
    \ar[phantom]{urr}{\rotatebox{45}{\(\le\)}}
     &
     \dviet(N \times N)
     \rar{\dviet(\cdot_N)}
     &
      \dviet N
      &&
    \urar[phantom]{\rotatebox{45}{\(\le\)}}
      N
      \rar{\eta}
      &
      \dviet N
    \end{tikzcd}
  \]
\end{enumerate}
Recall that \(\dviet \cong \hatF_{\lor}\hat\Uj\) for \(\Uj \colon \dl \rightarrow \jsl\), so
under extended duality \(\rho\) dualizes precisely to a corelational morphism of comonoids:
\[
  \begin{tikzcd}[baseline = (B.base)]
    \Uj\hat{M} \jsor \Uj\hat{M}
    &
    \Uj \hat{M}
    \lar[swap]{\Uj(\hat{\cdot}_{M})}
    &&
    2
    &
    \Uj\hat{M}
    \lar[swap]{\Uj(\hat{1}_{M})}
    \\
    \Uj\hat{N} \jsor \Uj\hat{N}
    \uar[swap]{\hat{\rho}^{-} \jsor \hat{\rho}^{-}}
    \ar[phantom]{ur}{\rotatebox{45}{\(\le\)}}
    &
    \Uj\hat{N}
    \uar[swap]{\hat{\rho}^{-}}
    \lar[swap]{\Uj(\hat{\cdot}_{M})}
    &
    &
    \phantom{X}
    \urar[phantom, near end]{\rotatebox{45}{\(\le\)}}
    &
    |[alias = B]|
    U_\lor \hat{N}
    \uar[swap]{\hat{\rho}^-}
    \ular{U_\lor(\hat{1}_N)}
  \end{tikzcd}
\]
Together with \iref{thm:loc-fin-dual}{rel} this extends the duality established in \autoref{item:1}. \qedhere
\end{proof}

\begin{rem}
    \autoref{thm:dual} restricts to a duality between the category of profinite monoids and Stone relational morphisms and the category of Boolean derivation algebras and corelational residuation morphisms.
\end{rem}

\section{Duality for the Category of Small Categories}
\label{sec:duality-cat}

As a final application of the abstract extended duality framework, we derive a concrete description of \(\ccat^{\mathrm{op}}\), the dual of the category of small categories and functors.
We instantiate the parameters of \autoref{asm} to the Diagram~\eqref{eq:asm-cat}.

  It is well known that small categories can be described in an object-free way as partial monoids~\cite{maclane-98,fahrenberg-23,schweizer-67} with an additional \emph{locality} condition.
  For our purposes it will be convenient  to describe partial monoids more generally as monoids in the monoidal category \rel.

  \begin{nota}\label{not:rel}
    Given a ternary relation \(r \colon X \times Y \rightarrow \pow Z\), we obtain the relation
    \begin{equation}\label{eq:lift-r}
      \bigcup { \cdot } \pow(r) \cdot \hat{\delta} \colon \pow X \times \pow Y \rightarrow \pow (X \times Y ) \rightarrow \pow \pow Z \rightarrow \pow Z, \qquad (A, B) \mapsto \bigcup_{x \in A, y \in B} r(x, y).
    \end{equation}
    Abusing notation, we denote the map \eqref{eq:lift-r} also by \(r \colon \pow X \times \pow Y \rightarrow \pow Z\).
    We write \(r_{@}(x, y)\) if \(r(x, y) \ne \emptyset\), and similarly \(x \mathrel{r_{@}} y\) if \(r\) is used as an infix operator.
    Furthermore, we identify singleton sets \(\{x\} \subseteq X\) with their unique inhabitants \(x \in X\).
  \end{nota}

  \begin{defi}\label{def:rel-mon}\hfill
    \begin{enumerate}
    \item 
    A \emph{relational monoid}  consists of a carrier set \(M\), a subset \(E \incl M\) of \emph{identities} and a \emph{multiplication relation} \(\circ  \colon M \times M \rightarrow \pow M\) such that the following diagrams commute in \rel:
\[
    \begin{tikzcd}
      M \bar{\times} M \bar{\times} M
      \rar{\id \bar{\times} \circ}
      \dar{\circ \bar{\times} \id}
      &
      M \bar{\times} M
      \dar{\circ}
      &
      \\
      M \bar{\times}  M
      \rar{\circ}
      &
      M
    \end{tikzcd}
    \begin{tikzcd}
      1 \bar{\times} M
      \rar[phantom]{\cong}
      \dar{E \bar{\times} \id}
      &
      M
      \dar{\id}
      \rar[phantom]{\cong}
      &
      M \bar{\times} 1
      \dar{\id \bar{\times} E}
      \\
      M \bar{\times} M
      \rar{\circ}
      &
      M
      &
      M \bar{\times} M
      \lar[swap]{\circ}
    \end{tikzcd}
  \]
  Using \autoref{not:rel}, the diagrams above read
  \[
    \forall x, y, z \in M \colon (x \circ y) \circ z = x \circ (y \circ z), \qquad E \circ x = x = x \circ E.
  \]

  \item
  A relational monoid is \emph{partial} if \(\circ\) is single-valued, and it is
  \emph{local}~\cite{fahrenberg-23} if it satisfies
  \begin{equation*}
    x \circ_{@} y \text{ and } v \in y \circ z \implies x \circ_{@} v.
  \end{equation*}

  \item
    A \emph{functorial} morphism \(h \colon (M, \circ, E) \rightarrow (M', \circ, E')\)  of partial monoids is a pure morphism of binary \(J_{\pow}\)-operators (that is, \(\pow(h)(x \circ y) = h(x) \circ h(y)\)) satisfying \(h(E) \subseteq E'\). The category of relational monoids and functorial morphisms is denoted by \(\relmon\).
  \end{enumerate}
  \end{defi}

  Every monoid is obviously also a relational monoid.
  An example of a partial monoid that is not local is given by  \((\pow(X), \circ, \emptyset)\), \(|X| \ge 1\), with \(A \circ B = A + B\) if \(A, B\) are disjoint and \(A \circ B = \emptyset\) otherwise.

  \begin{rem}\hfill
    \begin{enumerate}
      \item The full subcategory  of \relmon consisting of  local partial monoids is equivalent to the category of small categories~\cite{maclane-98,fahrenberg-23,schweizer-67}:
            \[\relmon_{\mathrm{fun}, \mathrm{loc}} \simeq \ccat.\]
            From right to left, a small category \(\cat{C}\) is mapped to the relational monoid \(\mor{\cat{C}}\) with identities \(\{\id_{C} \mid C \in \cat{C}\}\).
            It is easy to see that \(\mor{\cat{C}}\) is a local partial monoid.

            For left to right, one first has to show that in a partial monoid \(M\) every element \(f \in M\) has unique left and right units \(e^{l}_{f}, e^{r}_{f} \in E\).
            The equivalence then sends a monoid \(M\) to the category whose objects are given by the units \(E\) and with morphisms \[\hom(e, e') = \{f \in M \mid e = e^{l}_{f},\, e' = e^{r}_{f}\}.\]
      \item The multiplication of the relational monoid corresponding to a small category is single-valued, that is, it factorizes as \[\mor(\cat{C}) \times \mor(\cat{C}) \rightarrow \mor(\cat{C}) + \{\bot\} \incl \pow(\mor(\cat{C})),\]
            its set of units can be any element  of \(\pow(\mor(\cat{C}))\).
            This is the reason why we work with the full powerset in our assumptions~\eqref{eq:asm-cat} in lieu of the maybe submonad.
    \end{enumerate}
  \end{rem}

  From \autoref{sec:adjunct-distr-latt} we already know relational monoids correspond to certain associative residuation CABAs,
  so it remains to single out the images of local partial monoids.
\pagebreak
  \begin{defi}\hfill
    \begin{enumerate}
      \item A residuation CABA is \emph{functional} if for all atoms \(a \in \at(R)\) the \cslm-morphisms \(a \lres (-) \colon R \rightarrow R\) also preserve non-empty joins.
            It is \emph{local} if it satisfies \(x^{?} \lres x^{?} = \top\), where 
            \[(-)^{?} \in \cslj(R, R), \qquad x \mapsto \neg (x \lres \bot).\]
            A residuation CABA is \emph{categorical} if it has a (not necessarily atomic!) unit \(e \in R\) and is associative, functional and local.
      \item A (non-unital) \emph{morphism} from a residuation CABA \(R\) to \(R'\) is a complete residuation algebra morphism \(h \colon R \rightarrow R'\), that is, \(h\) preserves all joins and meets and satisfies
            \(x' \lres f(z) = f(f^{*}(x') \lres z) \).
            If \(R, R'\) are unital residuation CABAs, then a morphism is \emph{lax unital} if it satisfies  \(e' \le f(e)\).
            
            We denote the category of  residuation CABAs by \(\cat{ResCABA}\) and its subcategory  of categorical residuation CABAs and lax unital morphisms by \(\cat{CatResCABA}\).
    \end{enumerate}
  \end{defi}

  \begin{thm}
    The category of small categories is dually equivalent to the category of categorical residuation CABAs:
    \[\ccat \simeqop \cat{CatResCABA}.\]
  \end{thm}
  \begin{proof}
  Restricting the results from \autoref{sec:adjunct-distr-latt} to the order-discrete setting, we get a duality
  \begin{equation}
    \label{eq:dual-me}
    \rescaba \cong \coalg(\Vm) = \op_{\Vm}^{1, 2}(\caba) \simeqop \op_{J_{\pow}}^{2, 1}(\set).
  \end{equation}
    From right to left, the dual of a \(J_{\pow}\)-algebra \(\circ \colon M \times M \rightarrow \pow M\) is given by the residuation CABA with carrier \(\hatM \cong \pow(M)\), whose left residual is given by
    \[\lres \colon  \pow(M)^{\partial} \msorc  \pow(M) \rightarrow  \pow(M), \qquad   A \msor C \mapsto \{ b \mid A \circ b \subseteq C  \}. \]
    It remains to show that this duality restricts to a duality between the category of local partial monoids and functorial homomorphisms, and the category of categorical residuation CABAs and lax unital morphisms.
    \begin{enumerate}
    \item  We start by showing that the duality restricts on objects:
    a binary \(J_{\pow}\)-algebra \(\circ \colon M \bar{\times} M \rightarrow \pow M\) is a local, partial monoid iff its dual residuation CABA \(\hatM\) is unital, associative, functional and local.

    \smallskip\noindent
            (1a) Partiality of multiplication corresponds to functionality of \(\hatM\): By \autoref{lem:fin-op-equiv} restricted to non-empty joins, the residuation CABA \(\hatM\) is functional iff its \(\Vm\)-coalgebra structure preserves non-empty joins.
            An argument analogous to \autoref{cor:par-fun} then gives that a \(\Vm\)-coalgebra preserves non-empty joins iff its dual order-relation is a partial map.
            Combining these two equivalences we get that the residuation CABA \(\hatM\) is functional if its dual relation \(\circ \colon M \times M \rightarrow \pow(M)\) is a partial map.

    \smallskip\noindent
    (1b) By \autoref{lem:op-props} the residuation CABA \(\hatM\) has a unit \(e \in \hatM\) iff its comultiplication has a counit \(\epsilon \vdash e\) iff its dual algebra \(M\) has a unit \(E = e\). By the same lemma, \(\hatM\) is associative iff its comultiplication is coassociative iff \(M\) is associative.

    \smallskip\noindent
    (1c) We show that the residuation CABA \(\hatM\) satisfies \(\forall A \in \hatM \colon A^{?} \lres A^{?} = \top\) iff its dual algebra \(M\) is local:
    We prepare by making some observations about the locality of the residuation CABA \(\hatM\).
    Note that in \(\hatM\) we have \(\bot = \emptyset\) and \(\top = M\).
    For \(A \in \hatM\) we have
    \begin{equation}
      \label{eq:?}
      A^{?} = \neg (A \lres \bot) = \{n \mid \neg (A \circ n \subseteq \emptyset)\} = \{n \mid A \circ_{@} n\} = \{n \mid \exists m \in A \colon m \circ_{@} n\}.
    \end{equation}
    In particular, on atoms \(m \in \hatM\) membership in \(m^{?}\) simplifies to the condition
    \begin{equation}
      \label{eq:mdn}
    n \in m^{?} \iff m \circ_{@} n.
    \end{equation}
    We rewrite the locality condition using the adjunction \(B \circ (-) \dashv B \lres (-)\):
    \begin{equation}
      \label{eq:local}
      A^{?} \lres A^{?} = M  \iff M \subseteq A^{?} \lres A^{?} \iff A^{?} \circ M \subseteq A^{?},
    \end{equation}
    where \(\circ \colon \hatM \jsorc \hatM \rightarrow \hatM\) is the \(\Vj\)-algebra structure on \(\hatM\) corresponding to the extension $\pow M\times \pow M\to \pow M$ of the partial monoid multiplication.

    We start with the direction that if \(\hatM\) satisfies \eqref{eq:local}, then \(M\) is local:
    For this, let \(x, y \in M\) such that \(x \circ_{@} y\) and \(v
    \in y \circ z\). We have to show that \(x \circ_{@} v\).
    We first apply \eqref{eq:mdn} to $x \circ_{@} y$ to obtain $y \in
    x^?$. Monotonicity of~\(\circ\) in both
    arguments then yields the first inclusion below
    and~\eqref{eq:local} the second one: 
    \[v \in y \circ z \subseteq x^{?} \circ M \subseteq x^{?}. \]
    But by \eqref{eq:mdn} this means that \(x \circ_{@} v\), as
    required.
    
    For the other direction, we prove that if \(M\) is local, then \(\hatM\) is local.
    By \eqref{eq:local} it suffices to prove \(\forall A \in \hatM \colon A^{?} \circ M \subseteq A^{?}\).
    Let \(v \in A^{?} \circ M\). Then there exist \(y \in A^{?}\), \(z \in M\) such that \(v \in y \circ z\).
    Since \(y \in A^{?}\), there exists, by~\eqref{eq:?}, some \(x \in A\) with \(x \circ_{@} y\).
    Locality of \(M\) now implies that \(x \circ_{@} v\), so \(v \in
    A^{?}\) as required using~\eqref{eq:?} again.

  \item Finally, we show that the duality in~\eqref{eq:dual-me} restricts on morphisms as claimed.
            Let \(f \colon M \rightarrow M'\) be a pure morphism of \(J_{\pow}\)-algebras with dual morphism \(h = f^{-1} \colon \widehat{M'} \rightarrow \hatM\) of residuation CABAs.
            We have to show that \(f\) is functorial iff \(h\) is lax unital:
    By (1b) the unit \(E\) of a partial monoid $M$ dualizes to the residual unit $E \in \hatM$.
    Now \(f\) is a functorial morphism of partial monoids iff \(f[E]
    \subseteq E'\) iff \(E \subseteq f^{-1}[E'] = h(E')\) iff \(h\) is
    lax unital.\qedhere
  \end{enumerate}
\end{proof}

\section{Conclusion and Future Work}
\label{sec:concl-future-work}

We have presented an abstract approach to extending Stone-type dualities based on adjunctions
between monoidal categories and instantiated it to recover classical extended Stone and
Priestley duality, along with a generalization of it to relational morphisms.  Guided by these
foundations, we have investigated residuation and derivation algebras, leading to a new duality for
Priestley monoids, and we extended this duality to include relational morphisms. In addition, we
have derived a new dual characterization of the category of small categories.

Relational morphisms are an important tool in algebraic language theory, notably for characterizing language operations algebraically. For instance, Straubing~\cite{str81} first showed that relational morphisms are tightly connected to the concatenation product and the star operation on regular languages; see also the surveys by Pin~\cite{pin-88,pin11}.
In future work, we intend to apply our duality-theoretic insights on relational morphisms to
illuminate, and possibly recover, these results from a categorical perspective, much in the
spirit of the duality-theoretic view of Eilenberg's Variety Theorem by Gehrke
et~al.~\cite{gehrke-gregorieff-pin-08} and the categorical works it has inspired~(see e.g.~\cite{sal-17,urb-ada-che-mil-17-proc,boj15,blu21}).

Another goal is to apply our duality framework beyond
classical Stone and Priestley dualities.  Specifically, we aim to
develop an extended duality theory for the recently developed
\emph{nominal} Stone duality~\cite{birkmann-milius-urbat-23}, which
would enable a generalization of our present results on residuation algebras
to the nominal setting with applications to data languages.

A conceptually rather different dual characterization of the category of profinite monoids and
continuous monoid morphisms in terms of \emph{semi-Galois categories} has
been provided by Uramoto~\cite{uramoto-16}. Extending this result to
relational morphisms, similar to our \autoref{thm:dual}, is another
interesting point for future work.

Potential applications of our abstract approach to extended Stone duality are not limited to
algebraic language theory. In~\autoref{sec:towards-point-free} we have used the compositionality of extended (Stone) duality to recover results from modal correspondence theory by purely categorical methods.
We hope that the monoidal approach might bring a new impulse into the historic endeavor of correspondence theory.
In future work we will investigate the expressiveness of this idea, that is, study which
relational properties can be captured by suitable inequations of operators. Applying this
approach to more complex modal axioms that, for example, involve negation or implication or combine multiple modalities, can be expected to be non-trivial.

Our applications so far were based on Stone and Priestley duality, but the general framework of abstract extended duality applies far beyond this setting. For instance, Furber and Jacobs~\cite{furber-jacobs-15} showed how to extend the duality between \(C^{*}\)-algebras and compact Hausdorff spaces (`Gelfand duality') to a `probabilistic Gelfand duality', which emerges by employing a weaker notion  of morphism between \(C^{*}\)-algebras and by replacing the category of compact Hausdorff spaces by the  Kleisli category of the Radon monad.
This result seems to fit perfectly into our approach of extending dualities.
A thorough instantiation of the results from~\cite{furber-jacobs-15} to our framework could not only place these results in a larger categorical context but also uncover new results in probabilistic duality theory.

\section*{Acknowledgement}
\noindent
The authors would like to thank Mai Gehrke for the helpful discussion on Priestley monoids, and
the anonymous reviewers for their detailed comments and suggestions for improving the presentation.

\section*{Funding}
  \noindent Fabian Lenke and Stefan Milius are supported by Deutsche Forschungsgemeinschaft (DFG, German Research Foundation) under project 419850228. Henning Urbat is supported by Deutsche Forschungsgemeinschaft (DFG, German Research Foundation) under project 470467389.

\bibliographystyle{alphaurl}
\bibliography{bibliography}

\end{document}